\documentclass[11pt]{article}
\usepackage[margin=1in]{geometry}
\usepackage{libertine}

\usepackage{datetime}
\usepackage{multirow}
\usepackage{booktabs}
\usepackage{enumitem}
\usepackage{wrapfig}
\usepackage{subcaption}
\usepackage[T1]{fontenc}
\usepackage[utf8]{inputenc}
\usepackage{pgfplots}
\usepackage{amsmath, amssymb}
\usepackage{amsthm}
\usepackage{color}
\usepackage{cases}
\usepackage{xparse}
\usepackage{xargs}
\usepackage{appendix}
\usepackage[boxed,commentsnumbered,noend]{algorithm2e}
\usepackage{algpseudocode}
\usepackage{verbatim, xspace}
\usepackage{tikz}
\usepackage{mathtools}

\usepackage{hyperref}
\usepackage{pgfplots}

\usepackage{xcolor}

\newcommand{\citet}[1]{\cite{#1}}
\SetAlgorithmName{Isolation scheme}{isolation scheme}{List of isolaiton schemes}
\usepackage[colorinlistoftodos,prependcaption,textsize=tiny,disable]{todonotes}
\newcommandx{\unsure}[2][1=]{\todo[linecolor=green,backgroundcolor=green!25,bordercolor=green,#1]{\normalsize #2}}
\newcommandx{\improvement}[2][1=]{\todo[inline,linecolor=blue,backgroundcolor=blue!05,bordercolor=blue,#1]{\normalsize #2}}
\newcommandx{\info}[2][1=]{\todo[linecolor=yellow,backgroundcolor=yellow!25,bordercolor=yellow,#1]{#2}}
\newcommandx{\floatmodel}[2][1=]{\todo[inline,linecolor=red,backgroundcolor=yellow!25,bordercolor=yellow,#1]{#2}}
\newcommandx{\thiswillnotshow}[2][1=]{\todo[disable,#1]{#2}}
\newcommandx{\celine}[2][1=]{\todo[inline,linecolor=green,backgroundcolor=green!25,bordercolor=green,caption={\normalsize \textbf{Celine}},#1]{\normalsize #2}}
\newcommandx{\karol}[2][1=]{\todo[inline,linecolor=blue,backgroundcolor=blue!25,bordercolor=blue,caption={\normalsize \textbf{Karol}},#1]{\normalsize #2}}
\newcommandx{\jesper}[2][1=]{\todo[inline,linecolor=red,backgroundcolor=red!25,bordercolor=red,caption={\normalsize \textbf{Jesper}},#1]{\normalsize #2}}
\newcommandx{\michal}[2][1=]{\todo[inline,linecolor=gray,backgroundcolor=red!25,bordercolor=red,caption={\normalsize \textbf{Micha\l{}}},#1]{\normalsize #2}}

\newtheorem{theorem}{Theorem}
\newtheorem{definition}[theorem]{Definition}
\newtheorem{lemma}[theorem]{Lemma}
\newtheorem{corollary}[theorem]{Corollary}
\newtheorem{claim}[theorem]{Claim}

\newtheorem{conjecture}[theorem]{Conjecture}

\newtheorem{remark}[theorem]{Remark}

\newtheorem*{maingoal*}{Main Question}

\numberwithin{theorem}{section}
\newcommand{\IH}[1]{\smallskip  \noindent \fbox{ \begin{minipage}{\textwidth} \paragraph{Induction hypothesis.} {#1} \end{minipage}}}

\newcommand{\Prb}{\mathbb{P}}
\newcommand{\bnd}{\partial}

\newcommand{\floor}[1]{\left\lfloor #1 \right\rfloor}
\newcommand{\ceil}[1]{\left\lceil #1 \right\rceil}

\newcommand{\eps}{\varepsilon}

\newcommand{\Oh}{\mathcal{O}}

\newcommand{\nat}{\mathbb{N}}
\newcommand{\N}{\mathbb{N}}

\newcommand{\Ff}{\mathcal{F}}
\newcommand{\Rr}{\mathcal{R}}
\newcommand{\Ss}{\mathcal{S}}
\newcommand{\Cc}{\mathcal{C}}
\newcommand{\Mm}{\mathcal{M}}

\newcommand{\Gg}{\mathcal{G}}
\newcommand{\Pp}{\mathcal{P}}
\newcommand{\Hh}{\mathcal{H}}

\newcommand{\poly}{\mathrm{poly}}
\newcommand{\polylog}{\mathrm{polylog}}

\newcommand{\ol}{\overline}
\newcommand{\CMSOtwo}{\mathsf{CMSO}_2}
\newcommand{\angles}[1]{\langle {#1} \rangle}
\newcommand{\MIS}{\mathsf{MIS}}

\newcommand{\subtree}{\mathsf{subtree}}
\newcommand{\tail}{\mathsf{tail}}

\newcommand{\conf}{\mathsf{conf}}

\newcommand{\lvl}{\mathrm{lvl}}

\newcommand{\Min}{\mathsf{Min}}
\newcommand{\id}{\mathsf{id}}

\renewcommand{\leq}{\leqslant}
\renewcommand{\geq}{\geqslant}
\renewcommand{\le}{\leqslant}
\renewcommand{\ge}{\geqslant}

\title{Isolation schemes for problems on decomposable graphs}
\date{}

\author{
    Jesper Nederlof\footnote{Utrecht University, The
    Netherlands, \textsf{j.nederlof@uu.nl}. Supported by
    the project CRACKNP that has received funding from the European
    Research Council (ERC) under the European Union's Horizon 2020 research and
    innovation programme (grant agreement No 853234).}
    \and
    Micha\l{} Pilipczuk\footnote{Institute of Informatics, University of
    Warsaw, Poland, \textsf{michal.pilipczuk@mimuw.edu.pl}. This works is a part of
    the project TOTAL that has received funding from the European
    Research Council (ERC) under the European Union's Horizon 2020 research and
    innovation programme (grant agreement No 677651).}
    \and
    Céline M. F. Swennenhuis\footnote{Eindhoven University of Technology, The
    Netherlands, \textsf{c.m.f.swennenhuis@tue.nl}. Supported by the Netherlands
    Organization for Scientific Research under project no. 613.009.031b.}
    \and
    Karol W\k{e}grzycki\footnote{Saarland University and Max Planck Institute for Informatics,
        Saarbr\"ucken, Germany, \textsf{wegrzycki@cs.uni-saarland.de}. 
    This work is part of the project TIPEA that has
    received funding from the European Research Council (ERC) under the European Union's Horizon
    2020 research and innovation programme (grant agreement No 850979).}
}

\begin{document}

\maketitle

\thispagestyle{empty}
\begin{abstract}
The Isolation Lemma of Mulmuley, Vazirani and Vazirani~[Combinatorica'87] provides a self-reduction scheme that allows one to assume that a given instance of a problem has a unique solution, provided a solution
exists at all. Since its introduction, much effort has been dedicated towards
derandomization of the Isolation Lemma for specific classes of problems. So far, the focus was mainly on problems solvable in polynomial time.


In this paper, we study a setting that is more typical for $\mathsf{NP}$-complete problems, and obtain partial derandomizations in the form of significantly decreasing the number of required random bits.
In particular, motivated by the advances in parameterized algorithms, we focus on problems on
decomposable graphs. For example, for the problem of detecting a Hamiltonian cycle, we build upon the rank-based approach from~[Bodlaender et al., Inf. Comput.'15] and design isolation schemes that use
\begin{itemize}[nosep]
 \item $\Oh(t\log n + \log^2{n})$ random bits on graphs of treewidth at most $t$;
 \item $\Oh(\sqrt{n})$ random bits on planar or $H$-minor free graphs; and
 \item $\Oh(n)$-random bits on general graphs.
\end{itemize}
In all these schemes, the weights are bounded exponentially in the number of random bits used.
As a corollary, for every fixed $H$ we obtain an algorithm for detecting a Hamiltonian cycle in an $H$-minor-free graph that runs in deterministic time $2^{\Oh(\sqrt{n})}$ and uses polynomial space; this is the first algorithm to achieve such complexity guarantees. For problems of more local nature, such as finding an independent set of maximum size, we obtain isolation schemes on graphs of treedepth at most $d$ that use $\Oh(d)$ random bits and assign polynomially-bounded weights.

We also complement our findings with several unconditional and conditional lower bounds, which show that many of the results cannot be significantly improved.
\end{abstract}

\begin{picture}(0,0)
\put(462,-170)
{\hbox{\includegraphics[width=40px]{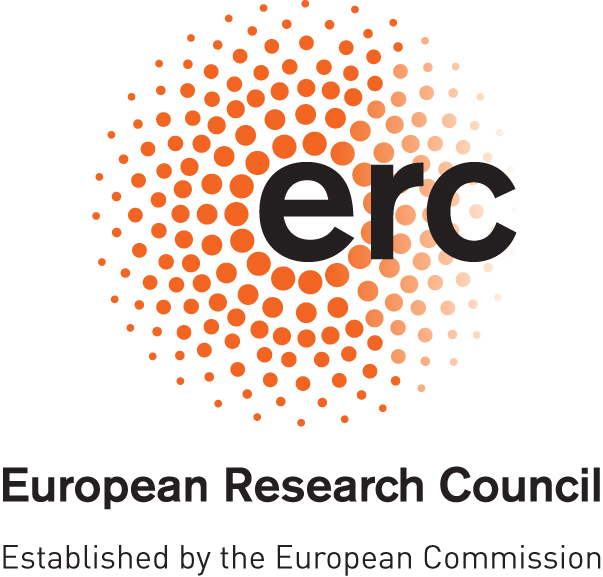}}}
\put(452,-230)
{\hbox{\includegraphics[width=60px]{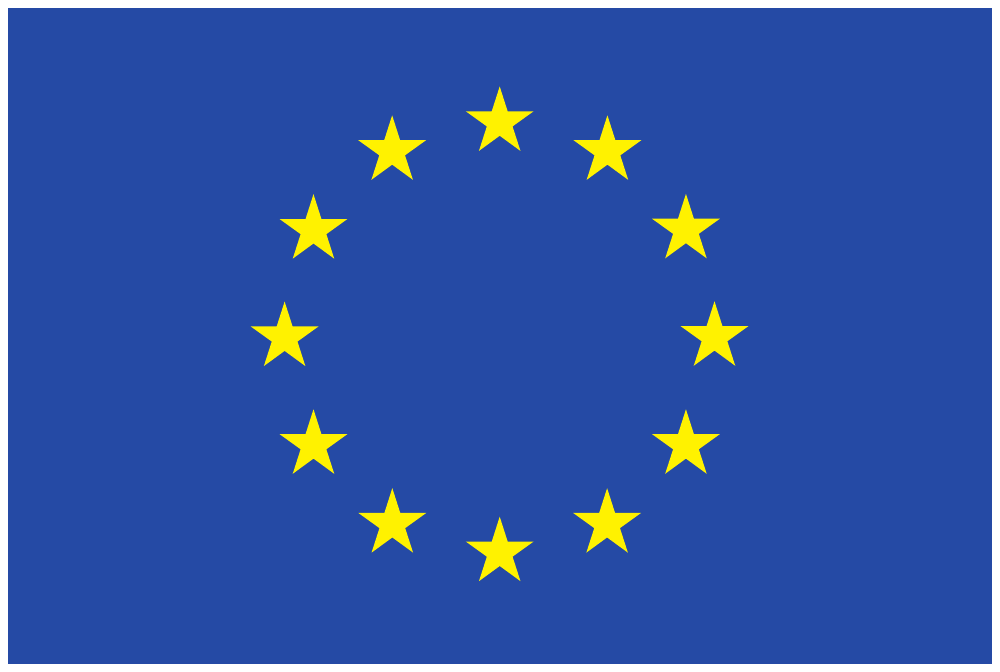}}}
\end{picture}

\clearpage
\setcounter{page}{1}

\section{Introduction}


{\em{Isolation}} is a procedure that allows to single out a unique solution to a given problem within a possibly larger solution space, thus effectively reducing the original problem to a variant where one may assume that if a solution exists, then there is a unique one. The classic Isolation Lemma of Mulmuley, Vazirani and Vazirani~\cite{mulmuley} can be used to achieve this at the cost of allowing randomization. In complexity theory, isolation is used to show that hard problems
are not easier to solve on instances with unique solutions~\cite{usat}. This idea has found numerous applications ranging from structural results in complexity theory (e.g. $\mathsf{NL/poly} \subseteq \oplus
\mathsf{L/poly}$~\cite{wigderson94} or
$\mathsf{NL/poly} = \mathsf{UL/poly}$~\cite{reinhardt00}) to the design of parallel algorithms~\cite{mulmuley,matchings-6,matchings-7,tarnawski17}. 



Since obtaining a general derandomization of the Isolation Lemma is impossible by counting arguments~\cite{iso-lower-bound,isolation-lemma2,iso-lower-bound2}, it is natural to ask whether the isolation step can be derandomized for specific problems with explicit representation.
In this context, there has recently been an exciting progress in isolation for perfect matchings~\cite{matchings-1,matchings-2,matchings-3,matchings-4,matchings-5,matchings-6}, which culminated in an isolation scheme that uses $\Oh(\log^3 n)$ random bits, implying a quasi-$\mathsf{NC}$ algorithm for detecting a perfect matching~\cite{tarnawski17}.

In contrast to this, derandomization of isolation procedures for $\mathsf{NP}$-complete problems is relatively less studied, and not because of a lack of motivation: Many contemporary fixed-parameter algorithms rely on the Isolation Lemma~\cite{LiN20,wg2020,bjorklund10,stacs2020,JansenN19,focs2011,Williams09}. Usually, the isolation procedure is the only subroutine requiring randomness. Many of the algorithms mentioned above apply the Isolation Lemma in combination with a decomposition-based method such as Divide\&Conquer or dynamic programming.  This motivates us to study the following:
\begin{maingoal*}
	How much randomness is required for isolating problems with decomposable structure?
\end{maingoal*}
%
%
More concretely, we focus on graph problems where the underlying graph is \emph{decomposable}, in the sense that it can be decomposed using small separators. Examples of such graphs are planar graphs or graphs of bounded treewidth.
It is well-known that for many $\mathsf{NP}$-complete problems, the nice structure of such graphs can be leveraged to solve these problems faster than in general graphs.
We show that a similar phenomenon occurs when one considers the amount of
randomness needed to isolate a single solution.

\paragraph{The model for isolation schemes.}
Suppose $U$ is a finite set and $\omega \colon U \rightarrow \nat$ is a weight function. For $X \subseteq U$ we write $\omega(X) \coloneqq \sum_{e \in X}
\omega(e)$. For a set family $\mathcal{F} \subseteq 2^U$ we say that $\omega$ \emph{isolates $\mathcal{F}$} if there is exactly one set $S\in \mathcal{F}$ such that $\omega(S)$ is the minimum possible among the weights of the sets in $\mathcal{F}$. 
The classic Isolation Lemma of Mulmuley et al.~\cite{mulmuley} states that a weight function $\omega\colon U \rightarrow \{1,\ldots,2|U|\}$ chosen uniformly at random isolates any family $\mathcal{F}\subseteq 2^U$ with probability at least $\tfrac{1}{2}$. Note that sampling such $\omega$ requires $\Oh(|U|\log|U|)$ random bits.



Most of our isolation schemes work in a very restricted model inspired by the discussion above, which we explain now. Intuitively, the scheme is not aware of the graph or its decomposition, but is only aware of the vertex count of the graph and the relevant width parameter, such as the treewidth or treedepth.

Formally, a {\em{vertex selection problem}} is a function $\Pp$ that maps every graph $G$ to a family $\Pp(G)\subseteq 2^{V(G)}$ consisting of subsets of the vertex set of $G$. Edge selection problems are defined analogously: $\Pp(G)$ consists of subsets of $E(G)$. For example, we could define a vertex selection problem $\MIS(\cdot)$ that maps every graph $G$ to the family $\MIS(G)$ comprising all maximum-size independent sets in $G$, or an edge selection problem $\mathsf{HC}(\cdot)$ that maps every graph $G$ to the family $\mathsf{HC}(G)$ comprising all (edge sets of) Hamiltonian cycles in~$G$. Further, let $\Cc$ be a class of graphs, that is, a set of graphs that is invariant under isomorphism. For instance, $\Cc$ could be the class of planar graphs, or the class of graphs of treewidth at most $k$, for any fixed $k$. Then  our definition of an isolation scheme reads as follows (here, we write $[n]\coloneqq \{1,\ldots,n\}$):

\begin{definition}
    \label{def:isolation-scheme}
    For a graph class $\Cc$, we say that a vertex selection problem $\Pp$ admits an
    {\em{isolation scheme}} on $\Cc$ if for every $n \in \nat$ there exist weight functions
    $\omega_1,\ldots,\omega_\ell\colon [n] \rightarrow \nat$ such that for every $G\in \Cc$ with vertex set $[n]$, $\omega_i$ isolates $\Pp(G)$ for at least half of the indices $i \in [\ell]$.
\end{definition}



Isolation schemes for edge selection problems are defined analogously: the weight functions $\omega_1,\ldots,\omega_\ell$ have domain $[m]$ and should assign weights to all the edges in $m$-edge graphs in $\Cc$, where the edges are assumed to be enumerated with numbers in $[m]$.


The two main parameters of interest for isolation schemes will be the number of \emph{random bits}, which is defined as $\log \ell$, and the \emph{maximum weight}, defined as the maximum value that any of the functions $\omega_i$ may take.
Although Definition~\ref{def:isolation-scheme} only assumes the
\emph{existence} of suitable weight functions, all the isolation schemes
proposed in this paper are extremely simple and can be use as an effective derandomization tool.


\subsection{Our contribution}


In the following discussion we restrict attention to Hamiltonian cycles and maximum-size independent sets for concreteness, that is, to the edge- and vertex-selection problems $\mathsf{HC}(\cdot)$ and $\MIS(\cdot)$ described above. However, our techniques have a wider applicability, which we comment on throughout the presentation. On a very high level, the natural idea that permeates all our arguments is to reduce the randomness using Divide\&Conquer along small separators: If a separator $X$ splits the given graph $G$ in a balanced way, then the same random bits can be reused in each part of $G-X$.


\paragraph{Isolation schemes for Hamiltonian cycles.}
We first consider the problem of detecting a Hamiltonian cycle, since it represents an important class of connectivity problems such as {\sc{Steiner Tree}} or {\sc{$k$-Path}}. For these problems, the Isolation Lemma has been particularly useful in the design of parameterized algorithms~\cite{LiN20,wg2020,bjorklund10,stacs2020,JansenN19,focs2011,Williams09}. Our first results concerns general graphs.

\begin{theorem}
    \label{thm:hc-gen}
	There is an isolation scheme for Hamiltonian cycles in undirected graphs that uses $\Oh(n)$ random bits and assigns weights upper bounded by $2^{\Oh(n)}$. 
\end{theorem}

Observe that in an $n$-vertex graph there can be as many as $n!$ different
Hamiltonian cycles. Hence, the application of the general-usage isolation scheme
of Chari et al.~\cite{isolation-lemma2} would give an isolation scheme for
Hamiltonian cycles in general graphs that uses $\Oh(\log (n!))=\Oh(n\log n)$ random bits. Note that as proved in~\cite{isolation-lemma2}, isolating a family $\Ff$ over a universe of size $n$ requires $\Omega(\log |\Ff|+\log n)$ random bits in general, hence the shaving of the $\log n$ factor reported in Theorem~\ref{thm:hc-gen} required a problem-specific insight into the family of Hamiltonian cycles in a graph. This insight is provided by the {\em{rank-based approach}}, a technique introduced in the context of detecting Hamiltonian cycles in graphs of bounded treewidth~\cite{rank-based}. 
The fact that this works is unexpected because
all known methods for derandomizing Hamiltonian cycle require at least
exponential space (see~\cite{rank-based} for overview). 

Let us note that isolation of Hamiltonian cycles was used by Bj{\"{o}}rklund~\cite{bjorklund10} in his $\Oh(1.657^n)$-time algorithm for detecting a Hamiltonian cycle in an undirected graph. This algorithm is randomized due to the usage of the Isolation Lemma, and derandomizing it, even within time complexity $\Oh((2-\eps)^n)$ for any $\eps>0$, is a major open problem. While the constant hidden in the $\Oh(\cdot)$ notation used in Theorem~\ref{thm:hc-gen} is too large to allow exploring the whole space of random bits within time $\Oh((2-\eps)^n)$, in principle we show that the amount of randomness needed is of the same magnitude as would be required for an efficient derandomization of the algorithm of Bj{\"{o}}rklund.


Next, we show that 
in the setting of graphs of bounded treewidth the amount of randomness can be reduced dramatically, to a polylogarithm in $n$.

\begin{theorem}
    \label{thm:hc-tree}
    For every $t\in \nat$, there is an isolation scheme for Hamiltonian cycles in graphs of treewidth at most $t$ that uses $\Oh(t\log n + \log^2(n))$ random bits and assigns weights upper bounded by $2^{\Oh(t\log n + \log^2 n)}$.
\end{theorem}

The proof of Theorem~\ref{thm:hc-tree} fully exploits the idea of using small separators to save on randomness. It also uses the rank-based approach to shave off a $\log t$ factor in the number of random bits.



Finally, we use the separator properties of $H$-minor free graphs to prove the following.

\begin{theorem}
    \label{thm:planar}
    For every fixed $H$, there is  an isolation scheme for Hamiltonian cycles in $H$-minor-free graphs that uses $\Oh(\sqrt{n})$ random bits and assigns weights upper bounded by $2^{\Oh(\sqrt{n})}$.
\end{theorem}

Recently, in~\cite{wg2020} the authors presented a randomized algorithm for
detecting a Hamiltonian cycle in a graph of treedepth at most $d$ that works in time $2^{\Oh(d)} \cdot  (W+n)^{\Oh(1)}$ time
and uses polynomial space; here, $W$ is the maximum weight assigned by isolation
scheme\footnote{They did not consider the weighted case, but the statement is implied by a standard extension, see Section~\ref{sec:algo} for details.}. The only source of randomness in the algorithm of~\cite{wg2020} is the
Isolation Lemma. Since $H$-minor free graphs have treedepth
$\Oh(\sqrt{n})$, we can use the isolation scheme of Theorem~\ref{thm:planar} to derandomize this algorithm, thus obtaining the following result.

\begin{theorem}
    \label{thm:alg}
    For every fixed $H$, there is a deterministic algorithm for detecting a
    Hamiltonian cycle in an $H$-minor-free graph that runs in time
    $2^{\Oh(\sqrt{n})}$ and uses polynomial space.
\end{theorem}

To the best of our knowledge, this is the first application of a
randomness-efficient isolation scheme for a full derandomization of an exponential-time algorithm without a significant loss on complexity guarantees. Further, we are not aware of any previous algorithms that would simultaneously achieve determinism, running time $2^{\Oh(\sqrt{n})}$, and polynomial space complexity, even in the setting of planar graphs\footnote{
Deterministic $2^{\Oh(\sqrt{n})}$-time algorithms were previously known, but all of these use exponential space~\cite{rank-based,matroid-derandomization}.}. Finally, let us note that the algorithm of Theorem~\ref{thm:alg} does not rely on any topological properties of $H$-minor-free graphs: the existence of balanced separators of size $\Oh(\sqrt{n})$ is the only property we use.

\paragraph{MSO-definable problems on graphs of bounded treewidth.}
We observe that the approach used in the proof of Theorem~\ref{thm:hc-tree} relies only on finite-state properties of the {\sc{Hamiltonian Cycle}} problem on graphs of bounded treewidth. The range of problems enjoying such properties is much wider and encompasses all problems definable in $\CMSOtwo$: the Monadic Second-Order logic with modular counting predicates. Consequently, we can lift the proof of Theorem~\ref{thm:hc-tree} to a generic reasoning that yields an analogous result for every $\CMSOtwo$-definable problem. This proves the following (see Section~\ref{sec:mso} for definitions).

\begin{theorem}\label{thm:mso}
	Let $\Pp$ be a $\CMSOtwo$-definable edge selection problem. There exists a computable function $f$ such that for every $k\in \N$, $\Pp$ admits an isolation scheme on graphs of treewidth at most $k$ that uses $R\coloneqq f(k)\cdot \log n+\Oh(\log^2 n)$ random bits and assigns weights upper bounded by $2^R$.
\end{theorem}

\paragraph{Lower bounds.} We show that a significant improvement of the parameters in the isolation schemes presented above is unlikely. First,
a counting argument shows that the $\log n$ factor is necessary.

\begin{theorem}\label{thm:unlb}
	There does not exist an isolation scheme for Hamiltonian cycles on graphs of treewidth at most~$4$ that uses $o(\log n)$ random bits and polynomially bounded weights.
\end{theorem}

Using similar constructions we also provide analogous $\Omega(\log n)$ lower bounds for isolating other families of combinatorial objects related to $\mathsf{NP}$-hard problems, such as maximum independent sets, minimum Steiner trees, and minimum maximal matchings. These lower bounds hold even in graphs of bounded {\em{treedepth}}, which is a more restrictive setting than bounded treewidth.

We also show using existing reductions that a significant improvement over the scheme of Theorem~\ref{thm:hc-gen} would imply a surprising partial derandomization of isolation schemes for {\sc{SAT}}.

\begin{theorem}\label{thm:lbcond}
	Suppose there is an isolation scheme for Hamiltonian cycles in undirected graphs that uses $o(n)$
	random bits and polynomially bounded weights.
	Then there is a randomized polynomial-time reduction from {\sc{SAT}} to {\sc{Unique SAT}} that uses $o(n)$ random bits, where $n$ is the number of variables.
\end{theorem}

Observe that since an $n$-vertex graph has treewidth at most $n-1$, Theorem~\ref{thm:lbcond} also implies that in Theorem~\ref{thm:hc-tree} one cannot expect reducing the number of random bits to $o(t)$. However, we stress  that the lower bounds of Theorems~\ref{thm:unlb} and~\ref{thm:lbcond} are not completely tight with respect to the upper bounds of Theorems~\ref{thm:hc-gen} and~\ref{thm:hc-tree}, because the latter allow superpolynomial weights. It remains open whether the weights used by the schemes of Theorems~\ref{thm:hc-gen},~\ref{thm:hc-tree}, and~\ref{thm:planar} can be reduced to polynomial.

In Section~\ref{sec:lbs} we further discuss consequences of the hypothetical existence of a polynomial-time reduction from {\sc{SAT}} to {\sc{Unique SAT}} that would use $o(n)$ random bits.

\paragraph{Level-aware isolation schemes for independent sets.}
In the light of the $\Omega(\log{n})$ lower bound of Theorem~\ref{thm:unlb}, we consider a relaxation of the model from Definition~\ref{def:isolation-scheme}, where the graph is provided together with an {\em{elimination forest}} (a decomposition notion suited for the graph parameter {\em{treedepth}}), and the weight of a vertex may depend both on the vertex' identifier and its level in the elimination forest. We demonstrate that in this relaxed model, the $\Omega(\log{n})$ lower bound can be circumvented.

\begin{definition}
	We say that vertex selection problem $\Pp$ admits a
	{\em{level-aware isolation scheme}} if for all $n,d\in \nat$ there exist functions 
	$\omega_1,\ldots,\omega_\ell \colon [n] \times [d] \rightarrow \nat$
	such that for every graph $G$ on vertex set $[n]$ and elimination forest $F$ of $G$ of height at most $d$, at least half of the functions $\omega_1,\ldots,\omega_\ell$ isolate $\Pp(G)$. Here, when evaluating $\omega_i$ on a vertex $u\in [n]$, we apply $\omega_i$ to $u$ and the index of the level of $u$ in $F$.
\end{definition}
%

\begin{theorem}\label{thm:mis}
    For every $d \in \nat$, there is a level-aware isolation scheme for maximum-size independent sets in graphs of treedepth at most $d$ that uses $\Oh(d)$ random bits and assigns weights bounded by $\Oh(n^6)$.
\end{theorem}

In the proof of Theorem~\ref{thm:mis} we describe an abstract condition, dubbed the \emph{exchange property}, which is sufficient for the argument to go through. This property is enjoyed also by other families of combinatorial objects defined through constraints of local nature, such as minimum dominating sets or minimum vertex covers. Therefore, we can prove analogous isolation results for those families as well. 

Also, in Section~\ref{sec:maximum-matching-treedepth} we discuss a similar reasoning for edge-selection problems on the example of maximum matchings, achieving a level-aware isolation scheme that uses $\Oh(d\log n)$ random bits and assigns weights bounded by $n^{\Oh(\log n)}$. This provides another natural class of graphs where isolation-based algorithms for finding a maximum matching can be derandomized (see~\cite{matchings-1,matchings-2,matchings-3,matchings-4}).

We summarize our results with Table~\ref{tab:summary}.

\renewcommand{\arraystretch}{1.3}
\begin{table}[ht!]
    \centering
    \caption{Summary of our results based on Theorems~\ref{thm:hc-gen}-\ref{thm:mis}.} 
	\begin{tabular}{l|l|l|l}
	\hline
	\textbf{Problem}     & \textbf{Random Bits}           &\textbf{Max Weight} & \textbf{Graph Class}                           \\ \hline
	\textsc{Hamiltonian Cycle}    & $\Oh(n)$                       &$2^{\Oh(n)}$ & \multirow{2}{*}{General Graphs}            \\
						 & $\Omega(n)$                    &$\poly(n)$&                                            \\\cline{2-4}
						 & $\Oh(\sqrt{n})$                &$2^{\Oh(\sqrt{n})}$& \multirow{2}{*}{$H$-minor free graphs}     \\
						 & $\Omega(\sqrt{n})$             &$\poly(n)$&                                            \\\cline{2-4}
						 & $\Oh(t \log(n) + \log^2(n))$   &$n^{\Oh(t + \log(n))}$& \multirow{2}{*}{Treewidth $t$ graphs} \\
						 & $\Omega(t + \log(n))$          &$\poly(n)$&                                            \\
	\hline
	$\CMSOtwo$           & $f(t) \log(n) + \Oh(\log^2(n))$ &$n^{f(t) +\Oh(\log(n))}$& Treewidth $t$ graphs\\
	\hline
	\textsc{Max Independent Set} & $\Oh(d)$                       &$\poly(n)$& \multirow{2}{*}{Treedepth $d$ graphs} \\
						 & $\Omega(d)$                    &$\poly(n)$&                                           \\\hline
	\end{tabular}
    \label{tab:summary}
\end{table}

\subsection{Organization}
In Section~\ref{sec:informal} we introduce the main techniques behind our isolation schemes in an informal way.
In Section~\ref{sec:prelim} we provide preliminaries.
Section~\ref{sec:hamcycle} is dedicated to the formal proofs of Theorems~\ref{thm:hc-gen},~\ref{thm:hc-tree} and~\ref{thm:planar}.
The derandomized algorithm from Theorem~\ref{thm:alg} is subsequently proved in Section~\ref{sec:algo}, and the general $\CMSOtwo$-result of Theorem~\ref{thm:mso} is formally supported in Section~\ref{sec:mso}.
The lower bounds from Theorem~\ref{thm:unlb} and Theorem~\ref{thm:lbcond} are proved in~\ref{sec:lbs}.
Finally, the level-aware isolation schemes for local vertex (respectively, edge) selection problems are given in Sections~\ref{sec:indepedent_set} (respectively, Section~\ref{sec:maximum-matching-treedepth}), and we finish the paper with possible directions for further research in Section~\ref{sec:conc}.

\section{An informal introduction to our techniques}
\label{sec:informal}

In this section we present an isolation scheme for Hamiltonian
cycle on graphs of bounded \emph{pathwidth} at most $k$ that uses $\Oh(k \log{k} \log{n} +
\log^2{n})$ random bits. 
The arguments in this section are informal in order to convey the underlying intuition, and merely serve as a preliminary overview of the general
framework that we formalize and further develop in the subsequent sections.

Throughout the paper we heavily build upon the approach proposed by Kallampally et
al.~\cite{mfcs16}, who showed an isolation scheme for
shortest paths that uses $\Oh(\log^2(n))$ random bits assigns weights upper bounded by $n^{\Oh(\log{n})}$ (see e.g., \cite{matchings-6} for a more recent application). In fact, our isolation schemes are almost identical to
Kallampally et al.~\cite{mfcs16} except for a different selection of prime
numbers. Our contribution comes with the new insight for $\mathsf{NP}$-complete problems.
To achieve this we use a modern toolset from parameterized algorithms. 


\begin{figure}[t]
    \centering
    \includegraphics[width=0.8\textwidth]{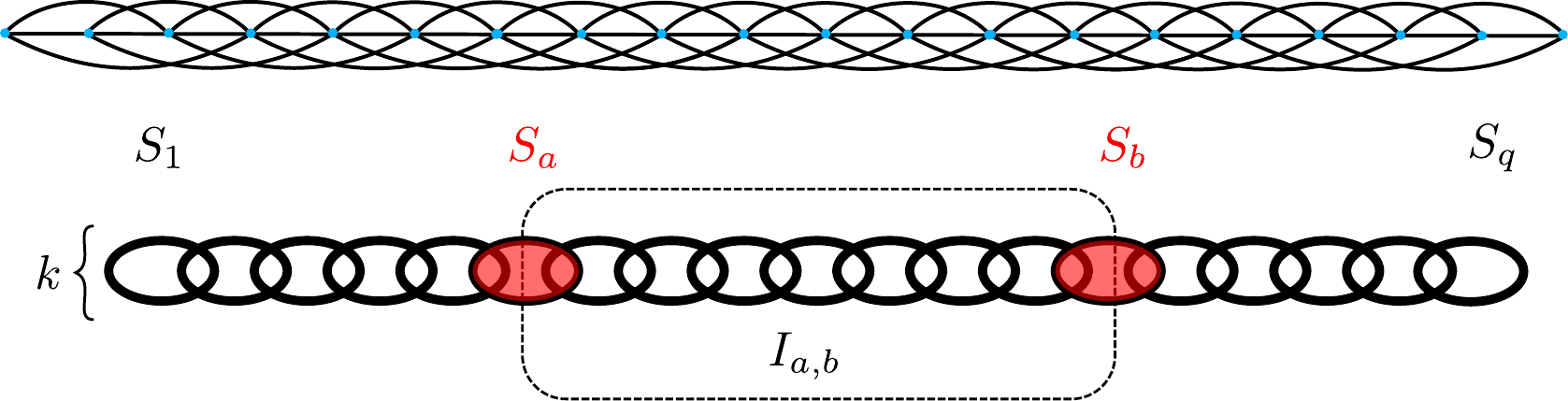}
    \caption{Top figure presents an example graph with bounded pathwidth. Figure
    below presents a definition of a path decomposition of width $k$, with bags $S_a,S_b$ and interval $I_{a,b}$.}
    \label{fig:pathwidth-definition}
\end{figure}

Let $G=(V,E)$ be the given graph.
Informally, the pathwidth of $G$ is parameter that measures how well $G$ can
be represented as a \emph{thickened path}, which formalized through the notion of a {\em{path decomposition}} of {\em{width}} $k$. The reader is invited to think about
a path decomposition of $G$ of width $k$ as a sequence of bags
$S_1,\ldots,S_q \subseteq V$, each of size at most $k$, that traverse the whole graph, i.e., $S_1\cup \ldots\cup S_q =V$ and subsequent bags differ
by exactly one vertex $|S_i \triangle S_{i+1}| \le 1$ (see
Figure~\ref{fig:pathwidth-definition} for an example of a graph with bounded
pathwidth and a schematic view of path decomposition). We may assume that $q\leq n$. Each bag is a separator in the sense that vertices present only in the bags to the left of it are pairwise non-adjacent to the vertices present only in the bags to the right of it.
In this section we focus on pathwidth in order to avoid several technical difficulties that arise when dealing with treewidth.

We first describe our isolation scheme for Hamiltonian cycles. The
crucial ingredient in our methods is a well-known hashing scheme due
to Fredman, Koml{\'{o}}s and Szemer{\'{e}}di~\cite{fks} (the FKS hashing lemma, see Section~\ref{sec:prelim} for a proof): For any set $A$ of $n$-bits integers with $|A| = n^{\Oh(1)}$, for most of the primes $p$ of order $|A|^{\Oh(1)}$ it holds that $x \not\equiv y \bmod p$ for all distinct $x,y \in A$.
An important property that is guaranteed by this lemma is that
after hashing modulo a prime $p$, \emph{every} element of the set $A$ is given a different value. 

\paragraph{Isolation scheme.}
Assume without loss of generality that $\log n$ is an integer. Our isolation scheme for $n$-vertex
graphs of pathwidth $k$ reads as follows.
Let $\id\colon E(G) \to \{1,\dots,|E(G)|\}$ be any bijection that assigns to each edge $e \in E(G)$ its unique \emph{identifier} $\id(e)$. 
First we select the range $M \coloneqq k^{\Oh(k)} \cdot n^{\Oh(1)}$ and
$\log{n}$ random prime numbers $p_1,\ldots,p_{\log{n}} \in
\{1,\ldots,M\}$. Note that we need $\Oh(k \log{k} \log{n} + \log^2{n})$
random bits to sample these prime numbers.

Next, we inductively define weights functions $\omega_1,\dots,\omega_{\log n}$ on $E(G)$ as follows:
\begin{itemize}
	\item Set $\omega_1(e) \coloneqq 2^{\id(e)} \bmod p_1$ for all $e \in E(G)$. 
	\item For each $e \in E(G)$ and $i=1,\dots,\log n$, set 
	$$ \omega_i(e) \coloneqq Mn \cdot \omega_{i-1}(e) + \left( 2^{\id(e)}\bmod p_i \right).  $$
\end{itemize}
Let $\omega \coloneqq \omega_{\log n}$ and observe that $\omega$ assigns weights bounded by $2^{\Oh(k \log{k} + \log^2{n})}$.
Note that the path decomposition $(S_1,\ldots,S_q)$ of $G$ is not used at all in the isolation scheme. We will use it only in the analysis, that is, the proof that the sampled weight function $\omega$ isolates the family of Hamiltonian cycles in $G$ with probability at least $\frac{1}{2}$.

%

\paragraph{Analysis.}
We first introduce the notion of an \emph{interval} in a path decomposition. This is just a graph induced by all the bags present between two given ones. More precisely, for $1\leq a\leq b\leq q$, we define
$$I_{a,b} \coloneqq \bigcup_{i \in [a,b]} S_i \subseteq V$$ to be the interval between bags $S_a$ and
$S_b$ (see Figure~\ref{fig:pathwidth-definition}). The length of this interval is  $|b-a|$ 
For an interval $I_{a,b}$ we say that $P_{a,b}\subseteq E(I_{a,b})$ is a \emph{partial solution} if it is a collection of vertex-disjoint paths with endpoints in $S_a\cup S_b$ that together visit all vertices of $I_{a,b}$.
Note that if we want to extend $P_{a,b}$ to a Hamiltonian cycle with another edge set $P'$, in order to check the feasibility of this extension we only need to know the pattern of connections induced by $P_{a,b}$ on $S_a$ and on $S_b$. 
More precisely, we only need to know the \emph{configuration} of $P_{a,b}$ on its boundary: such a configuration is represented by a matching $M$ on $S_a \cup S_b$, which indicates which vertices of $S_a\cup S_b$ are corresponding endpoints of a path in $P_{a,b}$, and information on how many edges of $P_{a,b}$ are incident on every vertex of $S_a\cup S_b$. 
Then $P_{a,b} \cup P'$ is a Hamiltonian cycle if and only if $P' \cup M$ is a Hamiltonian cycle on the vertices of $G - (I_{a+1,b-1}\cup V_2)$, where $V_2$ are vertices of $S_a\cup S_b$ incident on two edges of $P_{a,b}$. 

For an example of a realization of a configuration on $I_{a,b}$, see Figure~\ref{fig:configurations-pathwidth}.

\begin{figure}[ht!]
    \centering
    \includegraphics[width=0.7\textwidth]{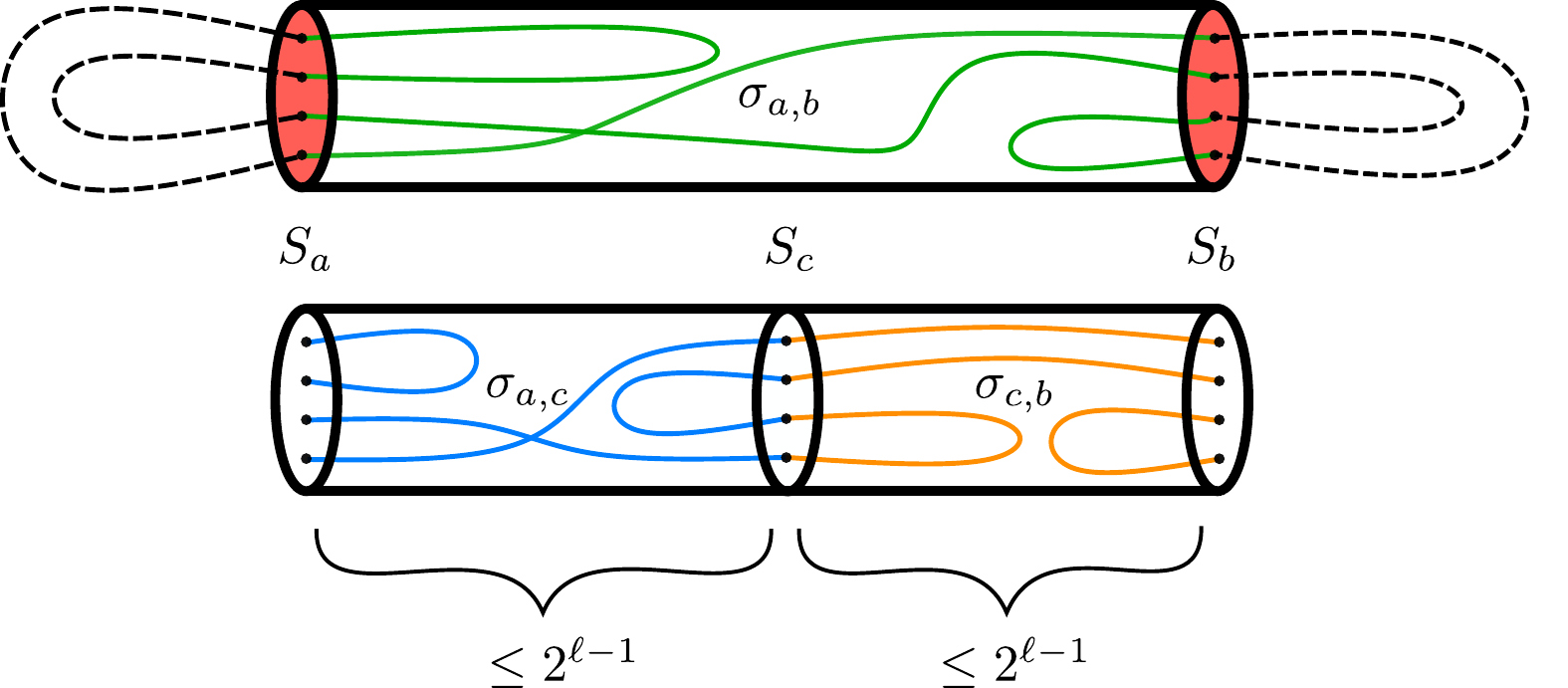}
    \caption{The top figure represents a realization of configuration $\sigma_{a,b}$
    in the interval $I_{a,b}$ (green stroked) and a complementary configuration
    for the rest of the graph (dashed lines). The figure below gives intuition for the
    induction argument. For a fixed $\sigma_{a,b}$ we know that $\omega_{\ell-1}$
    isolates all realizations in $I_{a,c}$ and $I_{c,b}$ and there are at most
    $2^{\Oh(k \log{k})}$ possible configurations $\sigma_{a,c},\sigma_{b,c}$
    that combined give $\sigma_{a,b}$.}
    \label{fig:configurations-pathwidth}
\end{figure}

Since a configuration is composed of an information about a
matching and a partition of vertices in $S_a \cup S_b$, the number of possible configurations within each interval $I_{a,b}$ is
at most $2^{\Oh(k \log k)}$. 

To prove that the weight function $\omega$ isolates the family of Hamiltonian cycles in $G$ with high probability, we prove the following claim by induction on $\ell$.

\IH{
    For every $\ell\in \{1,\ldots,\log{n}\}$, the following event happens with a sufficiently high probability: for every interval $I_{a,b}$ of length at most $2^\ell$ and a configuration $\sigma$ on $I_{a,b}$, the weight function $\omega_\ell$ isolates the family of all partial solutions in $I_{a,b}$ whose configuration is $\sigma$.
}\\

This induction hypothesis for $\ell=\log n$ immediately shows that $\omega = \omega_{\log{n}}$
isolates all Hamiltonian cycles with a sufficiently high probability. 

To prove the base case ($\ell=1$), 
we look at intervals of length at most $2$, that is, we look at the subgraphs $G[S_i]$ and $G[S_i\cup S_{i+1}]$, for $i \in \{1,\dots,q\}$.
Each of these subgraphs has at most $2k$ vertices, hence also at most $k^{\Oh(k)}$ different partial solutions. Also, there are at most $2n$ such subgraphs in total.
Hence, the total number of different partial solution in intervals of length $2$ is at most $n\cdot k^{\Oh(k)}$. We find from the FKS hashing lemma (see Lemma~\ref{FKS}) that if we choose the prime $p_1$ uniformly at random from the interval $\{1,\dots,M\}$, with a high probability all those partial solutions will be assigned pairwise different weights by the weight function $\omega_1$. Indeed, take $\psi(e)=2^{\id(e)}$ and $W = \{\psi(X): X$ a partial solution in interval of length at most 2$\}$. Note that all the $\psi(X)$ in $W$ for different $X$ are unique. Then FKS says that $\psi(X) \not\equiv \psi(X') \bmod p_1$ for any two such solutions $X$ and $X'$ with a high probability. This implies that with a high probability, $\omega_1(X)\neq \omega_1(X')$. The base case follows.

Now assume the induction hypothesis to be true for all $i < \ell$. Fix an interval $I_{a,b}$ for some $1\leq a\leq b\leq n$ of length $\le 2^{\ell}$. Observe that if $I_{a,b}$ has length 
at most $2^{\ell-1}$, then it is already appropriately taken care of by function $\omega_{\ell-1}$, and hence also by $\omega_\ell$.

Therefore, we can assume that the length of $I_{a,b}$ belongs to $[2^{\ell-1}+1, 2^{\ell}]$.
Fix some configuration $\sigma_{a,b}$ on $S_a \cup S_b$, the boundary of $I_{a,b}$. Observe that there exists $c \in
(a,b)$, such that $I_{a,c} \cup I_{c,b} = I_{a,b}$ and both $I_{a,c}$ and $I_{c,b}$ have length at most $2^{\ell-1}$. Further, there are at most $2^{\Oh(k \log{k})}$ different pairs of
configurations $\sigma_{a,c}$ and $\sigma_{c,b}$ that, when naturally combined, give a
configuration $\sigma_{a,b}$.
See Figure~\ref{fig:configurations-pathwidth} for a visualization.

The crucial observation is that by the induction hypothesis, for every pair of configurations $\sigma_{a,c},\sigma_{c,b}$ as above, the weight function $\omega_{\ell-1}$ already isolates the family of partial solutions in $I_{a,c}$ with configuration $\sigma_{a,c}$, as well as the family of partial solutions in $I_{c,b}$ with configuration $\sigma_{c,b}$. Therefore, for a fixed interval $I_{a,b}$ there can be at most $2^{\Oh(k \log{k})}$ different partial solution with configuration $\sigma_{a,b}$ that have minimum weight w.r.t. $\omega_{\ell-1}$. This is because they must be composed from  partial solutions in $I_{a,c}$ and $I_{c,b}$ that have minimum weights for their configurations. Moreover, there are at most $\Oh(n^2)$
different intervals of length in $[2^{\ell-1}+1,2^{\ell}]$. This means that in total, there
can be at most $2^{\Oh(k\log{k})} \cdot n^2$ different partial solutions in intervals $I_{a,b}$ of length at most $2^\ell$ that are minimum-weight realizations (w.r.t. $\omega_{\ell-1}$) of their respective configurations. Now moving from $\omega_{\ell-1}$ to $\omega_\ell$, we can argue using the FKS Lemma that all of these partial solutions will receive pairwise different values in $\omega_\ell$, with high probability.

This concludes the intuitive sketch of the proof of
the induction hypothesis. For a formal argument, see Section~\ref{sec:hamcycle}.

\paragraph*{Extensions of the method.}
All our isolation schemes for Hamiltonian cycle follow the same blueprint sketched above. The main difference, however, is
that we select our primes to be of the order $2^{\Oh(k)} \cdot n^{\Oh(1)}$. To argue that this is sufficient, we employ the rank-based approach to argue that the set of partial solutions that are ``representative enough'' is much smaller than $2^{\Oh(k \log{k})}$: it is actually of size $2^{\Oh(k)}$. In the above sketch, this reduces the number of random bits from $\Oh(k\log k\log n+\log^2 n)$ to $\Oh(k\log n+\log^2 n)$.

To complete the proof of Theorem~\ref{thm:hc-tree} it remains to lift the reasoning from graphs of bounded pathwidth to graphs of bounded treewidth. We do this by carefully generalizing the notion of an interval in a path decomposition to a notion of a {\em{segment}} in a tree decomposition. In particular, a segment of a tree decomposition can be always partitioned into at most five segments of twice smaller sizes, similarly we partitioned an interval into two intervals of at most half the length.


If we directly applied our analysis to the problem of isolating Hamiltonian cycles in $H$-minor-free graphs, then even with the rank-based approach employed we would only obtain an isolation scheme that uses $\Oh(\sqrt{n}\log{n})$ random bits.
To shave off the additional $\Oh(\log n)$ factor, we use certain properties of the decompositions of $H$-minor-free graphs that guarantee that size of separator decreases geometrically. 
In a nutshell, these properties will allow us to select $\Oh(\log{n})$ primes, but
each prime will be selected using a number of random bits that follows a geometric progression
(see Section~\ref{sec:HCdecomp} for details).

In Section~\ref{sec:mso} we generalize  the ideas to prove a meta-statement
about all problems definable in Monadic Second-Order logic, $\CMSOtwo$. The idea is that in the sketch above, we almost did not use any particular combinatorial properties of Hamiltonian cycles. The only property we relied on is that the behavior of a partial solution within an interval can be subsumed in a configuration on the interval's boundary, and the number of configurations is bounded by a function of $k$ only. Such a ``finite-state'' property is enjoyed by all problems definable in $\CMSOtwo$, which allows us to perform the whole reasoning on the meta-level.

Methods presented in Section~\ref{sec:indepedent_set} for isolating
local problems follow a completely different framework that uses
additional information about a graph. The analysis in this section is arguably
simpler. There, we use a technical contribution of Chari et
al.~\cite{isolation-lemma2} and extend it with
new observations regarding pivotal vertices in treedepth bounded graphs.

\section{Preliminaries}
\label{sec:prelim}


\paragraph*{Notation.} For an integer $k$, we write $[k]\coloneqq \{1,\ldots,k\}$. We use standard graph notation: $V(G)$ and $E(G)$ respectively denote the vertex set and the edge set of a graph $G$, for $X\subseteq V(G)$ the {\em{closed neighborhood}} $N_G[X]$ is $X$ plus all the neighbors of vertices of $X$, and the {\em{open neighborhood}} is $N_G(X)\coloneqq N_G[X]\setminus X$.

\paragraph*{Hashing modulo primes.} The following standard hashing lemma that dates back to the work of Fredman, Koml\'os, and Szemer\'edi~\cite{fks}, will be the main source of randomness in our isolation schemes.

\begin{lemma}[FKS hashing lemma~\cite{fks}]\label{FKS}
 Let $S\subseteq \{0,1,\ldots,2^n\}$ be a set of $k$ integers, where $n,k\geq 1$.
 Suppose that $p$ is a prime number chosen uniformly at random among prime numbers in the range $\{1,\ldots,M\}$, where $M\geq 2$. Then
 $$\Prb\left[x\not\equiv y\bmod p\quad\textrm{for all}\quad x,y\in S, x\neq y\right]\geq 1-\frac{nk^2}{\sqrt{M}}.$$
\end{lemma}
\begin{proof}
 Let
 $$R\coloneqq \prod_{x,y\in S, x\neq y}\ |x-y|.$$
 Note that $R\leq 2^{n\cdot \binom{k}{2}}$. This implies that $R$ may have at most $n\cdot \binom{k}{2}$ different prime divisors. On the other hand, from the prime number theorem it follows that $\pi(M)\in \Omega(\frac{M}{\log M})$, where $\pi(M)$ denotes the number of primes in the range $\{1,\ldots,M\}$. In fact, using a more precise estimate of Rosser~\cite{Rosser41}, for $M\geq 17$ we have $\pi(M)\geq \frac{M}{\ln M}$. For $2\leq M\leq 17$ a direct check shows that $\pi(M)\geq \sqrt{M}/2$. Since $\frac{M}{\ln M}\geq \sqrt{M}/2$ for all $M\geq 2$, we conclude that the probability that a random prime in the range $\{1,\ldots M\}$ is not among the at most $n\cdot \binom{k}{2}$ prime divisors of $R$ is at least
 $$1-\frac{n\cdot \binom{k}{2}}{\sqrt{M}/2}\geq 1-\frac{nk^2}{\sqrt{M}}.\qedhere$$
\end{proof}

\newcommand{\treedecomp}{\mathbb{T}}

\paragraph*{Graph decompositions.}
A {\em{rooted forest}} is directed acyclic graph $F$ where every node $x$ has at most one outneighbor, called the {\em{parent}} of $x$. A {\em{root}} is a node with no parent. If a node $y$ is reachable from $x$ by a directed path, then we write $y\preceq_F x$ and say that $y$ is an {\em{ancestor}} of $x$ and $x$ is a {\em{descendant}} of $y$. Note that every vertex is considered its own ancestor and descendant. For $x\in V(F)$, we write 
\begin{eqnarray*}
	\tail_F[x]\coloneqq \{y\colon y\preceq_F x\},& \qquad & \subtree_F[x]\coloneqq \{z\colon z\succeq_F x\},\\
	\tail_F(x)\coloneqq \tail_F[x]\setminus \{x\},& \qquad & \subtree_F(x)\coloneqq \subtree_F[x]\setminus \{x\}.
\end{eqnarray*}
The {\em{level}} of a node $x$ in $F$, denoted $\lvl_F(x)$, is the number of its strict ancestors, that is, $|\tail_F(x)|$. Note that roots have level $0$.
The {\em{height}} of a forest $F$ is the maximum level among its nodes, plus $1$.
If the forest $F$ is clear from the context, then we may omit it in the above notation.

An {\em{elimination forest}} of a graph $G$ is a rooted forest $F$ with $V(F)=V(G)$ such that for every edge $uv$ of $G$, either $u$ is an ancestor of $v$ in $F$ or vice versa. The {\em{treedepth}} of a graph $G$ is the least possible height of an elimination forest of $G$. Treedepth as a graph parameter plays a central role in the structural theory of sparse graphs, see~\cite[Chapters~6 and~7]{sparsity}. It also has several applications in parameterized complexity and algorithm design~\cite{Chenetal20,EisenbrandHKKLO19,furer-1,wg2020,PilipczukS19,pw-stacs2016}, as well as exhibits interesting combinatorial properties~\cite{Chenetal20,CzerwinskiNP19,DvorakGT12} and connections to descriptive complexity theory~\cite{ElberfeldGT12}. We refer to the introductory sections of the above works for a wider discussion.

A {\em{tree decomposition}} of a graph $G$ is a pair $\treedecomp = (T,\beta)$, where $T$ is an (unrooted) tree and $\beta\colon V(T)\to 2^{V(G)}$ is a function that assigns to each node $x\in V(T)$ its {\em{bag}} $\beta(x)\subseteq V(G)$ so that the following two conditions are satisfied:
\begin{itemize}[nosep]
 \item for each $u\in V(G)$, the set $\{x\colon u\in \beta(x)\}$ induces a nonempty and connected subtree of $T$; and
 \item for each $uv\in E(G)$, there exists $x\in V(T)$ such that $\{u,v\}\subseteq \beta(x)$. 
\end{itemize}
The {\em{width}} of $\treedecomp$ is $\max_{x\in V(T)}|\beta(x)|-1$ and the {\em{treewidth}} of $G$ is the minimum possible width of a tree decomposition of $G$. It is easy to see that the treedepth of a graph is at most its treewidth plus one. Conversely, the treewidth is upper bounded by the treedepth times the logarithm of the vertex count~\cite{sparsity}.

For surgery on tree decompositions we will use the following definition and standard lemma.

\begin{definition}[Segment of a tree]\label{def:segment} 
For an unrooted tree $T$, a {\em{segment}} of $T$ is a nonempty and connected subtree $I$ of $T$ such that there are at most two vertices of $I$ that have a neighbor outside of $I$. The set of those at most two vertices is the {\em{boundary}} of $I$, and is denoted by $\bnd I$. The {\em{size}} of $I$ is equal to $|E(I)|$. 
\end{definition}

\begin{lemma}\label{lemma:5SubTrees}
	Let $T$ be an unrooted tree and let $I$ be a segment of $T$ of size $\ell\geq 2$. Then there are at most $5$ segments $I_1,\ldots,I_t$ of $T$ ($t\leq 5$), each of size at most $\ell/2$, such that segments $I_1,\ldots,I_t$ have pairwise disjoint edge sets and $E(I_1)\cup \ldots \cup E(I_t)=E(I)$.
\end{lemma}
\begin{proof}
 For each edge $xy\in E(I)$, let $I_{y,x}$ and $I_{x,y}$ be the connected components of $I-xy$ that contain $x$ and $y$, respectively. Let $\vec{I}$ be the orientation of $I$ where each edge $xy$ is oriented towards $x$ if $|E(I_{y,x})|>|E(I_{x,y})|$ and towards $y$ if $|E(I_{y,x})|<|E(I_{x,y})|$; in case $|E(I_{y,x})|=|E(I_{x,y})|$, the edge $xy$ is oriented in any way. Since $I$ has $\ell$ edges and $\ell+1$ nodes, there is a node $z$ of $I$ that has outdegree $0$ in $\vec{I}$. This means that for every neighbor $x$ of $z$, we have $|E(I_{z,x})|\leq |E(I_{x,z})|$, implying $|E(I_{z,x})|<\ell/2$. Denote $I_x\coloneqq I_{z,x}$ and let $\widehat{I}_x$ be $I_x$ with the edge $xz$ added.
 
 We first argue that $I$ can be edge-partitioned into at most $3$ subtrees (not necessarily segments), each with at most $\ell/2$ edges. Consider first the corner case when there exists a neighbor $x$ of $z$ such that $\widehat{I}_x$ has more than $\ell/2$ edges. Then both $I_x=I_{z,x}$ and $I_{x,z}$ have exactly $\frac{\ell-1}{2}$ edges each, so we can partition $I$ into $I_{z,x}$, $I_{x,z}$, and a separate subtree consisting only of the edge $xz$. This case being resolved, we can assume that each tree $\widehat{I}_x$ has at most $\ell/2$ edges. Starting with the set of trees $\mathcal{T}\coloneqq\{\widehat{I}_x\colon x\textrm{ is a neighbor of }z\}$, iteratively apply the following procedure: take two trees from $\mathcal{T}$ with the smallest edge counts, and replace them with their union, provided this union has at most $\ell/2$ edges. The procedure stops when this assertion fails to be satisfied. Observe that the procedure can be carried out as long as $|\mathcal{T}|\geq 4$, for then the two trees from $\mathcal{T}$ that have the smallest edge counts together include at most half of the edges of $I$. Therefore, at the end we obtain the desired edge-partition of $I$ into at most three subtrees.
 
 All in all, in both cases we edge-partitioned $I$ into at most three subtrees, each having at most $\ell/2$ edges. Since $|\bnd I|\leq 2$, it is easy to see that all of those subtrees are already segments (i.e. have boundaries of size at most $2$) apart from at most one, say $J$, which may have a boundary of size $3$. Supposing that $J$ exists, let $\bnd J=\{a,b,c\}$. Then there exists a node $d$ of $J$ such that every connected component of $J-d$ contains at most one of the vertices $a,b,c$. It is now straightforward to edge-partition $J$ into three trees so that the boundary of each of them consists of $d$ and one of the vertices $a,b,c$. Thus, replacing $J$ with those three segments yields an edge-partition of $I$ into at most $5$ segments, each with at most $\ell/2$ edges.
\end{proof}

\section{Isolating Hamiltonian cycles}\label{sec:hamcycle}
In this section we prove Theorems~\ref{thm:hc-tree},~\ref{thm:hc-gen}, and~\ref{thm:planar}. We begin by defining {\em{configurations}} for Hamiltonian cycles, which reflect the states of a natural dynamic programming algorithm for detection of a Hamiltonian cycle in a bounded-treewidth graph. Then we use the rank-based approach to bound the number of {\em{minimum weight compliant edge sets}} (see Theorem~\ref{thm:rankbased}). This technical result captures the essence of the rank-based approach and will be used in all subsections that follow. Next, we prove Theorem~\ref{thm:hc-gen} in Section~\ref{sec:HCgeneral}. Then Theorem~\ref{thm:hc-tree} is proved in Section~\ref{sec:ham-tw}. Finally, in Section~\ref{sec:HCdecomp} we first recall basic definitions and facts about separable graph classes, then we give a decomposition theorem (Theorem~\ref{thm:separable}) for such classes that produces a low-depth elimination forest with several important technical properties, and finally we use this decomposition theorem to prove Theorem~\ref{thm:planar}. 

\subsection{Configurations for Hamiltonian cycles}
Let us fix a graph $G$. An edge set $S\subseteq E(G)$ is called a \emph{partial solution} if every vertex of $G$ is incident to at most two edges of $S$ and $S$ has no cycles.
The following notion of a \emph{configuration} describes the behavior of a partial solution with respect to a set of vertices.

\begin{definition}[Configurations]
	For $X \subseteq V(G)$, we define the set of {\em{configurations}} on $X$ as: 
	$$\conf(X)\coloneqq\{\,(V_0,V_1,V_2,M)\ \colon\ (V_0,V_1,V_2) \textrm{ is a partition of }X\textrm{ and }M\textrm{ is a perfect matching on }V_1\,\}.$$ 
\end{definition}

Given a subgraph $H$ of $G$, one can view the configurations on $X \subseteq V(H)$ as all possible different ways that a partial solution may behave on $X$. 
A vertex is then in the set $V_i$ if it is incident to exactly $i$ edges of the partial solution. The matching $M$ on $V_1$ describes the endpoints of each path in the partial solution. This intuition is formalized in the following definition.

\begin{definition}  \label{def:c(S)}
	Let $X\subseteq V(G)$ be a set of vertices of $G$ and let $S \subseteq E(G)$ be a partial solution. Then define the {\em{configuration of $S$ on $X$}} as	
	$c_X(S) \coloneqq (V_0,V_1,V_2,M) \in \conf(X)$, where
	\begin{itemize}[nosep]
		\item $V_0\coloneqq \{ v \in X\colon v$ is not incident to any edge of $S\}$,
		\item $V_1\coloneqq \{ v \in X\colon v$ is incident to exactly one edge of $S\}$,
		\item $V_2\coloneqq \{ v \in X\colon v$ is incident to exactly two edges of $S\}$,
		\item $M\coloneqq \{ \{u,v\} \in \binom{V_1}{2}\colon$ there is a path with edges from $S$ connecting $u$ and $v\}$,
	\end{itemize}
	We omit $X$ in the notation and write $c(S)$ when $X$ is clear from context. 
\end{definition}
Note that in the above definition $M$ is indeed a matching, because each $v \in V_1$ is connected to exactly one $u \in V_1$ through $S$, as any partial solution covers each vertex at most twice. For an example of deriving $c_X(S)$ from a partial solution $S$, see Figure~\ref{fig:configurations}.

\begin{figure}[ht!]
	\centering
	\begin{minipage}{.45\textwidth}
		\centering
		\vspace{2pt}
		\scalebox{0.7}{\begin{tikzpicture}[scale = 0.7,shorten >=1pt,auto,node distance=1cm,
	thick,main node/.style={circle,fill=white!20,draw,font=\sffamily\Large\bfseries}]

	\draw[black, rounded corners=2pt] (-0.5,-0.5) rectangle (0.5,8.5);
	\node[] at (0,9.5) {\Large $X$};
	
	\draw[blue,rounded corners= 30pt]  (0,0)-- (2.6,0.5) --  (0,1);

	\draw[blue] (0,1) .. controls (1.5,1.1) and (1.5,1.9) .. (0,2);
	
	\draw[blue] (0,4) .. controls (3.5,3.8) and (3.5,5.2) .. (0,5);

	\draw[blue,rounded corners = 10pt] (0,0) -- (2.5,-0.05) -- (4.2,0.5) -- (2.8,1) -- (4.2,2) -- (2,3.03) --(0,3);
	
		\filldraw[black] (0,0) circle (2pt);
	\filldraw[black] (0,1) circle (2pt);
	\filldraw[black] (0,2) circle (2pt);
	\filldraw[black] (0,3) circle (2pt);
	\filldraw[black] (0,4) circle (2pt);
	\filldraw[black] (0,5) circle (2pt);
	\filldraw[black] (0,6) circle (2pt);
	\filldraw[black] (0,7) circle (2pt);
	\filldraw[black] (0,8) circle (2pt);
	
	\filldraw[black] (1.7,0.5) circle (2pt);
	\filldraw[black] (2,4.9) circle (2pt);
	\filldraw[black] (2,4.1) circle (2pt);
	
	\filldraw[black] (2.5,0) circle (2pt);
	\filldraw[black] (3.9,0.5) circle (2pt);
	\filldraw[black] (3.1,1) circle (2pt);
	\filldraw[black] (3.9,2) circle (2pt);
	\filldraw[black] (2,2.95) circle (2pt);

	\draw[black,rounded corners=2pt] (8.5,5.6) rectangle (9.5,8.5);
	\draw[black,rounded corners=2pt] (8.5,1.6) rectangle (9.5,5.4);
	\draw[black,rounded corners=2pt] (8.5,-0.5) rectangle (9.5,1.4);
	\node[] at (9,9.5) {\Large $X$};

	\draw[red] (9,2) .. controls (10.5,2.1) and (10.5,2.9) .. (9,3);
	\draw[red] (9,4) .. controls (10.5,4.1) and (10.5,4.9) .. (9,5);
	\node[red] at (10,3.5) {\Large $M$};
	\node[] at (8,7) {\Large $V_0$};
	\node[] at (8,3.5) {\Large $V_1$};
	\node[] at (8,0.5) {\Large $V_2$};
	
	\node[blue] at (3,3.5) {\Large $S$};
	
	\filldraw[black] (9,0) circle (2pt);
	\filldraw[black] (9,1) circle (2pt);
	\filldraw[black] (9,2) circle (2pt);
	\filldraw[black] (9,3) circle (2pt);
	\filldraw[black] (9,4) circle (2pt);
	\filldraw[black] (9,5) circle (2pt);
	\filldraw[black] (9,6) circle (2pt);
	\filldraw[black] (9,7) circle (2pt);
	\filldraw[black] (9,8) circle (2pt);

	\draw[->] (4.8,4) -- (6.5,4);
	\node[] at (5.5,4.5) {\Large{$c_X(S)$}};
	
\end{tikzpicture}}
		\captionof{figure}{Example partial solution $S$ and its configuration $c_X(S)=(V_0,V_1,V_2,M)$ on a set~$X$.}
		\label{fig:configurations}
	\end{minipage} \hspace{30pt}
	\begin{minipage}{.4\textwidth}
		\centering
		\scalebox{0.7}{\begin{tikzpicture}[scale = 0.7,,shorten >=1pt,auto,node distance=1cm,thick,main node/.style={circle,fill=white!20,draw,font=\sffamily\Large\bfseries}]
		
		\draw [draw=black, fill=black,fill opacity=0.05, draw opacity = 0.5,rounded corners= 2pt]
		(-0.6,-0.6) rectangle (5,8.6) ;

		\draw[red] (0,2) .. controls (-1.5,2.1) and (-1.5,2.9) .. (0,3);
		\draw[red] (0,4) .. controls (-1.5,4.1) and (-1.5,4.9) .. (0,5);
		
		\draw[blue] (0,3) .. controls (3,2) and (3,5) .. (0,4);
		\draw[blue] (0,5) .. controls (2.7,5) and (2.7,6) .. (0,6);
		\draw[blue] (0,6) .. controls (1.5,6) and (1.5,7) .. (0,7);
		\draw[blue] (0,7) .. controls (2,7) and (2,8) .. (0,8);
		\draw[blue,rounded corners=10 pt] (0,8) -- (2.15,8.05) -- (3.2,7.4) -- (4.14,6.5) -- (3.75,3) -- (0,2);

		\filldraw[black] (0,0) circle (2pt);
		\filldraw[black] (0,1) circle (2pt);
		\filldraw[black] (0,2) circle (2pt);
		\filldraw[black] (0,3) circle (2pt);
		\filldraw[black] (0,4) circle (2pt);
		\filldraw[black] (0,5) circle (2pt);
		\filldraw[black] (0,6) circle (2pt);
		\filldraw[black] (0,7) circle (2pt);
		\filldraw[black] (0,8) circle (2pt);
		
		\filldraw[black] (3.5,3) circle (2pt);
		\filldraw[black] (2,5.5) circle (2pt);
		\filldraw[black] (2,4) circle (2pt);
		\filldraw[black] (2,4) circle (2pt);
		\filldraw[black] (1.5,7.5) circle (2pt);
		\filldraw[black] (2,8) circle (2pt);
		\filldraw[black] (4,6.5) circle (2pt);

		\filldraw[black] (3,7.5) circle (2pt);
		\filldraw[black] (2,3) circle (2pt);

		\draw[black,rounded corners= 2pt] (-.5,5.6) rectangle (0.5,8.5);
		\draw[black,rounded corners= 2pt] (-.5,1.6) rectangle (0.5,5.4);
		\draw[black,rounded corners= 2pt] (-.5,-0.5) rectangle (0.5,1.4);
		\node[] at (0,9.5) {\Large $X$};
		\node[white!50!black] at (3,9.5) {\Large $V(H)$};

		\node[red] at (-2,3) {\Large $M$};
		\node[] at (-1,7) {\Large $V_0$};
		\node[] at (-1,3.5) {\Large $V_1$};
		\node[] at (-1,0.5) {\Large $V_2$};
		\node[blue] at (4.5,4) {\Large $S$};

\end{tikzpicture}}
		\captionof{figure}{Example compliant partial solution $S$ for a configuration $c=(V_0,V_1,V_2,M) \in \conf(X)$.}
		\label{fig:compliant}
	\end{minipage}
\end{figure}

We can use configurations to tell whether two partial solutions together form a Hamiltonian cycle. 
Let $H$ be a subgraph of $G$ and let $X \subseteq V(H)$. 
Assume that there exists a partial solution $S$ that visits only vertices from $(V(G)\setminus V(H))\cup X$, where every vertex of $V(G)\setminus V(H)$ is visited exactly twice. Then we only need to know $c_X(S)$ to determine which partial solutions $S'\subseteq E(H)$ would combine with $S$ to a Hamiltonian cycle in $G$. We say that any such partial solution is \emph{compliant} with $c_X(S)$, as expressed formally in the next definition.

\begin{definition}[Compliant partial solution]
	Let $H$ be a subgraph of $G$ and let $X \subseteq V(H)$. A configuration $c = (V_0,V_1,V_2,M) \in \conf(X)$ and a partial solution $S \subseteq E(H)$ are \emph{compliant} if $S\cap M=\emptyset$ and $S \cup M$ forms a Hamiltonian cycle on $V(H) \setminus V_2$.
\end{definition}

See Figure~\ref{fig:compliant} for an example of a compliant partial solution.

In the sequel we will be trying to argue that some weight function $\omega$ is isolating the family of Hamiltonian cycles in the given graph $G$ with high probability. In all cases this will be done by induction on larger and larger subgraphs of $G$, where at each point we argue that a suitable family of partial solutions is isolated with high probability. The following definition facilitates this discussion.

\begin{definition}[Minimum weight compliant partial solution]
	Let $H$ be a subgraph of $G$, $X \subseteq V(H)$, $c \in \conf(X)$, and let $\omega\colon E(G) \to \N$ be a weight function on the edges of $G$. Then we define the set $\Min(\omega,H,c)$ of {\em{minimum weight partial solutions compliant with $c$}} as the set of those partial solutions $S\subseteq E(H)$ that
	\begin{itemize}[nosep]
	 \item are compliant with $c$, and
	 \item subject to the above, have the smallest possible weight $\omega(S)$.
	\end{itemize}
\end{definition}

\subsection{Rank-based approach}
We will use the {\em{rank-based approach}}, introduced by Cygan et al. in~\cite{CyganKN18}, as a tool in our analysis of isolation schemes. Let $X$ be a set of vertices. Then define the \emph{compatibility matrix} $\Hh_X$ as the matrix with entries indexed by $\Hh_X[M_1,M_2]$ for $M_1, M_2$ perfect matchings on $X$, where
$$\Hh_X[M_1,M_2]= \begin{cases}1 &\text{ if } M_1 \cup M_2 \text{ is a simple cycle,}\\ 0 &\text{ otherwise.} \end{cases}$$
Note that $\Hh_X[M_1,M_2]$ has $2^{\Oh(|X|\log |X|)}$ rows and columns. The crux of the rank-based approach is that in spite of that, this matrix has a small rank over the two-element field $\mathbb{F}_2$.

\begin{theorem}[Rank-based approach,\cite{CyganKN18}]\label{thm:RankBased}
For any set $X$, the rank of $\mathcal{H}_X$ over $\mathbb{F}_2$ is equal to $2^{|X|/2-1}$.	
\end{theorem}

We use Theorem~\ref{thm:RankBased} to prove that the total number of minimum weight compliant solutions is always relatively small, no matter what the weight function is. The following statement will be reused several times in the sequel. Note that a trivial cardinality argument would yield an upper bound of the form $2^{\Oh(|X|\log |X|)}$; the point of the rank-based approach is to reduce this to $2^{\Oh(|X|)}$.

\begin{theorem}\label{thm:rankbased} 
Let $G$ be a graph, $X\subseteq V(G)$, and $\omega\colon V(G) \to \N$ be a weight function such that for all $c \in \conf(X)$, we have $|\Min(\omega,G,c)|\le 1$. Then $$\left|\bigcup_{c \in \conf(X)}\Min(\omega,G,c)\right|\le 2^{\Oh(|X|)}.$$ 
\end{theorem}
\begin{proof}
	Let $K \coloneqq \bigcup_{c \in \conf(X)}\Min(\omega,G,c)$ and let $C \coloneqq \{c(S)\colon S \in K\}$. 
	
	We first verify that $|C|=|K|$.
	By construction, we have $|C|\le |K|$. Assume for contradiction that $|C| < |K|$. Then there are two different partial solutions $S_1, S_2 \in K$ such that $c(S_1)=c(S_2)$. By construction and the assumptions, there are two different configurations $d_1,d_2\in \conf(X)$ such that $\Min(\omega,G,d_1)=\{S_1\}$ and $\Min(\omega,G,d_2)=\{S_2\}$. However, since $c(S_1)=c(S_2)$, it follows that for any configuration $d\in \conf(X)$, $S_1$~is compliant with $d$ if and only if $S_2$ is compliant with $d$. In particular, $S_1$ is compliant with $d_2$ and $S_2$ is compliant with $d_1$. This implies that $\omega(S_1)=\omega(S_2)$ and $S_2\in \Min(\omega,G,d_1)$ and $S_1\in \Min(\omega,G,d_2)$, a contradiction. Hence $|C|=|K|$. 
	
	Define a matrix $\widehat{\Hh}$ with both coordinates indexed by $\conf(X)$ such that for $c,c' \in \conf(X)$, where $c=(V_0,V_1,V_2,M)$ and $c'=(V_0',V_1',V_2',M')$:
	$$\widehat{\Hh}[c,c'] = \begin{cases}1 &\text{ if } V_0 = V_2',\, V_2 = V_0' \text{, and }M\cup M'\text{ is a simple cycle,}\\ 0 &\text{ otherwise.} \end{cases}$$  
	Notice that if we sort the indices of $\widehat{\Hh}$ by the partitions $(V_0,V_1,V_2)$, then $\widehat{\Hh}$ can be seen as a block diagonal matrix with one block for each partition, and this block is a compatibility matrix on $V_1$. That is, $$\widehat{\Hh} = \bigoplus_{V_0 \uplus V_1 \uplus V_2 = X}\Hh_{V_1},$$
	where $\bigoplus$ denotes the operator of combining several matrices into a single block diagonal matrix.
	By Theorem~\ref{thm:RankBased}, the rank over $\mathbb{F}_2$ of each of these blocks is bounded by $2^{|X|/2 -1}$, hence the rank over $\mathbb{F}_2$ of $\widehat{\Hh}$ is bounded by  $2^{|X|/2 -1} \cdot 3^{|X|} \le 2^{\Oh(|X|)}$.
	
	\newcommand{\mx}{\mathsf{max}}
	
	Next, we claim that the set of rows of $\widehat{\Hh}$ corresponding to the configurations of $C$ is linearly independent over $\mathbb{F}_2$. Assume not, hence there is a nonempty set of configurations $D\subseteq C$ such that $$\sum_{d\in D} \widehat{\Hh}[d, \cdot] = \mathbf{0},$$
	where $\mathbf{0}$ is the all-zero vector (all computations are performed in $\mathbb{F}_2$).
	For each $d\in D$ there is some $S_d\in K$ such that $d=c(S_d)$. Let $d_{\mx}$ be a configuration of $D$ for which $\omega(S_{d_{\mx}})$ is the largest possible. Since $d_{\mx} \in C$, we have that $\Min(\omega,G,c)=\{S_{d_{\mx}}\}$ for some $c \in \conf(X)$ and hence $\widehat{\Hh}[d_{\mx},c]=1$. However, as $\sum_{d\in D} \widehat{\Hh}[d, \cdot] = \mathbf{0}$, there must be another $d'\in D$, $d'\neq d_{\mx}$, such that also $\widehat{\Hh}[d',c] =1$.
	This means that $d'$ is compliant with $c$, which implies that $\omega(S_{d'})>\omega(S_{d_{\mx}})$ by $\Min(\omega,G,c)=\{S_{d_{\mx}}\}$. This contradicts the maximality of $\omega(S_{d_{\mx}})$.
	
	We conclude that the set of rows of $\widehat{\Hh}$ corresponding to $C$ are indeed linearly independent over $\mathbb{F}_2$. Therefore, $|K|=|C|$ is upper bounded by the rank of $\widehat{\Hh}$ over $\mathbb{F}_2$, which is at most $2^{\Oh(|X|)}$. 
\end{proof}

\subsection{Hamiltonian cycles in general graphs using $\Oh(n)$ random bits} \label{sec:HCgeneral}

We now use the tools prepared so far to prove Theorem~\ref{thm:hc-gen}.
The goal is to isolate all Hamiltonian cycles in an undirected graph $G=(V,E)$ using $\Oh(n)$ random bits, where $n$ is the vertex count. First we give the isolation procedure. Then we analyze the probability of isolating all Hamiltonian cycles using configurations, compliant partial solutions, and the rank-based approach (through Theorem~\ref{thm:rankbased}). Throughout the subsection we assume without loss of generality that $\log n$ is an integer. 

As usual with isolation schemes, we assume that the vertex set of the considered graph $G$ is $V=[n]$. We will apply induction on specific subgraphs of $G$ called {\em{intervals}}.
\begin{definition}[Interval of $G$]
For integers $1\leq s\leq t\leq n$ and $1\leq s'\leq t'\leq n$, the {\em{interval}} $G\angles{s,t,s',t'}$ is the graph $(V',E')$, where 
$$V'\coloneqq \{s,\dots,t\} \cup \{s',\dots,t'\}\qquad\textrm{and}\qquad E' \coloneqq \{uv\colon u\in \{s,\dots,t\}, v\in \{s',\dots,t'\}, uv\in E\}.$$ By $V\angles{s,t,s',t'}$ we denote the vertex set $V'$ of the interval $G\angles{s,t,s',t'}$. 
\end{definition}
Note that $G\angles{s,t,s,t}$ is just the subgraph of $G$ induced by $\{s,\dots,t\}$. On the other hand, if $\{s,\dots,t\} \cap \{s',\dots,t'\} = \emptyset$, then $G\angles{s,t,s',t'}$ is a bipartite graph, with $\{s,\dots,t\}$ and $\{s',\dots,t'\}$ being the sides of the bipartition.

\paragraph{Isolation scheme.}
We first present the isolation scheme. Let $\id\colon E(G) \to \{1,\dots,|E(G)|\}$ be any bijection that assigns to each edge $e \in E(G)$ its unique \emph{identifier} $\id(e)$. Let $C$ be some large enough constant, to be chosen later. Then independently at random sample $1+\log n$ primes 
$p_0,p_1,\dots,p_{\log n}$
so that $p_i$ is sampled uniformly among primes in the range $\{1,\ldots,M_i\}$, where $M_i \coloneqq 2^{C(\log n + 2^i)}$. Note that choosing each $p_i$ requires $C(\log n + 2^i)$ random bits, hence we have used $\Oh(n)$ random bits in total. 

Next, we inductively define weights functions $\omega_0,\dots,\omega_{\log n}$ on $E(G)$ as follows:
\begin{itemize}
	\item Set $\omega_0(e) \coloneqq 2^{\id(e)} \bmod p_0$ for all $e \in E(G)$. 
	\item For each $e \in E(G)$ and $i=1,\dots,\log n$, set 
	$$ \omega_i(e) \coloneqq M_{i-1}n \cdot \omega_{i-1}(e) + \left( 2^{\id(e)}\bmod p_i \right).  $$
\end{itemize}
Let $\omega \coloneqq \omega_{\log n}$ and observe that $\omega$ assigns weights bounded by $2^{\Oh(n)}$, as required.


\paragraph*{Analysis.} We will prove the following statement for all $0\le i \le \log n$ using induction on $i$.
	
	\IH{With probability at least $\left(1 - \frac{1}{n^{2}} \right)^{i+1}$, for all intervals $G\langle s,t,s',t' \rangle$ s.t. $t-s \le 2^i$ and $t'-s' \le 2^i$ and for each configuration $c \in \conf(V\angles{s,t,s',t'})$, there is at most one minimum weight (w.r.t. $\omega_i$) compliant partial solution, i.e. $|\Min({\omega_i},{G\langle s,t,s',t' \rangle},c)|\le 1$.}\\
	
	For $i=\log n$, the induction hypothesis gives us that for the complete interval $G=G\angles{1,1,n,n}$ and for the configuration $c = (\emptyset,\emptyset,V(G),\emptyset)$, there is at most one minimum weight compliant partial solution w.r.t. $\omega$. In other words, w.r.t. $\omega$ there is at most one minimum weight Hamiltonian cycle in $G$. This happens with probability at least $\left(1  - \frac{1}{n^{2}} \right)^{\log n +1} \ge 1  - \frac{1}{n}$. So it remains to perform the induction.
	
    \paragraph{Base step.} For $i=0$, we have $t-s \le 1$ and $t'-s' \le 1$. Hence each such interval $G\angles{s,t,s',t'}$ has at most $4$ edges. 
	Let $$Y \coloneqq \bigcup_{\substack{t-s\leq 1\\ t'-s'\leq 1}} 2^{E(G\angles{s,t,s',t'})}$$ and for each $S \in Y$, let $$x_S \coloneqq \sum_{e\in S}2^{\id(e)}.$$
	Observe that since the identifiers assigned to the edges are unique, the numbers $x_S$ are also pairwise different. 
	Also, note that $|Y|\le 16n^2$ as there are at most $n^2$ intervals
    considered, and for each of them there are at most $16$ possible subsets of
    the at most four edges. Recall that $M_0 = 2^{C(\log n +1)}$ and $p_0$ is
    drawn uniformly at random among the primes in the range $\{1,\dots,M_0\}$.
    Therefore, from Lemma~\ref{FKS} we can conclude that with probability at least 
	$$\left(1 -\frac{(n^2+1)(16n^2)^2}{2^{(C/2)(\log n +1)}}\right) \ge \left(1-\frac{1}{n^2}\right)$$ all the numbers $\{x_S\colon S \in Y\}$ have pairwise different remainders modulo $p_0$; here the last inequality holds for a large enough constant $C$.
	Since $\omega_0(S) \equiv x_S \bmod p_0$, this means that with probability at least $\left(1 -\frac{1}{n^2} \right)$, all $S \in Y$ receive pairwise different weights with respect to $\omega_0$. 
	Therefore, the induction hypothesis is true for $i=0$.

    \paragraph{Induction step.} Assume the induction hypothesis is true for all intervals $G\langle s,t,s',t' \rangle$ such that $t-s\le2^{i-1}$ and $t'-s'\le2^{i-1}$. 
	Let 
	\[
		Y' \coloneqq \bigcup_{\substack{t-s\leq 2^{i-1}\\ t'-s'\leq 2^{i-1}}}\  \bigcup_{c\in \conf(V\angles{s,t,s',t'})}\  \Min(\omega_{i-1},G\angles{s,t,s',t'},c)
	\]
		be the set of all the minimal partial solutions for those intervals. 
	Further, let $$Y \coloneqq \{ S_1 \cup S_2 \cup S_3 \cup S_4 \colon S_1, S_2, S_3, S_4 \in Y'\}$$ 
	be the set containing all combinations of four such partial solutions. The strategy is as follows. We first prove in Claim~\ref{claim:1B} that any relevant minimum weight compliant partial solution should be in $Y$.
	Then Claim~\ref{claim:1A} says that with hight probability, all partial solutions $S \in Y$ have pairwise different weights with respect to $\omega_i$. Hence, proving these two claims will be sufficient to make the induction hypothesis go through.
		
	\begin{claim}\label{claim:1B}
	Let $1\leq a\leq b\leq n$ and $1\leq a'\leq b'\leq n$ be such that $b-a\le 2^{i}$ and $b'-a' \le 2^{i}$, and let $c \in \conf(a,b,a',b')$. Then $\Min({\omega_{i}},{G\langle a,b,a',b' \rangle},c)\subseteq Y$.
	\end{claim}
\begin{proof}
	Take any $S \in \Min({\omega_{i}},{G\langle a,b,a',b' \rangle},c)$. 
	Let 
	$$r=\lceil (a+b)/2\rceil\qquad\textrm{and}\qquad r' =\lceil (a'+b')/2\rceil$$
	and let us select
	\begin{eqnarray*}& S_1\subseteq E(G\langle a,r-1,a',r'-1\rangle),&\qquad S_2 \subseteq E(G\langle a,r-1,r',b'\rangle),\\ & S_3 \subseteq E(G \langle r,b,a',r'-1\rangle),&\qquad S_4 \subseteq E(G\langle r,b,r',b'\rangle)
	\end{eqnarray*}
  so that $S_1,S_2,S_3,S_4$ are disjoint and $S = S_1 \cup S_2 \cup S_3 \cup S_4$. See Figure~\ref{fig:intervals} for an example.
  
	We argue that $S_1 \in \Min(\omega_{i-1},G\angles{a,r-1,a',r'-1},c_1)$ for some $c_1\in \conf(V\langle a,r-1,a',r'-1\rangle)$. Let $c=(V_0,V_1,V_2,M)$. Since $S\cup M$ is a simple cycle that visits all vertices of $V\angles{a,b,a',b'}$, we see that $R\coloneqq S_2\cup S_3\cup S_4\cup M$ is a partial solution in the graph $G\angles{a,b,a',b'}$ with the edges of $M$ added. Letting $(V_0',V_1',V_2',M')\coloneqq c_{V\angles{a,r-1,b,r-1}}(R)$, it follows that $S_1$ is compliant with the configuration
	$$c_1\coloneqq (V_0'\setminus (V_2\cap V\angles{a,r-1,b,r-1})),V_1',V_2'\cup (V_2\cap V\angles{a,r-1,b,r-1}),M').$$
	Moreover, that $S\in \Min(\omega_i,G\angles{a,b,a',b'},c)$ implies that $S_1\in \Min(\omega_i,G\angles{a,r-1,a',r'-1},c_1)$, for otherwise $S_1$ could be replaced in $S$ with a smaller-weight partial solution $S_1'$ that would be still compliant with $c_1$, and this would turn $S$ into a smaller-weight partial solution $S'=S_1'\cup S_2\cup S_3\cup S_4$ that would be still compliant with $c$. Finally, by the construction of $\omega_i$, $S_1\in \Min(\omega_i,G\angles{a,r-1,a',r'-1},c_1)$ entails $S_1\in \Min(\omega_{i-1},G\angles{a,r-1,a',r'-1},c_1)$.
	
	Therefore $S_1\in Y'$. Analogously we argue that $S_2,S_3,S_4 \in Y'$, hence we conclude that $S \in Y$.
\end{proof}

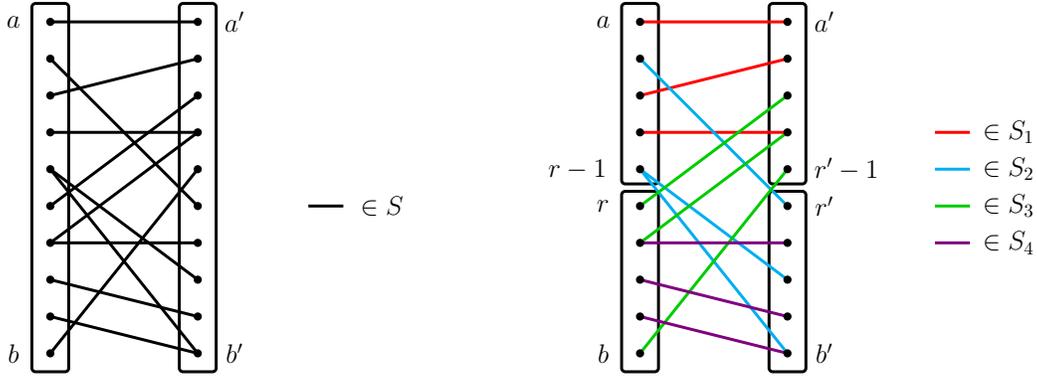
\begin{figure}[ht!]
	\centering
	\scalebox{0.7}{\begin{tikzpicture}[scale = 0.7,shorten >=1pt,auto,node distance=1cm,
	ultra thick,main node/.style={circle,fill=white!20,draw,font=\sffamily\Large\bfseries}]

	\draw[black,rounded corners=2pt] (-6.5,-0.5) rectangle (-5.5,9.5);
	\draw[black,rounded corners=2pt] (-2.5,-0.5) rectangle (-1.5,9.5);

	\node[] at (-7,9) {\Large $a$};
	\node[] at (-7,0) {\Large $b$};
	\node[] at (-1,9) {\Large $a'$};
	\node[] at (-1,0) {\Large $b'$};
	
	\draw[black] (-6,9) -- (-2,9);
	\draw[black] (-6,7) -- (-2,8);
	\draw[black] (-6,6) -- (-2,6);
	
	\draw[black] (-6,8) -- (-2,4);
	\draw[black] (-6,5) -- (-2,2);
	\draw[black] (-6,5) -- (-2,0);
	
	\draw[black] (-6,4) -- (-2,7);
	\draw[black] (-6,0) -- (-2,5);
	\draw[black] (-6,3) -- (-2,6);
	
	\draw[black] (-6,3) -- (-2,3);
	\draw[black] (-6,2) -- (-2,1);
	\draw[black] (-6,1) -- (-2,0);
	
	\filldraw[black] (-6,0) circle (2pt);
	\filldraw[black] (-6,1) circle (2pt);
	\filldraw[black] (-6,2) circle (2pt);
	\filldraw[black] (-6,3) circle (2pt);
	\filldraw[black] (-6,4) circle (2pt);
	\filldraw[black] (-6,5) circle (2pt);
	\filldraw[black] (-6,6) circle (2pt);
	\filldraw[black] (-6,7) circle (2pt);
	\filldraw[black] (-6,8) circle (2pt);
	\filldraw[black] (-6,9) circle (2pt);

	\filldraw[black] (-2,0) circle (2pt);
	\filldraw[black] (-2,1) circle (2pt);
	\filldraw[black] (-2,2) circle (2pt);
	\filldraw[black] (-2,3) circle (2pt);
	\filldraw[black] (-2,4) circle (2pt);
	\filldraw[black] (-2,5) circle (2pt);
	\filldraw[black] (-2,6) circle (2pt);
	\filldraw[black] (-2,7) circle (2pt);
	\filldraw[black] (-2,8) circle (2pt);
	\filldraw[black] (-2,9) circle (2pt);
	
	\draw[black] (1,4) -- (2,4);
	\node[] at (3,4) {\Large $\in S$};

	
	\draw[black,rounded corners=2pt] (9.5,-0.5) rectangle (10.5,4.4);
	\draw[black,rounded corners=2pt] (9.5,4.6) rectangle (10.5,9.5);
	\draw[black,rounded corners=2pt] (13.5,-0.5) rectangle (14.5,4.4);
	\draw[black,rounded corners=2pt] (13.5,4.6) rectangle (14.5,9.5);

	\node[] at (9,9) {\Large $a$};
	\node[] at (9,0) {\Large $b$};
	\node[] at (9,4) {\Large $r$};
	\node[] at (8.3,5) {\Large $r-1$};
	\node[] at (15,9) {\Large $a'$};
	\node[] at (15,0) {\Large $b'$};
	\node[] at (15,4) {\Large $r'$};
	\node[] at (15.6,5) {\Large $r'-1$};
	
	\draw[red] (10,9) -- (14,9);
	\draw[red] (10,7) -- (14,8);
	\draw[red] (10,6) -- (14,6);
	
	\draw[cyan] (10,8) -- (14,4);
	\draw[cyan] (10,5) -- (14,2);
	\draw[cyan] (10,5) -- (14,0);
	
	\draw[green!80!black] (10,4) -- (14,7);
	\draw[green!80!black] (10,0) -- (14,5);
	\draw[green!80!black] (10,3) -- (14,6);
	
	\draw[violet] (10,3) -- (14,3);
	\draw[violet] (10,2) -- (14,1);
	\draw[violet] (10,1) -- (14,0);
	
	\filldraw[black] (10,0) circle (2pt);
	\filldraw[black] (10,1) circle (2pt);
	\filldraw[black] (10,2) circle (2pt);
	\filldraw[black] (10,3) circle (2pt);
	\filldraw[black] (10,4) circle (2pt);
	\filldraw[black] (10,5) circle (2pt);
	\filldraw[black] (10,6) circle (2pt);
	\filldraw[black] (10,7) circle (2pt);
	\filldraw[black] (10,8) circle (2pt);
	\filldraw[black] (10,9) circle (2pt);

	\filldraw[black] (14,0) circle (2pt);
	\filldraw[black] (14,1) circle (2pt);
	\filldraw[black] (14,2) circle (2pt);
	\filldraw[black] (14,3) circle (2pt);
	\filldraw[black] (14,4) circle (2pt);
	\filldraw[black] (14,5) circle (2pt);
	\filldraw[black] (14,6) circle (2pt);
	\filldraw[black] (14,7) circle (2pt);
	\filldraw[black] (14,8) circle (2pt);
	\filldraw[black] (14,9) circle (2pt);
	
	\draw[red] (18,6) -- (19,6);
	\node[] at (20,6) {\Large $\in S_1$};
	\draw[cyan] (18,5) -- (19,5);
	\node[] at (20,5) {\Large $\in S_2$};	
	\draw[green!80!black] (18,4) -- (19,4);
	\node[] at (20,4) {\Large $\in S_3$};
	\draw[violet] (18,3) -- (19,3);
	\node[] at (20,3) {\Large $\in S_4$};
\end{tikzpicture}}
	\caption{Example of splitting a partial solution $S \in E(G\angles{a,b,a',b'}$ into four partial solutions $S_1,S_2,S_3,S_4$, where $S_1 \subseteq E(G\langle a,r-1,a',r'-1\rangle)$, $S_2 \subseteq E(G\langle a,r-1,r',b'\rangle)$, $S_3 \subseteq E(G \langle r,b,a',r'-1\rangle)$ and $S_4 \subseteq E(G\langle r,b,r',b'\rangle)$ with $r=\lceil (a+b)/2\rceil$ and $r'=\lceil (a'+b')/2\rceil$.}
	\label{fig:intervals}
\end{figure}

\begin{claim} \label{claim:1A}
	The following event happens with probability at least $\left(1-  \frac{1}{n^2} \right)^{i+1}$: for all different $S,S' \in Y$, it holds that $\omega_i(S) \neq \omega_i(S')$.		
\end{claim}
\begin{proof}
	For each $S \in Y$, let $$x_S \coloneqq \sum_{e \in S} 2^{\id(e)}.$$ 
	Observe that since identifiers assigned to the edges are unique, the numbers $x_S$ are pairwise different. The induction hypothesis gives us that the following event $A_{i-1}$ happens with probability at least $\left(1-\frac{1}{n^{2}}\right)^i$: for all $1\leq s\leq t\leq n$ and $1\leq s'\leq t'\leq n'$ with $t-s\le 2 ^{i-1}$ and $t'-s'\le 2^ {i-1}$, and all $c\in \conf(V\angles{s,t,s',t'})$, we have $|\Min(\omega_{i-1},G\langle s,t,s',t' \rangle,c)| \le 1$. Assuming now that $A_{i-1}$ indeed happens, by Theorem~\ref{thm:rankbased} we conclude that for every fixed choice of $s,t,s',t'$ as above, we have
	$$\left|\bigcup_{c \in \conf(V\angles{s,t,s',t'})} \Min({\omega_{i-1}},{G\langle s,t,s',t' \rangle},c)\right|\le 2^{\Oh(2^{i-1})}.$$
	Since there are at most $n^4$ choices of $s,t,s',t'$, this implies that
	$$|Y|\leq |Y'|^4\leq 2^{\Oh(2^{i-1})}\cdot n^{16}.$$
	Since $M_{i} = 2^{C(\log n + 2^i)}$ and $p_{i}$ is drawn uniformly at random
    among the primes in the range $\{1,\dots,M_i\}$, from Lemma~\ref{FKS} we can conclude that, for large enough $C$, with probability at least 
	\begin{align*}
		\left(1- \frac{(n^2+1)\left(n^{16}2^{\Oh(2^{i-1})}\right)^2}{2^{(C/2)(\log n + 2^i)}} \right) \cdot \left(1 - \frac{1}{n^{2}}\right)^i \ge \left(1 - \frac{1}{n^{2}}\right)^{i+1},
	\end{align*}
	all the numbers $\{x_S\colon S \in Y\}$ have pairwise different remainders modulo $p_i$; here, the term $(1-\frac{1}{n^2})^i$ corresponds to the probability that $A_i$ happens. As a consequence, with the same probability we have that $\omega_{i}(S) \neq \omega_{i}(S')$ for all different $S,S'\in Y$. 
\end{proof}
Now the induction step follows directly from combining Claim~\ref{claim:1B} with Claim~\ref{claim:1A}.

\subsection{Hamiltonian cycles in graphs of bounded treewidth}\label{sec:ham-tw}

We will now use the same approach to give a proof of Theorem~\ref{thm:hc-tree}. More precisely, assume we are given a graph $G$ of treewidth at most $k$. Our goal is to isolate the family of Hamiltonian cycles in $G$ using $\Oh(k \log n + \log^2 n)$ random bits.

The proof follows the same structure as that of Theorem~\ref{thm:hc-gen}.
We first describe the isolation scheme and then analyze the scheme using a tree decomposition $\treedecomp = (T,\beta)$ of $G$ of width at most $k$. Note that the actual decomposition is not needed for the isolation procedure, and is only used as a tool in the analysis.

\paragraph{Isolation scheme.}
We first present the isolation scheme. As before, we assume that $V(G)=[n]$ and $n$ is a power of $2$. Let $\id\colon E(G) \to \{1,\dots,|E(G)|\}$ be any bijection that assigns to each edge $e\in E(G)$ its unique \emph{identifier} $\id(e)$. Let $C$ be some large enough constant, to be chosen later. Then we independently sample $3\log n$ primes 
$p_1,\dots,p_{3\log n}$ so that each $p_i$ is sampled uniformly among all primes in the interval $\{1,\dots,M\}$, where $M=2^{C(k\log n)}$. Note that choosing each $p_i$ requires $C(k + \log n)$ random bits, hence we have used $\Oh(k \log n + \log^2 n)$ random bits in total, as required.

Next, we inductively define weights functions $\omega_0,\dots,\omega_{3\log n}$ on $E(G)$ as follows:
\begin{itemize}
	\item Set $\omega_0(e) \coloneqq 0$ for all $e \in E(G)$.
	\item For each $ e \in E(G)$ and $i=1,\dots,3\log n$, set
	$$\omega_i(e) \coloneqq Mn \cdot \omega_{i-1}(e) + \left( 2^{\id(e)}\bmod p_i\right).$$
\end{itemize} 
We let $\omega \coloneqq \omega_{3\log n}$ and we observe that $\omega$ assigns weights bounded by $2^{\Oh(k\log n + \log^2 n)}$.

\paragraph{Analysis.} Let $\treedecomp = (T,\beta)$ be a tree decomposition of $G$ of width at most $k$. It is well-known that $T$ can be chosen so that it has at most $n$ nodes. Further, let $\eta\coloneqq E(G)\to V(T)$ be any function that assigns to each edge $e$ of $G$ any node $x$ of $T$ such that $e\subseteq \beta(x)$. In the sequel we will assume that $\eta$ is injective. This can be achieved by adding, for each node $x\in V(T)$, $|\eta^{-1}(x)|-1$ new nodes with the same bag and adjacent only to $x$, and appropriately distributing the images of edges of $\eta^{-1}(x)$ among the new nodes. Note that after this modification, the number of nodes of $T$ is bounded by $\binom{k+1}{2}\cdot n\leq n^3$.

Compared to the proof of Theorem~\ref{thm:hc-gen}, instead of intervals we will use {\em{segments}} in the tree $T$ underlying the tree decomposition $\treedecomp$. Recall that segments have been defined and discussed in Section~\ref{sec:prelim}. We first observe that there are only few segments.

\begin{claim} \label{Claim:fewIntervals}
There are at most $n^9$ segments of $T$.
\end{claim}
\begin{proof}
 Note that a segment $I$ in $T$ can be uniquely determined by specifying the at most two vertices of $\bnd I$ and any vertex of $V(I)\setminus \bnd I$, provided there exists any. Since $T$ has at most $n^3$ nodes, there are at most $n^9$ choices for such a specification.
\end{proof}

For a set of nodes $Z\subseteq V(T)$, we write $\beta(Z)\coloneqq \bigcup_{z\in Z}\beta(z)$. Further, for a segment $I$ of $T$ we consider the graph
$$G\angles{I}\coloneqq \left(\beta(V(I)),\eta^{-1}(V(I))\right).$$
Usually when speaking about partial solutions in $G\angles{I}$, we consider their configurations on the vertex subset $\beta(\bnd I)$. Note that $G\angles{T}=G$.

We proceed to the induction. We will prove the following statement for all $0\le i \le \log n$.

\IH{With probability at least $\left( 1 - \frac{1}{n^{2}} \right)^i$, for all segments $I$ of $T$ of size at most $2^i$ and for each configuration $c \in \conf(\beta(\bnd I))$, there is at most one minimum weight (w.r.t. $\omega_i$) compliant partial solution in $G\angles{I}$, i.e. $|\Min({\omega_i},{G\angles{I}},c)|\le 1$.}\\

	Note that since $|V(T)|\leq n^3$, for $i = 3\log n$ the induction hypothesis gives that for $G\angles{T}=G$, there is at most one Hamiltonian cycle that has the minimum weight w.r.t. $\omega$ with probability at least $\left( 1-\frac{1}{n^{2}}  \right)^{3\log n} \ge \left( 1-\frac{1}{n} \right)$.
	
    \paragraph{Base step.} For $i=0$, we take segments of size at most $1$, i.e. we prove the induction hypothesis for every segment $I$ of $T$ that has either one or two nodes. More precisely, we have to prove that (with suitably large probability), for every such segment $I$ and configuration $c \in \conf(\beta(\bnd I))$, we have $|\Min(\omega_0,G\angles{I},c)|\leq 1$. Note that since $I$ has at most two nodes and $\eta$ is injective, the edge set $E(G\angles{I})$ consists of at most two edges. Moreover, it cannot be that two different edge subsets $E_1,E_2\subseteq E(G\angles{I})$ are simultaneously compliant with the same configuration $c\in \conf(\beta(\bnd I))$. It follows that sets $\Min(\omega_0,G\angles{I},c)$ have sizes at most $1$ always, so the induction hypothesis for $i=0$ is true.
	
    \paragraph{Induction step.}
	Assume the induction hypothesis is true for all segments of size at most $2^{i-1}$. Let
	$$
	Y'\coloneqq \bigcup_{I\colon \textrm{segment of size }\leq 2^{i-1}} \ \bigcup_{c\in \conf(\beta(\bnd I))}\  \Min({\omega_{i-1}},{G\angles{I}},c).$$
	 be the set of all minimum weight partial solutions for segments of size at most $2^{i-1}$. Further, let $$Y \coloneqq \{\, S_1 \cup S_2 \cup S_3 \cup S_4 \cup S_5\ \colon\ S_1, S_2, S_3, S_4, S_5 \in Y'\,\}$$ 
	 be the set comprising all combinations of five such partial solutions. 
	
	We first prove with Claim~\ref{claim:2A} that every relevant minimum weight compliant edge is contained in $Y$. Then Claim~\ref{claim:2B} says that with high probability, all $S\in Y$ receive pairwise different weights with respect to $\omega_i$. The induction hypothesis will follow directly from combining these two claims.
	
	\begin{claim}
		\label{claim:2A}
	Let $I$ be any segment of size at most $2^{i}$ and let $c \in \conf(\beta(\bnd I))$. 
	Then $\Min({\omega_{i}},{G\angles{I}},c)\subseteq Y$. 	
	\end{claim}
	\begin{proof}
            Consider any $S\in \Min({\omega_{i}},{G\angles{I}},c)$.
			By Lemma~\ref{lemma:5SubTrees}, there exist segments $I_1,\ldots,I_t$ ($t\leq 5$), each of size at most $2^{i-1}$, such that $E(I)$ is the disjoint union of $E(I_1),\ldots,E(I_t)$. For each $j\in \{1,\ldots,t\}$ choose $S_j\in E(G\angles{I_j})$ so that $S$ is the disjoint union of $S_1,\ldots,S_t$. The same argument as that was used in the proof of Claim~\ref{claim:1B} shows that there exists $c_j\in \conf(\beta(\bnd I_j))$ such that $S_j\in  \Min(\omega_{i-1}, G\angles{I_j}, c_j)$. Hence $S_j \in Y'$ for all $j\in \{1,\dots,t\}$, so it follows that $S \in Y$.
	\end{proof}

	\begin{claim}
		\label{claim:2B}
		The probability of the following event is at least $\left( 1 - \frac{1}{n^2}\right)^i$: for all different $S,S' \in Y$, it holds that $\omega_i(S) \neq \omega_i(S')$.
	\end{claim}
\begin{proof}
	For each $S \in Y$ let us define
	$$x_S = \sum_{e\in S}2^{\id(e)}.$$
	Observe that since the identifiers assigned to the edges are unique, the numbers $x_S$ are pairwise different. By the induction hypothesis, the following event $A_{i-1}$ happens with probability at least $\left( 1 - \frac{1}{n^{2}}\right)^{i-1}$: for every segment $I$ of size at most $2^{i-1}$ and each configuration $c\in \conf(\beta(\bnd I))$, we have $|\Min({\omega_{i-1}},{G\angles{I}},c)| \le 1$. By Theorem~\ref{thm:rankbased} it follows that provided $A_{i-1}$ happens, for every fixed segment $I$ of size at most $2^{i-1}$ we have
	$$\left|\bigcup_{c\in \conf(\beta(\bnd I))} \Min({\omega_{i-1}},{G\angles{I}},c)\right| \le 2^{\Oh(|\beta(\bnd I)|)}\leq 2^{\Oh(k)}.$$
	By Claim~\ref{Claim:fewIntervals} there are at most $n^9$ different segments, hence this implies that
	$$ |Y| \le |Y'|^5 \le 2^{\Oh(k)}\cdot n^{45}.$$
	Recall now that $M = 2^{C(k+ \log n)}$ and $p_i$ is drawn uniformly at
    random among the primes in the range $\{1,\dots,M\}$. Hence, from Lemma~\ref{FKS} we can conclude that, for large enough $C$, with probability at least  
	\begin{align*}
		\left( 1 - \frac{(n^2+1)\left(2^{\Oh(k)}\cdot n^{45}\right)^2}{2^{(C/2)(k+\log n)}} \right)\cdot \left( 1 - \frac{1}{n^{2}} \right)^{i-1} \ge \left( 1 - \frac{1}{n^{2}} \right)^{i},
	\end{align*}
	all the numbers in $\{x_S\colon S \in Y\}$ have pairwise different remainders modulo $p_i$. Here, the factor $(1-\frac{1}{n^2})^{i-1}$ corresponds to the probability that $A_{i-1}$ happens. As a consequence, with the same probability for all different $S, S' \in Y$ we have $\omega_i(S) \neq \omega_i(S')$. 
\end{proof}

The induction step now follows directly from combining Claims~\ref{claim:2A} and~\ref{claim:2B}.

\subsection{Separable graph classes}\label{sec:HCdecomp}

In this section we use our understanding of isolation schemes for Hamiltonian cycles in decomposable graphs to design such isolation schemes for {\em{separable}} graph classes, that is, classes of graphs that admit small balanced separators. More precisely, we will prove a generalization of Theorem~\ref{thm:planar}. First, we need to establish certain terminology and decomposition results.

\subsubsection{Definitions and a decomposition theorem}

A {\em{graph class}} is a (possibly infinite) set of graphs that is closed under taking isomorphisms. A graph class is {\em{hereditary}} if it is closed under taking induced subgraphs. The following notion of separability expresses the condition that graphs from a given class can be broken in a balanced way using small  separators.

\begin{definition}\label{def:separable}
 A graph class $\Cc$ is {\em{separable}} with {\em{degree}} $\alpha\in (0,1)$ if for every graph $G\in \Cc$, say on $n$ vertices, and vertex subset $S\subseteq V(G)$, there exists a set $X\subseteq V(G)$ such that $|X|\leq \Oh(n^{\alpha})$ and every connected component of $G-X$ contains at most $|S|/2$ vertices of $S$. Class $\Cc$ is {\em{separable}} if it is separable with some degree $\alpha\in (0,1)$.
\end{definition}

It is well-known that planar graphs~\cite{LiptonT80} and, more generally, $H$-minor-free graphs~\cite{AlonST90} for every fixed $H$ are separable with degree $\frac{1}{2}$. However, this notion is more general. For instance, the class of {\em{$1$-planar graphs}} --- graphs that admit a planar embedding where every edge has at most one crossing --- is also separable with degree $\frac{1}{2}$~\cite{DujmovicEW17}. More generally, every graph class of {\em{polynomial expansion}} is separable~\cite{PlotkinRS94} (see also the discussion in~\cite[Sections~16.3 and~16.4]{sparsity}). Examples here would include intersection graphs of bounded-ply families of fat objects in Euclidean spaces of fixed dimension~\cite{Har-PeledQ17}. In fact, subject to technical details, the notions of polynomial expansion and of separability coincide~\cite{PlotkinRS94,EsperetR18}.

Our isolation schemes will work on any graph class that is hereditary and separable. That is, we will prove the following generalization of Theorem~\ref{thm:planar}.

\begin{theorem}\label{thm:separable-iso}
 Let $\Cc$ be a hereditary class of graphs that is separable with degree $\alpha\in (0,1)$. Then there is an isolation scheme for Hamiltonian cycles in graphs from $\Cc$ that uses $\Oh(n^{\alpha})$ random bits and assigns weights upper bounded by $2^{\Oh(n^{\alpha})}$.
\end{theorem}

An important ingredient in the proof of Theorem~\ref{thm:separable-iso} is a decomposition theorem for graphs from a fixed separable class. Intuitively, the decomposition is obtained by recursively breaking the graphs by extracting small balanced separators. The shape of the decomposition will be captured by the following generalization of the notion of an elimination forest.

\begin{definition}
 A {\em{generalized elimination forest}} of a graph $G$ is a rooted forest $F$ together with a mapping $\eta\colon V(G)\to V(F)$ satisfying the following property: for every edge $uv\in E(G)$, we have $\eta(u)\preceq \eta(v)$ or $\eta(u)\succeq \eta(v)$ (note that it is possible that $\eta(u)=\eta(v)$). The {\em{topological height}} of $(F,\eta)$ is simply the height of $F$, while the {\em{height}} of $(F,\eta)$ is equal to
 $$\max_{x\in V(F)}\  \sum_{y\preceq x}\, |\eta^{-1}(y)|.$$
\end{definition}

It is easy to see that if a graph $G$ admits a generalized elimination forest $(F,\eta)$ of height $d$, then it also admits an elimination forest of height $d$: for every node $x$ of $F$, replace $x$ with a path consisting of vertices of $\eta^{-1}(x)$, in any order. Thus, the intuition is that a generalized elimination forest is a compressed representation of an elimination forest, where some sets of interchangable vertices --- the preimages $\eta^{-1}(x)$ for $x\in V(F)$ --- are grouped together in single nodes. The quality of this compression is measured by the parameter topological height.

In the sequel, we will use the following decomposition theorem that ties together separable graph classes and generalized elimination forests. We are not aware of the existence of this particular formulation in the literature, however the proof relies on rather standard techniques.

\begin{theorem} \label{thm:separable}
 Let $\Cc$ be a graph class that is hereditary and separable with degree $\alpha\in (0,1)$. Then every graph $G\in \Cc$, say on $n$ vertices, admits a generalized elimination forest $(F,\eta)$ satisfying the following conditions:
 \begin{enumerate}[label=(C\arabic*),ref=(C\arabic*),leftmargin=*]
  \item\label{c:children} $F$ has one root and every node of $F$ has at most seven children.
  \item\label{c:height} $(F,\eta)$ has topological height at most $1+\log_2 n$ and height $\Oh(n^{\alpha})$.
  \item\label{c:size} For every node $x$ of $F$, if $i$ is the depth of $x$ in $F$, then
  \begin{itemize}
  \item $|\eta^{-1}(x)|\leq \Oh\left(\left(n/2^i\right)^{\alpha}\right)$,\item $|\eta^{-1}(\subtree[x])|\leq n/2^i$, and
  \item $|N_G(\eta^{-1}(\subtree[x]))|\leq \Oh\left(\left(n/2^i\right)^{\alpha}\right)$.
  \end{itemize}
 \end{enumerate}
\end{theorem}
\begin{proof}
 Let $K$ be the constant hidden in the $\Oh(\cdot)$ notation in the bound on the sizes of balanced separators in graphs from $\Cc$, as prescribed by the definition of the separability of $\Cc$. We will use the following simple claim.
 
 \begin{claim}\label{cl:partition}
  Suppose $\Omega$ is a finite set and there are two weight functions $\omega_1,\omega_2\colon \Omega\to \mathbb{R}_{\geq 0}$ satisfying the following conditions:
  \begin{itemize}[nosep]
   \item $\omega_1(\Omega)\leq 1$ and $\omega_2(\Omega)\leq 1$; and
   \item for each $e\in \Omega$, $\omega_1(e)\leq 1/2$ and $\omega_2(e)\leq 1/2$.
  \end{itemize}
  Then there exists a partition ${\cal P}$ of $\Omega$ into at most seven parts such that for each $P\in {\cal P}$, we have $\omega_1(P)\leq 1/2$ and $\omega_2(P)\leq 1/2$.
 \end{claim}
 \begin{proof}
 Let ${\cal P}_0$ be the partition of $\Omega$ that puts every element of $\Omega$ into a separate part. By the second assumed condition, ${\cal P}_0$ respects the following assertion ($\star$): for each part $P$, we have $\omega_1(P)\leq 1/2$ and $\omega_2(P)\leq 1/2$. We will gradually transform ${\cal P}_0$ into a partition ${\cal P}$ consisting of at most seven parts while preserving assertion ($\star$).
 
 More precisely, starting with ${\cal P}_0$ we define a sequence of partitions ${\cal P}_0,{\cal P}_1,{\cal P}_2,\ldots$ that finishes at the first partition ${\cal P}_i$ that has at most seven parts; then we set ${\cal P}\coloneqq {\cal P}_i$. Each partition ${\cal P}_{i+1}$ is obtained from ${\cal P}_i$ by merging two parts as follows. Note that ${\cal P}_i$ has at least eight parts, for otherwise the construction should have already finished. Call a part $P$ of ${\cal P}_i$ {\em{bad}} if $\omega_1(P)>1/4$ or $\omega_2(P)>1/4$. Clearly, there can be at most six bad parts (at most three due to the first reason, and at most three due to the second reason), which leaves us with at least two parts that are not bad. Then construct ${\cal P}_{i+1}$ from ${\cal P}_i$ by merging any two not bad parts. It is clear that in this way, assertion ($\star$) is preserved during the construction and we are done.
 \end{proof}
 
 We proceed to the construction of the generalized elimination forest $(F,\eta)$, which will be done by means of a recursive procedure. For a nonempty subset of vertices $A\subseteq V(G)$, the procedure constructs a generalized elimination forest $(F_A,\eta_A)$ of $H\coloneqq G[A]$ as follows.
 \begin{itemize}
  \item Let $H'\coloneqq G[N_G[A]]$. Note that since $\Cc$ is hereditary, we have $H'\in \Cc$.
  \item By the separability of $\Cc$, there are vertex subsets $X,Y\subseteq V(H')$, each of size at most $K\cdot |V(H')|^\alpha$, such that
  \begin{itemize}
  \item every connected component of $H'-X$ contains at most $|A|/2$ vertices of $A$; and
  \item every connected component of $H'-Y$ contains at most $|S|/2$ vertices of $S$, where $S\coloneqq N_G(A)$.
  \end{itemize}
  \item Let $Z\coloneqq X\cup Y$. For every connected component $C$ of $H'-Z$, let
  $$\omega_1(C)\coloneqq \frac{|V(C)\cap A|}{|A|}\qquad\textrm{and}\qquad \omega_2(C)\coloneqq \frac{|V(C)\cap S|}{|S|}.$$
  In case $S=\emptyset$, we set $\omega_2(C)\coloneqq 0$.
  \item Noting that the set of connected components of $H'-Z$ with weight functions $\omega_1$ and $\omega_2$ satisfy the prerequisites of Claim~\ref{cl:partition}, we can group the connected components of $H'-Z$ into at most seven graphs, say with vertex sets $B_1,\ldots,B_t$ ($t\leq 7$), such that
  \begin{equation}\label{eq:halving}
  |B_i\cap A|\leq |A|/2 \qquad\textrm{and}\qquad |B_i\cap S|\leq |S|/2\qquad\textrm{for each }i\in [t].\end{equation}
  \item Recursively apply the procedure to the sets $B_1\cap A,\ldots,B_t\cap A$, thus obtaining generalized elimination forests $(F_{1},\eta_{1}),\ldots,(F_{7},\eta_{7})$ of graphs $G[B_1\cap A],\ldots,G[B_t\cap A]$, respectively.
  \item Construct the generalized elimination forest $(F_A,\eta_A)$ of $H$ by taking the union of $(F_{j},\eta_{j})$ for $j\in [t]$, adding a single root node $r$ with $\eta_A^{-1}(r)=Z\cap A$, and making the roots of forests $F_j$ into children of~$r$.
 \end{itemize}
 It is clear that $(F_A,\eta_A)$ constructed in this manner is a generalized elimination forest of $G[A]$. 
 We construct the generalized elimination forest $(F,\eta)$ of $G$ by applying the procedure to $A=V(G)$. 
 Condition~\ref{c:children} is clear from the construction, hence we need to verify conditions~\ref{c:height} and~\ref{c:size}.
 
 We start with condition~\ref{c:size}. Observe that it suffices to prove that for every recursive call of the construction procedure, say at recursion depth $i$, it holds that
 $$|Z|\leq \Oh\left(\left(n/2^i\right)^{\alpha}\right),\qquad |A|\leq n/2^i\qquad\textrm{and}\qquad |N_G(A)|\leq \Oh\left(\left(n/2^i\right)^{\alpha}\right),$$
 where $Z$ is as defined in the construction procedure.
 The second bound follows from a straightforward induction on $i$ using the first part of~\eqref{eq:halving}. For the first and third bound, we shall prove by induction on $i$ that
 \begin{equation}\label{eq:ind-hyp}
 |N_G(A)|\leq L\cdot \left(n/2^i\right)^\alpha,  
 \end{equation}
 where
 $$L\coloneqq \max\left(1,\left(\frac{8K}{1-\frac{1}{2^{1-\alpha}}}\right)^{\frac{1}{1-\alpha}}\right).$$
 Note that in the base step, for $i=0$, we have $N_G(A)=N_G(V(G))=\emptyset$. During the proof of~\eqref{eq:ind-hyp} we will argue that in the considered recursive call of the construction procedure, it holds that
 \begin{equation}\label{eq:Z-bound}
|Z|\leq 4KL^\alpha\cdot \left(n/2^i\right)^\alpha.
 \end{equation}
 Thus,~\eqref{eq:ind-hyp} and~\eqref{eq:Z-bound} imply the first and the third bound of condition~\ref{c:size}.

 Assume then that~\eqref{eq:ind-hyp} holds for the call on a vertex subset $A$. We need to prove that, assuming the notation from the description of the procedure, for each subsequent call on a subset $B_j\cap A$, $j\in [t]$, we have $|N(B_j\cap A)|\leq L\cdot \left(n/2^{i+1}\right)^\alpha$. Observe that
 $$|V(H')|=|A|+|N_G(A)|\leq n/2^i+L\cdot \left(n/2^i\right)^\alpha\leq 2L\cdot \left(n/2^i\right),$$
 hence
 $$|X|\leq K\cdot |V(H')|^\alpha\leq 2KL^\alpha\cdot \left(n/2^i\right)^\alpha,$$
 and similarly
 $$|Y|\leq 2KL^\alpha\cdot \left(n/2^i\right)^\alpha.$$
 Therefore,
 $$|Z|=|X\cup Y|\leq 4KL^\alpha \cdot \left(n/2^i\right)^\alpha,$$
 which in particular proves~\eqref{eq:Z-bound}.
 Now observe that for each $j\in [t]$, we have
 $$N_G(B_j\cap A)\subseteq Z\cup (B_j\cap S).$$
 Hence, using the second part of~\eqref{eq:halving}, we conclude that
 \begin{eqnarray*}
 |N_G(B_j\cap A)| & \leq & |Z|+|B_j\cap S| \\
                & \leq & |Z|+|S|/2 \\
                & =    & |Z|+|N_G(A)|/2 \\
                & \leq & 4KL^\alpha \cdot \left(n/2^i\right)^\alpha + L/2\cdot \left(n/2^i\right)^\alpha \\
                & =    & (4KL^\alpha+L/2)\cdot \left(n/2^i\right)^\alpha \\
                & \leq & (8KL^\alpha+2^{\alpha-1}L)\cdot \left(n/2^{i+1}\right)^\alpha \\
                & \leq & L\cdot \left(n/2^{i+1}\right)^\alpha.
 \end{eqnarray*}
 Here, the last inequality follows from the choice of $L$. This concludes the inductive proof of~\eqref{eq:ind-hyp} and finishes the proof of condition~\ref{c:size}.
 
 We are left with showing condition~\ref{c:height}. The first assertion --- that the topological height of $(F,\eta)$ is bounded by $1+\log_2 n$ --- follows immediately from the first assertion of condition~\ref{c:size}. For the second assertion, observe that in the proof of condition~\ref{c:size} we argued that $|Z|\leq 4KL^\alpha \cdot \left(n/2^i\right)^\alpha$ for calls at recursion depth $i$. This implies that whenever $x$ is a node of $F$ at depth $i$, we have $|\eta^{-1}(x)|\leq 4KL^\alpha \cdot \left(n/2^i\right)^\alpha$. Therefore, the height of $(F,\eta)$ is bounded by
 $$\sum_{i=0}^{\lfloor \log_2 n\rfloor} 4KL^\alpha \cdot \left(n/2^i\right)^\alpha\leq 4KL^\alpha \cdot n^\alpha \cdot \sum_{i=0}^\infty \frac{1}{(2^\alpha)^i}\leq \Oh(n^\alpha).$$
 Here, the last inequality is implied by the convergence of the geometric series in question.
\end{proof}

\subsubsection{Isolation scheme}

With Theorem~\ref{thm:separable} established, we can proceed to the proof of Theorem~\ref{thm:separable-iso}. Let us then fix a hereditary graph class $\Cc$ that is separable with degree $\alpha\in (0,1)$, and an $n$-vertex graph $G\in \Cc$.
We first present the isolation scheme. Let $\id\colon E(G) \to \{1,\ldots,|E(G)|\}$ be any bijection that assigns to each edge $e\in E(G)$ its unique \emph{identifier} $\id(e)$. 
Let $C$ be some large enough constant, to be determined later. We independently select $1 + \log n^\alpha+C$ primes 
$$ p_0, p_1,\ldots,p_{\log n^\alpha+C}$$
so that each $p_i$ is sampled uniformly among primes in the range $\{1,\ldots,M_i\}$, where $M_i\coloneqq 2^{C(\log n^\alpha +2^i)}$. Further, we independently sample $1+\log n$ primes
$$ q_0, q_1,\ldots,q_{\log n}$$
so that each $q_i$ is sampled uniformly among primes in the range $\{1,\ldots,N_i\}$, where $N_i\coloneqq 2^{C(\log n + (n/2^i)^\alpha)}$. Note that choosing each $p_i$ requires $C(\log n^\alpha + 2^i)$ random bits and choosing each $q_i$ requires $C(\log n + (n/2^i)^\alpha)$ random bits. Hence we have used $\Oh(n^\alpha)$ random bits in total, as required. 

Next, we inductively define weight functions $\omega_0,\dots,\omega_{\log n^\alpha+C}$ and $\xi_0,\dots,\xi_{\log n}$ on $E(G)$ as follows:

\begin{itemize}
	\item Set $\omega_0(e) = 2^{\id(e)} \bmod p_0$ for all $e \in E(G)$.
	\item For $e \in E(G)$ and $i = 1,\dots,\log n^\alpha+C$, set 
	$$\omega_i (e) = M_{i-1}n \cdot \omega_{i-1}(e) + \left(2^{\id(e)} \bmod p_i \right).$$
	\item Then, let for all $e \in E(G)$, set: 
	$$\xi_{\log n}(e) = M_{\log n^\alpha + C} n \cdot \omega_{\log n^\alpha+C}(e) + \left(2^{\id(e)} \bmod q_{\log n} \right)$$ 
	\item Finally, for all $e \in E(G)$ and $i = \log n -1,\dots,0$, set
	$$\xi_i (e) = N_{i+1}n \cdot \xi_{i+1}(e) + \left(2^{\id(e)} \bmod q_i \right).$$
\end{itemize}
Note that the weight functions $\xi_i$ depend on the weight functions $\omega_i$. Furthermore, the order of defining the weight functions $\xi_i$ might seem inverse to what one might expect. Intuitively, this corresponds to proving a suitable isolation property by a bottom-up induction on the generalized elimination forest provided by Theorem~\ref{thm:separable}. Finally, we define $\omega \coloneqq \omega_{\log n^\alpha+C}$ and $\xi \coloneqq \xi_0$.  Then we return $\xi$ as the output weight function. Note that $\xi$ assigns weights upper bounded by $2^{\Oh(n^\alpha)}$, as promised. Hence, it remains to prove that $\xi$ isolates the family of Hamiltonian cycles in $G$ with high probability.

The proof of isolation is done in two steps. First we show that the weight function $\omega = \omega_{\log n^\alpha+C}$ isolates (with high probability) partial solutions on all graphs that intuitively correspond to single nodes of the generalized elimination forest $F$ provided by Theorem~\ref{thm:separable}. The second step is to use this knowledge to perform a bottom up induction on $F$, using weight functions $\xi_{\log n},\ldots,\xi_0$ for the consecutive steps.

\subsubsection{Isolation of partial solutions in single nodes}

Let $(F,\eta)$ be a generalized elimination forest of $G$ provided by Theorem~\ref{thm:separable}. We may assume that $G$ is connected (as otherwise there are no Hamiltonian cycles in $G$), hence $F$ is a tree. We may assume that $\eta^{-1}(x)\neq \emptyset$ for every leaf $x$ of $F$ (otherwise $x$ can be disposed of), hence $F$ has at most $(1+\log_2 n)\cdot n\leq n^2$ nodes.

We first show that the weight function $\omega$, which uses the prime numbers $p_0,\ldots,p_{\log n^\alpha+C}$, is enough to isolate all relevant partial in graphs $H_x$ for $x\in V(F)$ defined as follows:
\begin{definition}[Graph $H_x$]
	For each node $x$ in $F$, define $H_x\coloneqq (V_x,E_x)$, where
	$$V_x\coloneqq N_G[\eta^{-1}(x)] \cap \eta^{-1}(\tail[x])\qquad\textrm{and}\qquad E_x\coloneqq \{uv\colon uv\in E(G), u \in \eta^{-1}(x),\textrm{ and } v\in V_x\}.$$
\end{definition}  

In other words, $H_x$ is the subgraph of $G$ whose vertex set consists of all the vertices of $\eta^{-1}(x)$ and their neighbors that are mapped to a node of $\tail[x]$ by $\eta$. Among the edges with endpoints in this vertex set we keep only those whose at least one endpoint belongs to  $\eta^{-1}(x)$.


The following statement is a generalization of Theorem~\ref{thm:hc-gen}, where the identifiers come from a larger codomain and we assert a stronger isolation property. The proof follows from a straightforward adjustment of the proof of Theorem~\ref{thm:hc-gen}, so we only sketch it.

\begin{theorem} \label{thm:H_x}
	Let $G$ be a graph with $n$ vertices and let $\id\colon E(G) \to \{1,\ldots,Z\}$ be an injective function. Choose prime numbers $p_0,\ldots,p_{\log n}$ independently at random so that $p_i$ is chosen uniformly among the primes in the range $\{1,\dots,M_i\}$, where $M_i = 2^{C(\log Z + \log n+ 2^i)}$ for some large constant $C$. Then with probability at least $1- \frac{1} {Z\cdot \rho(n)}$, for all configurations $c \in \conf(V(G))$ we have $|\Min(\omega,G,c)|\leq 1$, where $\omega$ is defined as in Subsection~\ref{sec:HCgeneral} and $\rho(n)$ is any fixed polynomial. 
\end{theorem}
\begin{proof}
	The proof follows the exact same reasoning as the proof of Theorem~\ref{thm:hc-gen}. The induction hypothesis then becomes:
	 
	\IH{With probability at least $\left(1 - \frac{1}{Z \cdot n\rho(n)} \right)^{i+1}$, for all intervals $G \langle s,t,s',t' \rangle $ such that $t-s \le 2^i$ and $t'-s' \le 2^i$ and for each configuration $c \in \conf(V\angles{s,t,s',t'})$, there is at most one minimum weight (with respect to $\omega_{i}$) compliant partial solution, i.e. $|\Min({\omega_{i}},{G\langle s,t,s',t' \rangle},c)|\le 1$.} \smallskip
	
	The only major change is that after replacing the codomain of the identifier function with $\{1,\ldots,Z\}$, we now have $\{x_S\colon S \in Y\} \subseteq \{1,\dots,2^{Z+1}\}$ instead of $\{x_S\colon S \in Y\} \subseteq \{1,\dots,2^{n^2+1}\}$, and therefore we need to choose each prime $p_i$ among primes in the range $\{1,\dots,2^{C(\log Z + \log n +2)}\}$. Hence, the success probability accordingly also changes. Note that the constant $C$ will need to be larger, but will remain a constant.
	Finally, similarly as argued in the proof of Theorem~\ref{thm:hc-gen}, the probability that for each configuration $c\in \conf(V(G)$ we have $|\Min(\omega,G,c)|\leq 1$ is at least $\left(1 - \frac{1}{Z \cdot n\rho(n)} \right)^{\log n} \ge \left(1 - \frac{1}{Z \cdot \rho(n)} \right)$.
\end{proof}

We now use Theorem~\ref{thm:H_x} to argue the following.

\begin{lemma}\label{lem:allMx}
	Assuming $C$ is chosen large enough, the following event happens with probability at least $1 - \frac{1} {n^2}$: for all $x\in V(F)$ and all $c \in \conf(V_x)$, we have $|\Min(\omega,{H_x},c)| \le 1$.
\end{lemma}
\begin{proof}
	Apply Theorem~\ref{thm:H_x} on each graph $H_x$ for $x\in V(F)$, where each time we let the identifier function on $E_x$ be the identifier function on $E(G)$, restricted to $E_x$. Note that $|V_x| \le \Oh(n^\alpha)$ for each $x$, for large enough $C$ hence we can use the primes $p_0,p_1,\dots,p_{\log n^\alpha+C}$ in each of these applications, and therefore obtain the weights function $\omega_0,\dots,\omega_{\log n^\alpha}$ defined in the same way on all the graphs $H_x$. By choosing $C$ appropriately large we can guarantee that for every fixed $x\in V(F)$, with probability at least $1-\frac{1}{n^4}$ we have $|\Min(\omega,{H_x},c)| \le 1$ for all $c\in \conf(V_x)$. Since $|V(F)|\leq n^2$, it now follows from union bound that this assertion holds for all $x\in F$ simultaneously with probability at least $1-\frac{1}{n^2}$.
\end{proof}
%
%
%

\subsubsection{Isolation partial solutions in subtrees}

Our goal now is to extend the conclusion of Lemma~\ref{lem:allMx} from graphs $H_x$ that are associated with single nodes $x$ of $F$ to graphs $G_x$ that reflect the whole subtree of $F$ comprising the descendants of $x$.

\begin{definition}[Graph $G_x$]
	For each node $x$ in $F$, define $G_x=(V_{x\downarrow},E_{x\downarrow})$, where 
	$$V^{\downarrow}_x\coloneqq N_G[\eta^{-1}(\subtree[x])]\qquad\textrm{and}\qquad E^\downarrow_x\coloneqq \{uv\colon uv\in E(G), u \in \eta^{-1}(\subtree[x]),\textrm{ and } v\in V^\downarrow_x\}.$$
\end{definition}  

In other words, $G_x$ is a subgraph of $G$, but now its vertex set comprises all the vertices that are mapped to nodes of $\subtree[x]$ by $\eta$ and their neighbors. Among edges with both endpoints in this vertex set we keep only those with at least one endpoint in $\eta^{-1}(\subtree[x])$. Observe that actually,
$$V^{\downarrow}_x=\bigcup_{y\in \subtree[x]} V_y\qquad\textrm{and}\qquad E^{\downarrow}_x=\bigcup_{y\in \subtree[x]} E_y.$$
Also, denote
$$W_x\coloneqq N_G(\eta^{-1}(\subtree[x])).$$
    As explained before, we will perform a bottom-up induction on $F$ to prove that for each node $x$, the relevant partial solutions in the graph $G_x$ are appropriately isolated. This will be done under that condition that the event $A$ described in Lemma~\ref{lem:allMx} holds: for all $x\in V(F)$ and all $c \in \conf(V_x)$, we have $|\Min(\omega,{H_x},c)| \le 1$.
    Formally we will prove the following induction hypothesis for all $i=\log n,\ldots,0$, starting with $i=\log n$ and decreasing $i$ at each step.
    
	\IH{Conditioned on $A$ happening, the following event happens with probability at least $\left(1 - \Oh\left(\frac{1}{n^{2}}\right) \right)^{\log n-i}$: for all nodes $x$ at level $i$ in $F$ and for any configuration $c \in \conf(W_x)$, we have $|\Min(\xi_i,{G_x},c)|\le 1$.}\\
	
	Note that if the induction hypothesis is true for $i=0$, that is, for the unique root node $r$ of $F$, then $\xi$ isolates the family of all Hamiltonian cycles in $G_r=G$ with probability at least
	$\left(1- \frac{1}{n^2}\right)\cdot \left(1- \frac{1}{n^2}\right)^{\log n}\geq 1-\frac{1}{n}$; here, the first factor is the lower bound on the probability of $A$ provided by Lemma~\ref{lem:allMx}. Therefore, it remains to perform the induction.
	
\paragraph{Base step.} For $i=\log n$, every node $x$ of $F$ at level $\log n$ is a leaf with $|\eta^{-1}(x)|=1$, say $\eta^{-1}(x)=\{v\}$. Hence we have to prove that for any configuration $c \in \conf(N_G(v))$, we have $|\Min(\xi_{\log n}),G_x,c)|\leq 1$. Notice that $G_x$ only contains edges between $v$ and its neighbors, hence for every configuration $c\in \conf(N_G(v))$ there is at most one partial solution in $G_x$ that is compliant with $c$. So for $i=\log n$ the induction hypothesis is true.
	
\paragraph{Induction step.}
	Assume the induction hypothesis is true for all nodes $x$ at level $i+1$. Let 
	$$Y' \coloneqq \bigcup_{y\colon \textrm{node at level }i+1}\ \bigcup_{c\in \conf(W_y)}\ \Min(\xi_{i+1},{G_y},c)$$ 
	be the set of all relevant minimum weight partial solutions in graphs $G_x$ for $x$ at level $i+1$.
	Furthermore, let
	\[
	Z \coloneqq \bigcup_{x\colon \textrm{node at level }i}\ \bigcup_{c\in \conf(V_x)}\ \Min(\omega,{H_x},c)
	\]
	be the set of all minimum weight compliant partial solutions for configurations on graphs $H_x$ for $x$ at depth $i$.  
	Finally, let 
	$$Y \coloneqq \{\,S_0\cup S_1 \cup \ldots \cup S_7\ \colon\ S_0\in Z\textrm{ and }S_1,\ldots,S_7 \in Y'\,\}$$ be the set comprising all combinations of $7$ partial solutions from $Y'$ and a single partial solution from $Z$. 
	
	We first prove with Claim~\ref{claim:3A} that every relevant minimum weight compliant partial solution is included in $Y$. Then Claim~\ref{claim:3B} says that with high probability, all partial solutions in $Y$ have pairwise different weights with respect to $\xi_i$. Hence, proving these two claims is sufficient to make the induction step go through.
	
	\begin{claim}
		\label{claim:3A}
		Let $x$ be a node of depth $i$ and let $c \in \conf(W_x)$. Then $\Min(\xi_i,{G_x},c)\subseteq Y$.
	\end{claim}
	\begin{proof}
	    Take any $S\in \Min(\xi_i,{G_x},c)$.
		Let $y_1,\ldots,y_t$ ($t\leq 7$) be the (at most) seven child nodes of $x$ at depth $i+1$. Let $S_0\coloneqq S\cap E_x$ and $S_j\coloneqq S\cap E^{\downarrow}_{y_j}$ for $j\in \{1,\ldots,t\}$. Note that the partial solutions $S_0,S_1,\ldots,S_t$ are pairwise disjoint and their union is equal to $S$. Further, since $S\in \Min(\xi,G_x,c)$, an argument analogous to the one used in the proof of Claim~\ref{claim:1B} shows that
		\begin{itemize}[nosep]
		 \item $S_0\in \Min(\omega,H_x,c_0)$ for some $c_0\in \conf(V_x)$; and
		 \item for each $j\in \{1,\ldots,t\}$, $S_j\in \Min(\xi_{i+1},G_{y_j},c_j)$ for some $c_j\in \conf(W_{y_j})$.
		\end{itemize}
        This means that $S_0\in Z$ and $S_1,\ldots,S_t\in Y'$, implying that $S\in Y$. 
	\end{proof}
	
	\begin{claim}
		\label{claim:3B}
		Conditioned on $A$, the probability of the following event is at least $\left(1 - \frac{1}{n^{2}} \right)^{\log n -i}$:
		for all different $S, S' \in Y$, we have $\xi_i(S) \neq \xi_i(S')$.
	\end{claim}
	\begin{proof}
        For each $S \in Y$ let us define
		$$ x_S \coloneqq \sum_{e \in S}2^{\id (e)}.$$
		Observe that since the identifiers assigned to the edges are unique, the numbers $x_S$ are pairwise different. The induction hypothesis gives us that conditioned on $A$, the probability of the following event is at least $\left(1 - \frac{1}{n^{2}} \right)^{\log n -(i+1)}$: for every node $y$ of $F$ at level $i+1$ and every $c \in \conf(W_y)$, we have $|\Min(\xi_{i+1},G_y,c)| \le 1$. We may then use Theorem~\ref{thm:rankbased} to conclude that for each such $y$,
		$$\left|\bigcup_{c \in \conf(W_y)} \Min(\xi_{i+1},{G_y},c)\right| \le 2^{\Oh(|W_y|)}\leq 2^{\Oh((n/2^{i+1})^\alpha)}.$$
		Since we assume that $A$ happens, we have $|\Min(\omega,H_x,c)|\leq 1$ for each node $x$ at level $i$. We can use Theorem~\ref{thm:rankbased} again to infer that for each such $x$,
		$$\left|\bigcup_{c \in \conf(V_x)} \Min(\omega,{H_x},c)\right| \le 2^{\Oh(|V_x|)}\le 2^{\Oh((n/2^{i})^\alpha)}.$$
		Since $F$ has at most $n^2$ nodes in total, the above bounds imply that
		$$|Y| \le |Z|\cdot |Y'|^7\leq 2^{\Oh((n/2^i)^\alpha)})\cdot n^{16}.$$ 
		Since $N_i = 2^{C(\log n + (n/2^i)\alpha)}$, and prime $q_i$ is sampled
        uniformly at random among the primes in the range $\{1,\ldots,N_i\}$,
        from Lemma~\ref{FKS} we can conclude that, provided $C$ is chosen large enough, with probability at least 
		\begin{align*}
			\left(1 - \frac{(n^2+1)\left( 2^{\Oh((n/2^i)^\alpha)}\cdot n^{16} \right)^2}{2^{(C/2)(\log n + (n/2^i)\alpha)}}\right)\cdot \left(1 - \frac{1}{n^{2}}\right)^{\log n -(i+1)} \ge \left(1 - \frac{1}{n^{2}} \right)^{\log n -i},
		\end{align*} 
		all the numbers in $\{x_S\colon S \in Y\}$ have pairwise different remainders modulo $q_i$. As a consequence, with at least the same probability we have $\xi_i(S)\neq \xi_i(S')$ for all distinct for $S,S' \in Y$.
	\end{proof}
	
	As argued, the induction step follows directly from combining Claims~\ref{claim:3A} and~\ref{claim:3B}.
	

\section{Deterministic algorithm for Hamiltonian cycle in separable classes}
\label{sec:algo}

A graph class $\Cc$ shall be called {\em{efficiently separable}} with degree $\alpha$ if it is separable with degree $\alpha$ in the sense of Definition~\ref{def:separable}, and moreover given $G\in \Cc$ and a vertex subset $S\subseteq V(G)$, a suitable balanced separator $X$ witnessing separability can be computed in polynomial time. In this section we prove the following result.

\begin{theorem}\label{thm:HamCycleSeparable}
 Let $\Cc$ be a hereditary graph class that is efficiently separable with degree $\alpha\in (0,1)$. Then there is an algorithm for {\sc{Hamiltonian Cycle}} on graphs from $\Cc$ that runs in deterministic time $2^{\Oh(n^\alpha)}$ and uses polynomial space.
\end{theorem}

It is well-known that for every fixed $H$, the class of $H$-minor-free graphs is efficiently separable with degree $\frac{1}{2}$~\cite{AlonST90,KawarabayashiR10}. Hence, Theorem~\ref{thm:HamCycleSeparable} implies Theorem~\ref{thm:alg}.

The first step towards the proof of Theorem~\ref{thm:HamCycleSeparable} is to revisit the approach of~\cite{wg2020} and extend it to obtain the following result.
\begin{lemma}
    \label{lem:wg2020}
    There is a deterministic algorithm that takes as input an undirected graph
    $G=(V,E)$ along with an elimination forest of height at most $d$, a
    weight function $\omega\colon E\rightarrow \{1,\ldots,W\}$, and a target integer~$t$. The algorithm runs in time $5^dW (n \log W)^{\Oh(1)}$, uses space that is polynomial in $n$ and $\log W$, and detects whether $G$ has a Hamiltonian cycle $C$ satisfying $\omega(C)=t$, provided there is at most one such $C$.
\end{lemma}
The extension is similar to that performed by Bj\"orklund in~\cite{bjorklund10}, where he extended his $\Oh(1.66^n)$-time algorithm for {\sc{Hamiltonian Cycle}} to an $1.66^n W (n \log W)^{\Oh(1)}$ time algorithm for the {\sc{Traveling Salesman}} problem on $n$ cities. Therefore, we sketch here the extension assuming (but recalling) the basic understanding of the approach of~\cite{wg2020}.

\begin{definition}
	Suppose $\mathbb{F}$ is a finite field. An element $\rho \in \mathbb{F}$ is a \emph{primitive $N$-root of unity} if $\rho^N=1$ and for every $0< N'< N$ it holds that $\rho^{N'}\neq 1$.
\end{definition}

It is well known that the multiplicative group of every finite field is cyclic, that is, there is a generator $g \in \mathbb{F}$ such that $\{g^0,g^1,\ldots,g^{|\mathbb{F}|-1} \}$ are all the elements of field. Then we must have $g^{-1}=g^{|\mathbb{F}|-1}$. So $g^{|F|}=1$, and therefore $g$ is a primitive $(|\mathbb{F}|-1)$-root of unity.
We will work with the field $\mathbb{F}_p$, for some prime $p$.
First we address the issue of finding a generator of $\mathbb{Z}^*_p$, the multiplicative group of $\mathbb{F}_p$:
\begin{lemma}
    A generator of $\mathbb{Z}^*_p$ can be found in $\Oh(p \cdot  \polylog(p))$
    deterministic time and $\Oh(\polylog(p))$ space.
\end{lemma}
\begin{proof}
 First, find the prime factors $p_1,\ldots,p_\ell$ using any deterministic $\Oh(p)$-time algorithm. Note that we have $\ell\leq \log p$.
 Next we rely on the the well-known fact that an element $e$ is a generator of
 $\mathbb{Z}^*_p$ if and only if for every $i=1,\ldots,\ell$ it holds that
 $e^{(p-1)/p_i} \not\equiv 1\ (\text{mod } p)$. This fact follows from
 Lagrange's theorem (see for example the discussion preceding~\cite[Theorem
 14.16]{MotwaniR95}, where this fact was used in a similar way to find
 generators probabilistically). By this fact, we can check whether $e$ is a
 generator or not in $\Oh(\polylog(p))$ time. Thus we can find a generator by simply iterating over all elements $e \in \mathbb{Z}_p$ until the check succeeds.
\end{proof}

 We will use the following well-known statement about discrete Fourier transform in finite fields (see e.g.~\cite[Equation 30.11]{cormenetal}).
 
\begin{theorem}[Discrete Fourier Inversion]\label{thm:DFT}
	Let $\mathbb{F}$ be a finite field, let $\rho\in \mathbb{F}$ be a primitive $N$-root of unity, and let $P(x) \in \mathbb{F}[x]$ be a polynomial of degree at most $N-1$ in $x$ with coefficients from $\mathbb{F}$. If $P(x)=\sum_{i=0}^{N-1}c_i x^i$, then for every $0\leq t \leq N-1$ it holds that
	\[
	c_t = \frac{1}{N}\sum_{i=0}^{N-1} \rho^{-it} P(\rho^{i}).
	\]
\end{theorem}

Let $N\coloneqq W \cdot n \cdot \log n$ and consider the field $\mathbb{F}\coloneqq \mathbb{F}_{p}$, where $p$ is a prime satisfying $p \in \Theta(N)$. Such a prime can be deterministically found in time $\Oh(p)$ and using $\polylog(p)$ space using brute-force and the polynomial-time deterministic primality testing algorithm~\cite{agrawal2004primes}.
By the above discussion, the field $\mathbb{F}$ has a
$(p-1)$-root of unity $\rho \in \mathbb{F}$,
and such a root can be found in time $\Oh(p\cdot \polylog(p))$ and space $(n \log W)^{\Oh(1)}$. Next, we continue with analysis of methods
presented in~\cite{wg2020}.

Recall that we are given a graph $G$ and an elimination forest $F$ of $G$ of height at most $d$. Since we are interested in Hamiltonian cycles in $G$, we may assume that $G$ is connected, hence $F$ is a tree.
The central objects studied in~\cite{wg2020} are polynomials $P[u,f],P(u,f) \in \mathbb{Z}[\alpha,\beta,\gamma]$, defined for each $u\in V$ and function $f\colon \tail[u]\to \{0,1_{\mathsf{L}},1_{\mathsf{R}},2_{\mathsf{L}},2_{\mathsf{R}}\}$ (or $f\colon \tail(u)\to \{0,1_{\mathsf{L}},1_{\mathsf{R}},2_{\mathsf{L}},2_{\mathsf{R}}\}$ in case of $P(u,f)$). Here, $\alpha,\beta,\gamma$ are formal variables. In~\cite{wg2020} it is shown how to compute those polynomials in a bottom-up manner over the given elimination forest $F$. Further, the parity of the number of Hamiltonian cycles of total weight $w$ can be inferred from the coefficient of the monomial $\alpha^w\beta^n\gamma^n$ in the polynomial $P(r,\emptyset)$, where $r$ is the root of~$F$. Therefore, the idea in~\cite{wg2020} was to use Isolation Lemma to ensure that provided the graph is Hamiltonian, with high probability there exists $w\in [N]$ for which there is exactly one Hamiltonian cycle of weight $w$. Then one explicitly computes all the polynomials $P[u,f],P(u,f)$ in a bottom-up manner over the tree $F$ in time $5^dW (n \log W)^{\Oh(1)}$. Finally, the existence of a Hamiltonian cycle can be inferred from the analysis of the coefficients of $P(r,\emptyset)$.

In the setting of Lemma~\ref{lem:wg2020} we can almost use the same strategy, however there is a caveat. Namely, the expansion of each polynomial $P(u,f)$ and $P[u,f]$ into a sum of monomials of the form $\alpha^a\beta^b\gamma^c$ may have length as large as $(N+1)(n+1)^2$, because the relevant values of $a,b,c$ are $a\in \{0,\ldots,N\}$ and $b,c\in \{0,\ldots,n\}$. Therefore, storing the coefficients of $P[u,f]$ explicitly would take at least space $\Oh(Nn^2)=\Oh(Wn^3)$, which is more than $(n+\log W)^{\Oh(1)}$ promised in the statement of Lemma~\ref{lem:wg2020}.

Therefore, the idea is not to compute the whole polynomial $P(r,\emptyset)$ explicitly, but evaluate the relevant coefficients of $P(r,\emptyset)$ one by one using Theorem~\ref{thm:DFT}. Precisely, let $Q=\sum_{w=0}^N c_{w,n,n}\cdot\alpha^w\in \mathbb{Z}[\alpha]$, where $c_{w,n,n}$ is the coefficient of $\alpha^w\beta^n\gamma^n$ in $P(r,\emptyset)$. After casting $Q$ as a polynomial $Q' \in \mathbb{F}_p[\alpha]$, we can use the method presented in~\cite{wg2020} to give an algorithm that evaluates $Q'(e)$ for a given $e\in \mathbb{F}_p$ in time $5^dW (n \log W)^{\Oh(1)}$ and using $(n \log W)^{\Oh(1)}$ space, because storing an element of $\mathbb{F}_p$ requires $ (n \log W)^{\Oh(1)}$ space. This is enough to compute the formula described in Theorem~\ref{thm:DFT} within the promised complexity guarantees.
This concludes the description of how to compute the coefficient $c'_{t,n,n}$ of $Q'[\alpha]$ within the stated resource bounds.

Now we can compute the matching coefficient $c_{t,n,n}$ of $Q[\alpha]$ by observing that $c'_{t,n,n}=c_{t,n,n} \bmod p$, applying the above step $\Theta(n)$ times for different primes $p$, and reconstructing $c_{t,n,n}$ with the Chinese Remainder Theorem. Here, it is important to note that the coefficient $c_{t,n,n}$ is of the order $2^{\Oh(n)}$, hence the information about $c_{t,n,n}$ modulo $\Theta(n)$ different primes is sufficient to reconstruct $c_{t,n,n}$ completely.
%
This concludes the sketch of the proof of Lemma~\ref{lem:wg2020}.
%
%

We now use Lemma~\ref{lem:wg2020} to prove Theorem~\ref{thm:HamCycleSeparable}.
\begin{proof}[Proof of Theorem~\ref{thm:HamCycleSeparable}]

    Let $G=(V,E)\in \Cc$ be the given graph, where $n\coloneqq |V|$.  By
    iteratively extracting balanced separators (cf., \cite[Theorem A.1]{wg2020})
    we can compute an elimination forest of $G$ of height $\Oh(n^{\alpha})$ in
    polynomial time. 
    
    Next we use Theorem~\ref{thm:separable-iso} that gives us a
    set of weight functions $\omega_1,\ldots,\omega_\ell \colon E \rightarrow
    [2^{\Oh(n^{\alpha})}]$ with $\ell = 2^{\Oh(n^{\alpha})}$ such that at least
    half of the functions $\omega_i$ isolate the family of Hamiltonian cycles in
    $G$. 
    
    It can be easily seen by inspecting the construction of the isolation scheme
    of Theorem~\ref{thm:separable-iso} that the functions
    $\omega_1,\ldots,\omega_\ell$ can be enumerated one by one using polynomial
    working space and $2^{\Oh(n^{\alpha})}$ time. Namely, we simply need to iterate
    over every tuple of $\log{n}+1$ primes $p_0 \in [M_0],\ldots,p_{\log{n}} \in [M_{\log{n}}]$. To
    achieve that, we can iterate over every prime number $p_i \in [M_i]$ in time $M_i^{\Oh(1)}$
    (just iterate through $[M_i]$ and deterministicaly check
    for primality in $M_i^{\Oh(1)}$ time with a brute-force algorithm). Therefore,
    enumerating all weight functions $\omega_1,\ldots,\omega_\ell$ can be done in
    $\text{poly}(M_1 \cdot \ldots \cdot M_{\log{n}}) \le 2^{\Oh(n^\alpha)}$
    additional time and polynomial space.
    Theorem~\ref{thm:separable-iso} guarantees that among the enumerated weight functions $\omega_1,\ldots,\omega_\ell$, there is at least one (and even half of them) that isolates the family of Hamiltonian cycles in $G$.

    Therefore, it remains to apply the algorithm of 
    Lemma~\ref{lem:wg2020} to each consecutive function $\omega_i$ and each possible minimum weight $t \leq 2^{\Oh(n^\alpha)}$, and report a positive outcome if any of these applications finds a Hamiltonian cycle in $G$. The time complexity is bounded by $2^{\Oh(n^\alpha)}$ and the space complexity is polynomial in $n$ and $\log 2^{\Oh(n^\alpha)}=\Oh(n)$.
\end{proof}

\newcommand{\inc}{\mathsf{inc}}
\newcommand{\ind}{\mathsf{ind}}
\newcommand{\Gb}{\mathbf{G}}
\newcommand{\forget}{\mathsf{forget}}
\newcommand{\select}{\mathsf{select}}
\newcommand{\Formulas}{\mathsf{Formulas}}
\newcommand{\sgm}[2]{#1\langle #2\rangle}
\newcommand{\asm}{($\bigstar$)}

\section{MSO-definable problems}
\label{sec:mso}

\subsection{Definitions}

\paragraph*{$\CMSOtwo$ Logic.}
We work with problems definable in logic $\CMSOtwo$, which stands for Monadic Second-Order logic on graphs with modular counting predicates and quantification over edge subsets. Recall that in this logic we have variables for individual vertices, individual edges, sets of vertices, and sets of edges; the latter two kinds are called {\em{monadic variables}}. The basic constructs in $\CMSOtwo$ are {\em{atomic formulas}} of the following forms:
\begin{itemize}[nosep]
 \item {\em{Equality:}} $x=y$, checking equality of $x$ and $y$;
 \item {\em{Membership:}} $x\in X$, checking that $x$ belongs to $X$;
 \item {\em{Incidence:}} $\inc(u,e)$, checking that vertex $u$ is incident on the edge $e$; and
 \item {\em{Congruence:}} $|X|\equiv a\bmod p$, where $a,p$ are constants, with the expected semantics.
\end{itemize}
$\CMSOtwo$ formulas can be constructed from atomic formulas using standard boolean connectives, negation, and quantification over variables of each of the four kinds (both existential and universal). Note that a $\CMSOtwo$ formula can have {\em{free variables}} that are not bound by any quantification. A formula can be applied on a graph supplied with a valuation of the free variables. For example, the formula
\begin{equation}\label{eq:indset}
\varphi(X)=\big[\forall_u\forall_v \left(u\in X\wedge v\in X\wedge u\neq v\right)\implies \left(\neg\exists_e\, \inc(u,e)\wedge \inc(v,e)\right)\big], 
\end{equation}
when applied on a graph $G$ and a vertex subset $A$, checks whether $A$ is an independent set in $G$. If this is the case, we write $G\models \varphi(A)$.

Let $\varphi(X)$ be a $\CMSOtwo$ formula with one free vertex set variable $X$. For a graph $G$, we define
$$\select_\varphi(G)\coloneqq \{S\subseteq V(G)~|~G\models\varphi(S)\}.$$
For example, if $\varphi(X)$ is the formula presented in~\eqref{eq:indset}, then $\select_\varphi(G)$ consists of all independent sets in~$G$. If $X$ is an edge set variable, then $\select_\varphi(G)$ is defined analogously: it comprises all subsets $S$ of edges of $G$ for which $\varphi(S)$ is satisfied. Thus, if $\varphi(X)$ is a $\CMSOtwo$ formula with a free monadic variable $X$, then $\select_\varphi$ is a vertex or edge selection problem, depending on whether $X$ is a vertex set or an edge set variable. A vertex/edge selection problem is {\em{$\CMSOtwo$-definable}} if it is of the form $\select_\varphi$ for a formula $\varphi(X)$ as above.

\paragraph*{Boundaried graphs.} Throughout this section we assume that all considered graphs have vertices from a fixed countable set $\Omega$. The reader may think that $\Omega=\mathbb{N}$.

A {\em{boundaried graph}} is a pair consisting of a graph $G$ and a subset of its vertices $B$, called the {\em{boundary}}. We have two natural operations on boundaried graphs:
\begin{itemize}[nosep]
 \item Suppose $\Gb_1=(G_1,B_1)$ and $\Gb_2=(G_2,B_2)$ are boundaried graphs such that $V(G_1)\cap V(G_2)\subseteq B_1\cup B_2$. Then the {\em{sum}} of $\Gb_1$ and $\Gb_2$ is the boundaried graph
 $$\Gb_1\oplus \Gb_2\coloneqq (G_1\cup G_2,B_1\cup B_2),$$
 where $G_1\cup G_2=(V(G_1)\cup V(G_2),E(G_1)\cup E(G_2))$. Note that the sum is not defined if the condition $V(G_1)\cap V(G_2)\subseteq B_1\cup B_2$ does not hold.
 \item Suppose $\Gb=(G,B)$ is a boundaried graph and $A\subseteq B$. Then the operation of {\em{forgetting}} $A$ in $\Gb$ yields the boundaried graph
 $$\forget_A(\Gb)\coloneqq (G,A).$$
 Note that for notational convenience, the set indicated in the subscript is the new boundary set, and not the set $B\setminus A$ of vertices that get forgotten, i.e., removed from the boundary.
\end{itemize}
The following standard lemma connects the operations on boundaried graphs with the concept of treewidth.

\begin{lemma}
 A graph $G$ has treewidth at most $k$ if and only if $(G,\emptyset)$ can be obtained from graphs on at most $k+1$ vertices by a repeated use of the sum and forget operations, where at each moment in the construction all boundaried graphs have boundaries of size at most $k+1$. 
 \end{lemma}
 
\paragraph*{Configuration schemes.} We now present the formalism of configuration schemes for selection problems. The notational layer is directly taken from the recent work of Chen et al.~\cite{Chenetal20}. However, in general, the algebraic approach to graphs of bounded treewidth and recognizability of their properties dates back to the foundational work of Courcelle and others done in the 90s. See the book of Courcelle and Engelfriet for an introduction to the area~\cite{CourcelleE12}.

For concreteness we focus on edge selection problems.
Adjusting the definitions to vertex selection problems is straightforward.

A {\em{configuration scheme}} is a pair of functions $(\conf,c)$ with the following properties:
\begin{itemize}[nosep]
 \item $\conf$ assigns to each finite subset $B\subseteq \Omega$ a finite {\em{configuration set}} $\conf(B)$. 
 We require that the cardinality of the configuration set is uniformly and effectively bounded in the size of $B$, that is, there exists a computable function $g$ such that $|\conf(B)|\leq g(|B|)$ for each finite $B\subseteq \Omega$.
 \item For every boundaried graph $\Gb=(G,B)$ and a subset of edges $S\subseteq E(G)$, $c$ maps the pair $(\Gb,S)$ to a configuration $c(\Gb,S)\in \conf(B)$. 
\end{itemize}
We say that a configuration scheme $(\conf,c)$ is {\em{compositional}} if the following two conditions hold:
\begin{itemize}
 \item For every pair of boundaried graphs $\Gb_1=(G_1,B_1)$ and $\Gb_2=(G_2,B_2)$ (with defined sum), and subsets of edges $S_1\subseteq E(G_1)$ and $S_2\subseteq E(G_2)$, the configuration
 $$c(\Gb_1\oplus \Gb_2,S_1\cup S_2)$$
 depends only on the pair of configurations
 $$c(\Gb_1,S_1)\qquad\textrm{and}\qquad c(\Gb_2,S_2).$$
 \item For every boundaried graph $\Gb=(G,B)$, $A\subseteq B$, and a subset of edges $S\subseteq E(G)$, the configuration
 $$c(\forget_A(\Gb),S)$$
 depends only on the configuration
 $$c(\Gb,S).$$
\end{itemize}
In other words, the first condition means that we can endow the set $\conf(B_1)\times \conf(B_2)$ with a sum operation $\oplus$ so that the operators $\oplus$ and $c(\cdot,\cdot)$ commute. Similarly, the second condition means that $\conf(B)$ can be endowed with an operation $\forget_A(\cdot)$ so that the operators $\forget_A(\cdot)$ and $c(\cdot,\cdot)$ commute. We will therefore use the operators $\oplus$ and $\forget_A$ also as operators defined on configuration sets provided by $\conf$. Note here that since $\oplus$ is commutative on boundaried graphs, the corresponding operator $\oplus$ on configurations can also be chosen to be commutative.

Suppose $\Pp$ is an edge selection problem, that is, a function that with each graph $G$ associates a family of edge subsets $\Pp(G)\subseteq 2^{E(G)}$. We say that a configuration scheme $(\conf,c)$ {\em{recognizes}} $\Pp$ if there exists a set of {\em{final configurations}} $F\subseteq \conf(\emptyset)$ such that for every graph $G$ and a subset of edges $S\subseteq E(G)$, 
$$S\in \Pp(G)\qquad\textrm{if and only if}\qquad c((G,\emptyset),S)\in F.$$
The next lemma follows from well-known compositionality properties of the $\CMSOtwo$ logic. We sketch the proof for completeness, but we remark that an essentially the same sketch was also provided in~\cite{Chenetal20}.

\begin{lemma}\label{lem:mso-conf}
 Every $\CMSOtwo$-definable edge selection problem is recognized by a compositional configuration scheme. 
\end{lemma}
\begin{proof}[Proof sketch]
 Let $\Pp=\select_\varphi$ be the problem in question, where $\varphi(X)$ is a $\CMSOtwo$ formula with a free edge set variable $X$. Let $q$ be the quantifier rank of $\varphi$, that is, the maximum number of nested quantifiers in $\varphi$.
 
 Consider a finite set $B\subseteq \Omega$ and all $\CMSOtwo$ formulas $\psi(X)$
 of quantifier rank at most $q$, where $X$ is an edge set variable, which can
 also use the elements of $B$ as constants (formally, the signature additionally
 contains every element of $B$ as a constant). It is well-known that there are
 only finitely many pairwise non-equivalent such formulas, in the sense that
 $\psi$ and $\psi'$ are equivalent if for every boundaried graph $\Gb=(G,B)$ and
 $S\subseteq E(G)$, we have $\Gb\models \psi(S)$ if and only if $\Gb\models
 \psi'(S)$. Moreover, the number of equivalence classes is bounded by a
 computable function of $|B|$. Then, let $\Formulas^q(B)$ be a set comprised of
 one arbitrarily selected representative from each equivalence class.

 We now define the configuration scheme $(\conf,c)$ as follows:
 \begin{itemize}
  \item For each finite $B\subseteq \Omega$, we set
  $$\conf(B)\coloneqq 2^{\Formulas^q(B)},$$
  that is, $\conf(B)$ is the powerset of $\Formulas^q(B)$.
  \item For each boundaried graph $\Gb=(G,B)$ and an edge subset $S\subseteq E(G)$, we set
  $$c(\Gb,S)\coloneqq \{\psi(X)\in \Formulas^q(B)~|~\Gb\models \psi(S)\}.$$
 \end{itemize}
 A standard argument involving Ehrenfeucht-Fra\"isse games shows that this configuration scheme is compositional. It remains to observe that $(\conf,c)$ recognizes $\select_\varphi$ by taking
 $$F\coloneqq \{\Delta\subseteq \Formulas^q(\emptyset)~|~\varphi\in \Delta\}.\qedhere$$
 \end{proof}

\subsection{Isolation}



The remainder of this section is devoted to the proof of Theorem~\ref{thm:mso}. We remark that the same technique can be used to prove the same result also for vertex selection problems. We leave the straightforward modifications to the reader.

\medskip

By Lemma~\ref{lem:mso-conf}, $\Pp$ is recognized by a compositional configuration scheme $(\conf,c)$. Let $F\subseteq \conf(\emptyset)$ be the set of final configurations, as in the definition of recognizing $\Pp$. Recall that there exists a computable function $g\colon \N\to \N$ such that
$$|\conf(B)|\leq g(|B|)\qquad\textrm{for each finite }B\subseteq \Omega.$$
We may assume that the configuration scheme $(\conf,c)$ satisfies the following assertion for every boundaried graph $\Gb=(G,B)$ and $S,S'\subseteq E(G)$:
\begin{equation}\label{eq:boundary-fixed}
 \textrm{if }c(\Gb,S)=c(\Gb,S'),\quad\textrm{then }S\cap \binom{B}{2}=S'\cap \binom{B}{2}.
\end{equation}
Indeed, if assertion~\eqref{eq:boundary-fixed} is not satisfied, then we can instead use the configuration scheme $(\conf',c')$ defined as
$$\conf'(B)\coloneqq \conf(B)\times 2^{\binom{B}{2}}\qquad\textrm{and}\qquad c'(\Gb,S)=\left(c(\Gb,S),S\cap \binom{B}{2}\right),$$
which already satisfies assertion~\eqref{eq:boundary-fixed}. Note here that $(\conf',c')$ still recognizes $\Pp$ (by taking $F'\coloneqq F\times \{\emptyset\}$) and $|\conf'(B)|$ is still bounded by a computable function of $|B|$. We remark, however, that the configuration scheme provided by Lemma~\ref{lem:mso-conf} actually already satisfies~\eqref{eq:boundary-fixed}, provided the quantifier rank $q$ is positive.

To prove Theorem~\ref{thm:mso}, it suffices to give an isolation scheme for $\Pp$ on graphs of treewidth at most $k$ that uses at most
$\Oh(\log \Gamma\log n+\log^2 n)$ random bits, where
$$\Gamma\coloneqq \max\left(g(2k+2)^5,2^{\binom{2k+2}{2}}\right),$$
and assigns weights that are at most exponential in the number of random bits. From now on let us fix $k$, the graph $G$ given to the isolation scheme on input, and a tree decomposition $\treedecomp=(T,\beta)$ of $G$ of width at most $k$. We may assume that $T$ has at most $n$ nodes, where $n\coloneqq |V(G)|$ is the vertex count of $G$.

\paragraph*{Isolation scheme.} We first present the isolation scheme, which is just a general version of the scheme for {\sc{Hamiltonian Cycle}} presented in Section~\ref{sec:ham-tw} (but without the refinement of using the rank-based approach). First, let $\id\colon E(G)\to \{0,\ldots,|E(G)|-1\}$ be any bijection that assigns to each edge $e\in E(G)$ its unique {\em{identifier}} $\id(e)$. Then choose $1+\log n$ primes $$p_0,p_1,\ldots,p_{\log n}\in \{1,\ldots,M\}$$
uniformly and independently at random among primes in $\{1,\ldots,M\}$,
where
$$M\coloneqq \Gamma^4\cdot n^{14}.$$
Note that selecting one prime $p_i$ requires $\Oh(\log M)=\Oh(\log \Gamma+\log n)$ random bits, hence we have used $\Oh(\log \Gamma\log n+\log^2 n)$ random bits in total.

Next, we inductively define weight functions $\omega_0,\omega_1,\ldots,\omega_{\log n}$ on $E(G)$ as follows:
\begin{itemize}[nosep]
 \item Set $\omega_0(e)=2^{\id(e)}\bmod p_0$ for all $e\in E(G)$.
 \item For $e\in E(G)$ and $i=1,2,\ldots,\log n$, set
 $$\omega_i(e)\coloneqq Mn^2\cdot \omega_{i-1}(e)+\left (2^{\id(e)}\bmod p_i\right).$$
\end{itemize}
Let
$$\omega\coloneqq \omega_{\log n}$$
and observe that $\omega$ assigns weights upper bounded by $2\cdot (Mn^2)^{\log n}=2^{\Oh(\Gamma\log n+\log^2 n)}$, as required. 

It remains to verify that $\omega$ isolates $\Pp(G)$ with high probability. We do this by induction. Recall that a {\em{segment}} in the tree $T$ is a connected subtree $I$ of $T$ such that the boundary $\bnd I$ --- the set of nodes of $I$ with neighbors outside of $I$ --- has size at most $2$. For a segment $I$ we define the boundaried graph
$$\sgm{\Gb}{I}\coloneqq \left(G\left[\bigcup_{x\in I}\beta(x)\right],\bigcup_{x\in \bnd I} \beta(x)\right).$$
Note that this definition is slightly different than the one used in Section~\ref{sec:hamcycle}, as $\sgm{\Gb}{I}$ contains all the edges of $G$ with both endpoints in $\bigcup_{x\in I}\beta(x)$.
We will also write $\bnd \sgm{\Gb}{I}\coloneqq \bigcup_{x\in \bnd I} \beta(x)$ for the boundary of $\sgm{\Gb}{I}$.
For $\gamma\in \conf(\bnd \sgm{\Gb}{I})$, let us define the set of {\em{partial solutions}} in $\sgm{\Gb}{I}$ that yield configuration $\gamma$ as: 
$$\Ss(I,\gamma)\coloneqq \{S\subseteq E(\sgm{\Gb}{I})~|~c(\sgm{\Gb}{I},S)=\gamma\}.$$
We will prove the following statement by induction on $i$.

\begin{claim}\label{cl:mso-induction}
 For each $i\in \{0,1,2,\ldots,\log n\}$ and each segment $I$ of $T$ with at most $2^i$ edges, the probability that both of the following events happen is at least $1-\frac{8^{i}}{n^5}$: 
 \begin{itemize}
  \item $\omega_i$ isolates $\Ss(I,\gamma)$ for each $\gamma\in \conf(\bnd \sgm{\Gb}{I})$ for which $\Ss(I,\gamma)\neq \emptyset$; and
  \item $\min \omega_i(\Ss(I,\gamma))\neq \min \omega_i(\Ss(I,\gamma'))$ for all $\gamma,\gamma'\in \conf(\bnd \sgm{\Gb}{I})$ such that $\gamma\neq \gamma'$, $\Ss(I,\gamma)\neq \emptyset$, and $\Ss(I,\gamma')\neq \emptyset$.
 \end{itemize}
\end{claim}

Note that
$$\Pp(G)=\bigcup_{\gamma\in F} \Ss(T,\gamma).$$
Hence, since $T$ has less than $n=2^{\log n}$ edges, Claim~\ref{cl:mso-induction} for $i=\log n$ and $I=T$ implies that $\omega$ isolates $\Pp(G)$ with probability at least $1-\frac{8^{\log n}}{n^5}=1-\frac{1}{n^2}$. So it remains to prove Claim~\ref{cl:mso-induction}, which we do by induction on $i$.

\paragraph*{Base step.} For $i=0$ we have that segment $I$ has at most one edge, so it has at most two nodes. Then $\sgm{\Gb}{I}$ is a boundaried graph on at most $2k+2$ vertices, hence it has at most $\binom{2k+2}{2}$ edges. For each $S\subseteq E(\sgm{\Gb}{I})$, let $$x_S=\sum_{e\in S} 2^{\id(e)}.$$
Observe that since the identifiers assigned to edges are pairwise different, the numbers $x_S$ for $S\subseteq E(\sgm{\Gb}{I})$ are also pairwise different. 
Since those numbers are upper bounded by $2^{\binom{n}{2}}$ and there are at most
$2^{\binom{2k+2}{2}}$ of them, $M\geq 2^{4\binom{2k+2}{2}}\cdot n^{14}$ and
$p_0$ is drawn uniformly at random among primes in the range $\{1,\ldots,M\}$,
from Lemma~\ref{FKS} we can conclude that with probability at least $1-\frac{1}{n^5}$, all the numbers $\{x_S\colon S\subseteq E(\sgm{\Gb}{I})\}$ have pairwise different remainders modulo $p_0$. Since $\omega_0(S)\equiv x_S\bmod p_0$, this means that with probability at least $1-\frac{1}{n^5}$, all subsets of $E(\sgm{\Gb}{I})$ receive pairwise different weights with respect to~$\omega_0$. This implies the conclusion of Claim~\ref{cl:mso-induction} for the segment $I$.

\paragraph*{Induction step.} Consider now any $i\geq 1$ and a segment $I$ in $T$ that has more than one but at most $2^i$ edges. By Lemma~\ref{lemma:5SubTrees}, segment $I$ can be partitioned into at most five segments $I_1,\ldots,I_t$ ($t\leq 5$), each with at most $2^{i-1}$ edges, so that
\begin{equation}\label{eq:addI}
\sgm{\Gb}{I}=\forget_{\bnd\sgm{\Gb}{I}}\left(\sgm{\Gb}{I_1}\oplus\ldots\oplus \sgm{\Gb}{I_t}\right).
\end{equation}
By induction assumption, for each segment $I_j$, $j\in \{1,\ldots,5\}$, the property described in Claim~\ref{cl:mso-induction} holds with probability at least $1-\frac{8^{i-1}}{n^5}$. By union bound, this property holds for all $I_j$, $j\in \{1,\ldots,5\}$, simultaneously with probability at least $1-\frac{5\cdot 8^{i-1}}{n^5}$. We proceed under the assumption that this is the case; we call this supposition~\asm.

We define
$$\Lambda\coloneqq \conf(\bnd \sgm{\Gb}{I_1})\times \ldots \times \conf(\bnd \sgm{\Gb}{I_t}).$$
Consider any $S\subseteq E(\sgm{\Gb}{I})$. We shall say that $S$ is {\em{compatible}} with $(\gamma_1,\ldots,\gamma_t)\in \Lambda$ if there exists a partition $\{S_1,\ldots,S_t\}$ of $S$ such that $$S_j\subseteq E(\sgm{\Gb}{I_j})\quad\textrm{and}\quad S_j\in \Ss(I_j,\gamma_j)\quad\textrm{for each }j\in \{1,\ldots,t\}.$$ (Note that $S$ can be simultaneously compatible with several elements of $\Lambda$, as edges of $S$ that belong to several of the graphs $\sgm{\Gb}{I_j}$ may be placed in different parts $S_j$.) For $\bar\gamma\in \Lambda$, we denote
$$\Rr_{\bar \gamma}\coloneqq \{S\subseteq E(\sgm{\Gb}{I})~|~S\textrm{ is compatible with }\bar \gamma\}.$$
We claim the following.

\begin{claim}\label{cl:unique-minimizers}
 For each $\bar\gamma\in \Lambda$ the weight function $\omega_{i-1}$ isolates the family $\Rr_{\bar \gamma}$, provided this family is nonempty.  
\end{claim}
\begin{proof}
 Let $\bar \gamma=(\gamma_1,\ldots,\gamma_t)$.
 Suppose there exist two edge subsets $S$ and $S'$ that are both compatible with $\bar \gamma$ and have the same minimum weight with respect to $\omega_{i-1}$. Let $\{S_1,\ldots,S_t\}$ and $\{S_1',\ldots,S_t'\}$ be partitions of $S$ and~$S'$, respectively, that witness compatibility. Observe that for each $j\in \{1,\ldots,t\}$, $S_j$ must be an element of minimum weight with respect to $\omega_{i-1}$ in the family $\Ss(I_j,\gamma_j)$. Indeed, otherwise we could replace $S_j$ with a partial solution $\widehat{S}_j\in \Ss(I_j,\gamma_j)$ with $\omega_{i-1}(\widehat{S}_j)<\omega_{i-1}(S_j)$, thus obtaining an edge subset $\widehat{S}\coloneqq (S\setminus S_j)\cup \widehat{S}_j$ of weight smaller than that of $S$ that is also compatible with $\bar \gamma$. Note here that assumption~\eqref{eq:boundary-fixed} implies that $\{S_1,\ldots,S_{j-1},\widehat{S}_j,S_{j+1},\ldots,S_t\}$ is still a partition of $S$. The same observation applies also to~$S_j'$. Assumption \asm{} implies that $\omega_{i-1}$ isolates $\Ss(I_j,\gamma_j)$, which means that $S_j=S_j'$. This applies to each $j\in \{1,\ldots,t\}$, hence $S=S'$.
\end{proof}

For $\gamma\in \conf(\bnd \sgm{\Gb}{I})$, let
$\Lambda_\gamma$ be the set of all those $\bar\gamma=(\gamma_1,\ldots,\gamma_t)\in \Lambda$ for which
\begin{itemize}[nosep]
 \item $\forget_{\bnd\sgm{\Gb}{I}}\left(\gamma_1\oplus\ldots\oplus \gamma_t\right)=\gamma$ and
 \item the family $\Rr_{\bar \gamma}$ is nonempty.
\end{itemize}
Note that sets $\{\Lambda_\gamma\colon \gamma\in \conf(\bnd \sgm{\Gb}{I})\}$ are pairwise disjoint. For each $\bar \gamma\in \Lambda_\gamma$, let $S_{\bar \gamma}$ be the element of $\Rr_{\bar \gamma}$ that has the minimum weight with respect to~$\omega_{i-1}$. By Claim~\ref{cl:unique-minimizers}, there is a unique such element. Let now
$$\Mm_\gamma\coloneqq \{S_{\bar \gamma}~|~\bar\gamma\in \Lambda_\gamma\textrm{ and }\omega_{i-1}(S_{\bar \gamma})=\min \{\omega_{i-1}(S_{\bar \delta})\colon \bar \delta\in \Lambda_\gamma\}\}.$$
In other words, $\Mm_\gamma$ comprises sets $S_{\bar \gamma}$ of minimium weight with respect to $\omega_{i-1}$ among $\{S_{\bar \delta}\colon \bar \delta\in \Lambda_\gamma\}$. Let us observe the following.

\begin{claim}\label{cl:inMm}
 Suppose $S\in \Ss(I,\gamma)$ is such that $\omega_i(S)=\min \omega_i(\Ss(I,\gamma))$. Then $S\in \Mm_\gamma$.
\end{claim}
\begin{proof}
 We first observe that
 \begin{equation}\label{eq:magnitude}
\sum_{e\in S} \left(2^{\id(e)}\bmod p_i\right)\leq p_i\cdot |S|<Mn^2.  
 \end{equation}
 Therefore, by the definition of $\omega_i$, assertion $\omega_i(S)=\min \omega_i(\Ss(I,\gamma))$ implies also that
 \begin{equation}\label{eq:om-prev}
 \omega_{i-1}(S)=\min \omega_{i-1}(\Ss(I,\gamma)).
 \end{equation}
 Consider any partition $\{S_1,\ldots,S_t\}$ of $S$ such that $S_j\subseteq E(\sgm{\Gb}{I_j})$, for each $j\in \{1,\ldots,t\}$. Let $\gamma_j\coloneqq c(\sgm{\Gb}{I_j},S_j)$. Clearly, partition $\{S_1,\ldots,S_t\}$ witnesses that $S$ is compatible with $\bar \gamma\coloneqq(\gamma_1,\ldots,\gamma_t)$, hence $S\in \Rr_{\bar \gamma}$. Further, by~\eqref{eq:addI} we infer that $\forget_{\bnd\sgm{\Gb}{I}}\left(\gamma_1\oplus\ldots\oplus \gamma_t\right)=\gamma$. We conclude that $\bar \gamma \in \Lambda_\gamma$. As $\Rr_{\bar \gamma}\subseteq \Ss(I,\gamma)$,~\eqref{eq:om-prev} implies that $S=S_{\bar \gamma}$. We can now use~\eqref{eq:om-prev} again to conclude that $S\in \Mm_\gamma$. 
\end{proof}

Let now
$$\Mm\coloneqq \bigcup_{\gamma\in \conf(\bnd \sgm{\Gb}{I})} \Mm_{\gamma}.$$
Note that
$$|\Mm|\leq |\Lambda|\leq g(2k+2)^5\leq \Gamma.$$
Since $p_i$ is chosen uniformly at random among primes in the range
$\{1,\ldots,M\}$, where $M=\Gamma^4\cdot n^{14}$, from Lemma~\ref{FKS} we infer that the following event happens with probability at least $1-\frac{1}{n^{5}}$ (conditioned on \asm{} happening):
\begin{equation}\label{eq:FKScor}
\left(\sum_{e\in S} 2^{\id(e)}\right)\bmod p_i \neq \left(\sum_{e\in S'} 2^{\id(e)}\right)\bmod p_i\qquad\textrm{for all }S,S'\in \Mm, S\neq S'.
\end{equation}
As argued in~\eqref{eq:magnitude}, we have $\sum_{e\in S} \left(2^{\id(e)}\bmod p_i\right)<Mn^2$, hence event~\eqref{eq:FKScor} happening entails that the mapping
$$S\mapsto (\omega_i(S)\bmod Mn^2)\bmod p_i$$
is injective on $\Mm$. In particular, the elements of $\Mm$ receive pairwise different weights w.r.t.~$\omega_i$.

All in all, we conclude that the probability that \asm{} happens and that the elements of $\Mm$ receive pairwise different weights w.r.t.~$\omega_i$ is at least
$$\left(1-\frac{5\cdot 8^{i-1}}{n^5}\right)\left(1-\frac{1}{n^{5}}\right)\geq 1-\frac{8^{i}}{n^5}.$$
It now suffices to observe that Claim~\ref{cl:inMm} together with injectivity of $\omega_i$ on $\Mm$ directly imply the conclusion of Claim~\ref{cl:mso-induction} for segment $I$.

\newcommand{\wt}[1]{\widetilde{#1}}

\section{Lower Bounds}\label{sec:lbs}

\subsection{Unconditional Lower Bounds}

In this subsection we present several lower bounds against the existence of (oblivious) isolation schemes
for $\mathsf{NP}$-hard problems on graphs with constant treedepth or pathwidth.
We start with a standard information-theoretical lower bound on the number of random bits needed for construction of a weight function. A similar statement can be found
in~\cite{isolation-lemma2}

\begin{lemma}[cf. proof of Theorem 1 in~\cite{isolation-lemma2}]
    \label{pigeonhole}
    For every $\beta>0$ there exists $\alpha>0$ such that the following holds. 
    Suppose $n,W\in \nat$ are large enough (depending on $\beta$) and $\omega_1,\ldots,\omega_t \colon [n] \rightarrow [W]$ are weight assignments, where $t\leq \alpha\cdot \frac{n}{\log (nW)}$. Then there exist two different subsets $A,B \subseteq [n]$ such that
    \begin{itemize}[nosep]
        \item $A\cup B=[n]$,
        \item $|A\setminus B| = |B \setminus A| \le \beta n$, and
        \item for every $i \in [t]$ it holds that $\omega_i(A) = \omega_i(B)$.
    \end{itemize}
\end{lemma}
\begin{proof}
    We may assume that $\beta\leq \frac{1}{2}$.
    Observe that then there exists $\alpha>0$ such that
    \begin{equation}\label{eq:alpha-choice}
        \binom{n}{\lfloor\beta n\rfloor}>1+2^{\alpha n}\qquad\textrm{for all }n\in \nat\textrm{ large enough.}
    \end{equation}
    We verify that the statement of the lemma holds for $\alpha$ chosen as above.

    Consider a function $\Phi \colon 2^{[n]} \rightarrow \nat^t$ defined as:
    \begin{displaymath}
        \Phi(X) \coloneqq (\omega_1(X),\omega_2(X),\ldots,\omega_t(X)).
    \end{displaymath}
    Since for every nonempty $X \subseteq [n]$ and $i \in [t]$ we have $1\leq \omega_i(X)\leq nW$, function $\Phi$ can take at most $1+(nW)^t$ different values. As we assumed that $t\leq \alpha\cdot \frac{n}{\log (nW)}$, by~\eqref{eq:alpha-choice} we have 
    $$1+(nW)^t \leq 1+2^{\alpha n} <\binom{n}{\lfloor \beta n\rfloor}.$$
    By pigeonhole principle it follows that there exist two
    different sets $A',B' \subseteq [n]$, each of size exactly $\lfloor\beta n\rfloor$, that have the same value assigned by
    $\Phi$. This means that $\omega_i(A')=\omega_i(B')$ for all $i\in [t]$. Let now
    $$A\coloneqq A'\cup ([n]\setminus (A'\cup B'))\qquad\textrm{and}\qquad B\coloneqq B'\cup ([n]\setminus (A'\cup B')).$$
    Note that for each $i\in [t]$, we have
    $$\omega_i(A)=\omega_i(A')+\omega_i([n]\setminus (A'\cup B'))=\omega_i(B')+\omega_i([n]\setminus (A'\cup B'))=\omega_i(B).
    $$
    Further,
    $$|A\setminus B|=|A'\setminus B'|\leq |A'|\leq \beta n,$$
    and similarly $|B\setminus A|\leq \beta n$.
    Thus, sets $A$ and $B$ have the desired property.
\end{proof}

\begin{figure}[ht!]
    \centering
    \includegraphics[width=\textwidth]{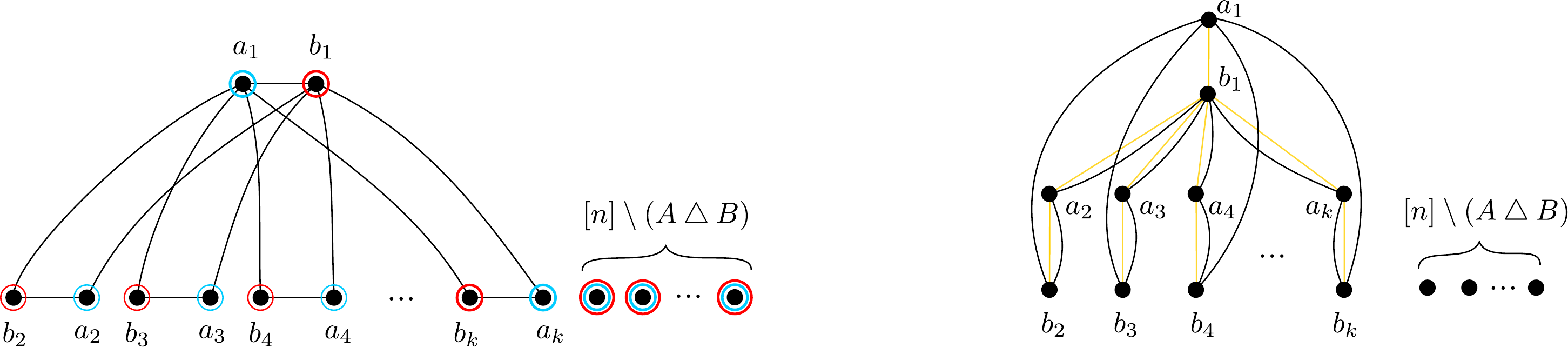}
    \caption{The left panel presents the construction of the graph $G$. The
    vertices of the only two maximum-size independent sets of $G$ are marked blue and red, respectively. The right panel presents an elimination forest of $G$ of height $4$.}
\label{fig:lower-bound-mis}
\end{figure}

We now use Lemma~\ref{pigeonhole} to prove the following statement, which will be used for establishing a lower bound against isolation schemes for maximum-size independent sets.

\begin{lemma}
    \label{mis-lowerbound}
    There exists $\alpha>0$ such that the following holds. 
    Suppose $n,W \in \nat$ are large enough and $\omega_1,\ldots,\omega_t \colon [n] \rightarrow [W]$ are weight
    assignments, where $t\leq \alpha\cdot \frac{n}{\log (nW)}$. Then there exists a graph $G$ on vertex set $[n]$ and such that
    \begin{itemize}[nosep]
     \item the treedepth of $G$ is at most $4$, and
     \item $G$ has exactly two different maximum-size independent sets $A$ and $B$, which moreover satisfy $\omega_i(A)=\omega_i(B)$ for all $i\in [t]$.
    \end{itemize}
\end{lemma}
\begin{proof}
    By Lemma~\ref{pigeonhole} applied for $\beta=\frac{1}{2}$, we can choose $\alpha$ small enough so that the assertion $t\leq \alpha\cdot \frac{n}{\log (nW)}$ implies the existence of two different sets $A,B \subseteq [n]$ such that 
    $|A\setminus B|=|B\setminus A|$, $A\cup B=[n]$, and
    $\omega_i(A)=\omega_i(B)$ for all $i\in [t]$.
    We are going to construct a graph $G$ on vertex sets $[n]$ of treedepth at most~$4$ such that $A$ and $B$ are the only two
    maximum-size independent sets in $G$. Let $k\coloneqq |A\setminus B|$.

    Let us arbitrarily enumerate $A\setminus B$ as $\{a_1,\ldots,a_k\}$ and $B\setminus A$ as $\{b_1,\ldots,b_k\}$. The edge set of $G$ consists of the following edges, see the left panel of Figure~\ref{fig:lower-bound-mis}:
    $$\{a_ib_i\colon i\in [k]\}\cup \{a_i b_1\colon i\in [k]\setminus \{1\}\} \cup \{a_1 b_i\colon i\in [k]\setminus \{1\}\}.$$
    This concludes the construction of $G$. Note that vertices of $A\cap B$ are isolated in $G$.
    It is easy to see that $G$ admits an elimination forest of height $4$, see the right panel of Figure~\ref{fig:lower-bound-mis}.
    
    It remains to prove that $A$ and $B$ are the only two maximum-size
    independent sets in $G$. Note first that $A$ and $B$ are indeed independent and there are no larger independent sets, because every independent set in $G$ contains at most one endpoint from each edge of the matching $M\coloneqq \{a_ib_i\colon i\in [k]\}$. Therefore, if $I$ is a maximum-size independent set in $G$, then $I$ needs to contain all the (isolated) vertices from $A\cap B$, and one endpoint of each edge of $M$. Now if $I$ contains the endpoint $a_1$ of the edge $a_1b_1$, then $I$ cannot contain any of the vertices $b_i$ for $i\in [k]$, because all these vertices are adjacent to $a_1$. Hence $I$ must contain all vertices $a_i$ for $i\in [k]$, implying that $I=A$. Analogously, if $I$ contains $b_1$, then $I=B$.
\end{proof}

\begin{corollary}[Lower bound for maximum-size independent sets]\label{cor:indset}
 There does not exist an isolation scheme for maximum-size independent sets on graphs of treedepth at most $4$ that would use $o(\log n)$ random bits and assign polynomially bounded weights.
\end{corollary}
\begin{proof}
 By Lemma~\ref{mis-lowerbound}, such an isolation scheme would need to produce $\Omega(n/\log n)$ different weight assignments on $n$-vertex graphs, for otherwise there would exist a graph of treedepth at most $4$ where the only two maximum-size independent sets receive the same weights in all possible weight assignments. It follows that the isolation scheme in question needs to use $\Omega(\log n)$ random bits. 
\end{proof}

We now present similar constructions for three other types of objects: minimum Steiner trees, minimum maximal matchings, and Hamiltonian cycles. In each case we first give a lemma presenting the construction, which is followed by a corollary stating the lower bound. In each case, the corollary follows from the same argument as that used in the proof of Corollary~\ref{cor:indset}.

\begin{figure}[ht!]
    \centering
    \includegraphics[width=\textwidth]{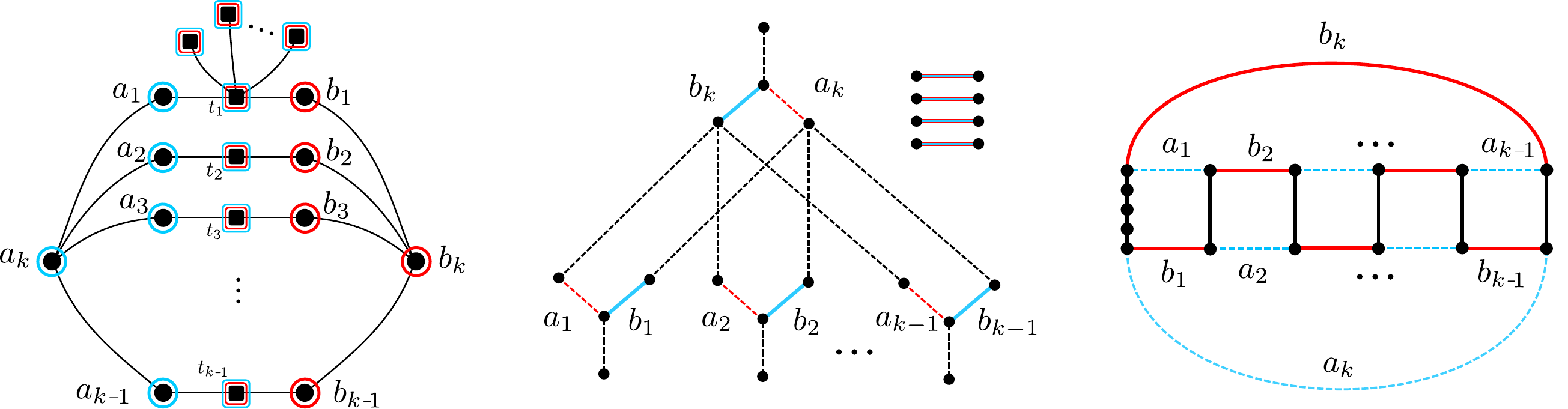}
    \caption{Construction used in the proofs of
    Lemmas~\ref{steiner-tree-lowerbound},~\ref{min-max-matching-lowerbound}, and~\ref{hamiltonian-cycle-lower-bound}, respectively.}
\label{fig:lower-bound2}
\end{figure}


For a graph $G$ and a set of terminals $T\subseteq V(G)$, a {\em{minimum Steiner tree}} is a minimum-size set of vertices $S\subseteq V(G)$ such that $T\subseteq S$ and $G[S]$ is connected.

\begin{lemma}
    \label{steiner-tree-lowerbound}
    There exists $\alpha>0$ such that the following holds. 
    Suppose $n,W \in \nat$ are large enough and $\omega_1,\ldots,\omega_t \colon [n] \rightarrow [W]$ are weight
    assignments, where $t\leq \alpha\cdot \frac{n}{\log (nW)}$. Then there exists a graph $G$ on vertex set $[n]$ and a set of terminals $T\subseteq [n]$ such that
    \begin{itemize}[nosep]
     \item the treedepth of $G$ is at most $4$, and
     \item there are exactly two different minimum Steiner trees for $G$ and $T$, say $A$ and $B$, and they satisfy $\omega_i(A)=\omega_i(B)$ for all $i\in [t]$.
    \end{itemize}
\end{lemma}
\begin{proof}
    As in the proof of Lemma~\ref{mis-lowerbound}, we can use Lemma~\ref{pigeonhole} for $\beta=\frac{1}{3}$ to make sure that provided $\alpha$ is chosen small enough, there are two different sets $A,
    B \subseteq [n]$ such that $A\cup B=[n]$, $|A\setminus B|=|B\setminus A|\leq n/3$, and $\omega_i(A)=\omega_i(B)$ for all $i\in [t]$. It is easy to modify the proof of Lemma~\ref{pigeonhole} so that the following property is also guaranteed: if $k\coloneqq |A \setminus B|$, then $k\geq 2$.
    We are going to construct a graph $G$ on vertex set $[n]$ together with a set of terminals $T\subseteq V(G)$ so that $A$ and $B$ are the only two
    minimum Steiner trees for $G$ and $T$, and $G$ has treedepth at most $4$.

    Let us arbitrarily enumerate $A\setminus B$ as $\{a_1,\ldots,a_k\}$ and $B\setminus A$ as $\{b_1,\ldots,b_k\}$. Note that $k\leq n/3$ and $|A\cap B|\geq n/3$, hence $k\leq |A\cap B|$. We set $T\coloneqq A\cap B$ to be the terminals. Further, let $T' \subseteq T$ be any subset
    of $k-1$ terminals and let us arbitrarily enumerate $T'$ as $\{t_1,\ldots,t_{k-1} \}$. First, we make every terminal in $T\setminus T'$ adjacent to the terminal $t_1$. Next for every $i \in
    [k-1]$, we add edges $t_ia_i$, $t_ib_i$, $a_ia_k$ and $b_ib_k$. This concludes the construction of $G$. See
    the left panel of Figure~\ref{fig:lower-bound2} for a visualization.

    Observe that each connected component of the graph $G-\{a_k,b_k\}$ is a star, and hence has treedepth at most $2$. It follows that $G$ has treedepth at most $4$. That $A$ and $B$ are the only two minimum Steiner trees for $G$ and $T$ is straightforward; we leave the verification to the reader.
\end{proof}

\begin{corollary}[Lower bound for minimum Steiner trees]\label{cor:st-tree}
 There does not exist an isolation scheme for minimum Steiner trees on graphs of treedepth at most $4$ that would use $o(\log n)$ random bits and assign polynomially bounded weights.
\end{corollary}

Next, recall that a matching in a graph is {\em{maximal}} if no strict superset of it is a matching, and it is moreover a {\em{minimum maximal matching}} if it has the smallest possible size among maximal matchings.

\begin{lemma}
    \label{min-max-matching-lowerbound}
    There exists $\alpha>0$ such that the following holds. 
    Suppose $m,W \in \nat$ are large enough and $\omega_1,\ldots,\omega_t \colon [m] \rightarrow [W]$ are weight
    assignments, where $t\leq \alpha\cdot \frac{m}{\log (mW)}$. Then there exists a graph $G$ with edge set $[m]$ such that
    \begin{itemize}[nosep]
     \item the treedepth of $G$ is at most $4$, and
     \item there are exactly two different minimum maximal matchings in $G$, say $A$ and $B$, and they satisfy $\omega_i(A)=\omega_i(B)$ for all $i\in [t]$.
    \end{itemize}
\end{lemma}
\begin{proof}
    By applying Lemma~\ref{pigeonhole} for $\beta=\frac{1}{5}$, we can choose $\alpha$ small enough so that the assertion $t\leq \alpha\cdot \frac{m}{\log (mW)}$ implies that there exists a pair of different subsets $A,B\subseteq [m]$ such that $|A\setminus B|=|B\setminus A|\leq m/5$, $A\cup B=[m]$, and $\omega_i(A)=\omega_i(B)$ for all $i\in [t]$. Note that then $|A\cap B|\geq 3m/5$, hence $|A\setminus B|\leq 3|A\cap B|$. Let then $\ol{K} \subseteq A \cap B$ be any subset of $A \cap B$ of size exactly $3k-2$, where $k\coloneqq |A\setminus B|$, and let $K \coloneqq (A \cap B) \setminus \ol{K}$. As in Lemma~\ref{steiner-tree-lowerbound}, we may assume that $k\geq 2$.
    
    Let $\wt{A} \coloneqq (A \setminus B) \cup K$ and
    $\wt{B} \coloneqq (B \setminus A) \cup K$. Note that $\wt{A} \cap \wt{B} =
    K$. Further, for every $i\in [t]$ we have
    $$\omega_i(\wt{A})=\omega_i(A)-\omega_i(\ol{K})=\omega_i(B)-\omega_i(\ol{K})=\omega_i(\wt{B}).$$
    We now construct a graph $G$ with edge set $[m]$ and treedepth at most $4$
    which has exactly two different minimum maximal matching: $\wt{A}$ and
    $\wt{B}$. See the center panel of Figure~\ref{fig:lower-bound2} for the construction.

    Let us arbitrarily enumerate $A\setminus B$ as $\{a_1,\ldots,a_k\}$ and $B \setminus A$ as $\{b_1,\ldots,b_k\}$.
    Further, recalling that $|\ol{K}|=3k-2$, we arbitrarily enumerate $\ol{K}$ as $\{d_1,\ldots,d_k,
    c_1,\ldots,c_{k-1},c_1',\ldots,c_{k-1}'\}$. Let us create $4k$ vertices:
    $\{v^c_1,\ldots,v^c_k\}$,$\{v^a_1,\ldots,v^a_k\}$,$\{v^b_1,\ldots,v^b_k\}$,
    $\{v^d_1,\ldots,v^d_k\}$.
    For every $i \in [k]$ we connect:
    \begin{itemize}[nosep]
     \item vertices $v^c_i$ and $v^a_i$ using edge $a_i$;
     \item vertices $v^c_i$ and $v^b_i$ using edge $b_i$; and
     \item vertices $v^c_i$ and $v^d_i$ using edge $d_i$.
    \end{itemize}
    Next, for every $i \in [k-1]$ we connect:
    \begin{itemize}[nosep]
     \item vertices $v^b_{k}$ and $v^a_i$ using edge $c_i$; and
     \item vertices $v^a_{k}$ and $v^b_i$ using edge $c'_i$.
    \end{itemize}
    In this way we have defined the endpoints of all the edges of $[m]$ apart from the edges of $K$. To finish the construction, for each edge $e\in K$ we add two extra vertices and connect them using $e$. Thus, $K$ becomes a matching in $G$ consisting of isolated edges.
    
    This concludes the construction of $G$. Note that if from $G$ we remove $v_k^a$ and $v_k^b$, then each of the remaining connected components is a star, and hence has treedepth at most $2$. This proves that $G$ has treedepth at most $4$. 
    
    We are left with proving that $\wt{A}$ and $\wt{B}$ are the only two minimum maximal matchings in $G$. Clearly $\wt{A}$ and $\wt{B}$ are maximal matchings. Let $M$ be any maximal matching of $G$. Since for every $i \in [k]$, vertex $v_i^d$ has degree $1$ and $v_i^c$ is its only neighbor, it follows that each vertex $v_i^c$ needs to be incident to an edge of $M$. Clearly, $M$ also needs to contain each edge of $K$, as these edges are isolated in $G$. Since vertices $v_i^c$ are pairwise nonadjacent, we conclude that every maximal matching of $G$ has at least $k+|K|$ edges. 
    As $|\wt{A}|=|\wt{B}|=k+|K|$, this implies that both $\wt{A}$ and $\wt{B}$ are minimum maximal matchings in $G$.
    
    It remains to argue that there are no minimum maximal matchings other than $\wt{A}$ and $\wt{B}$. Let $M$ be any minimum maximal matching in $G$. As we argued, $M$ needs to contain $K$, one edge incident to $v_i^c$ for each $i\in [k]$, and no other edges. In particular, no edge $c_i$ or $c_i'$, for any $i\in [k-1]$, is contained in $M$. By the maximality of $M$ this means that for each $i\in [k-1]$, at least one of the edges $\{a_k,b_i\}$ is included in $M$, and at least one of the edges $\{b_k,a_i\}$ is included in $M$. If neither $a_k$ nor $b_k$ was included in $M$, then this would mean that both $a_1$ and $b_1$ necessarily belong to $M$, a contradiction. If now $a_k\in M$ and $b_k\notin M$, then $a_i\in M$ for all $i\in [k-1]$ and $M=\wt{A}$. Similarly, if $a_k\notin M$ and $b_k\in M$, then $M=\wt{B}$.
\end{proof}

\begin{corollary}[Lower bound for minimum maximal matchings]\label{cor:mmm-lb}
 There does not exist an isolation scheme for minimum maximal matchings on graphs of treedepth at most $4$ that would use $o(\log n)$ random bits and assign polynomially bounded weights.
\end{corollary}

\begin{lemma}
    \label{hamiltonian-cycle-lower-bound}
    There exists $\alpha>0$ such that the following holds. 
    Suppose $m,W \in \nat$ are large enough and $\omega_1,\ldots,\omega_t \colon [m] \rightarrow [W]$ are weight
    assignments, where $t\leq \alpha\cdot \frac{m}{\log (mW)}$. Then there exists a graph $G$ with edge set $[m]$ such that
    \begin{itemize}[nosep]
     \item the pathwidth of $G$ is at most $4$, and
     \item there are exactly two different Hamiltonian cycles in $G$, say with edge sets $A$ and $B$, and they satisfy $\omega_i(A)=\omega_i(B)$ for all $i\in [t]$.
    \end{itemize}
\end{lemma}
\begin{proof}
    We apply Lemma~\ref{pigeonhole} again, this time for $\beta=\frac{1}{3}$. Hence, by selecting $\alpha$ small enough we can ensure that there exists different sets $A,B\subseteq [m]$ such that $|A\setminus B|=|B\setminus A|\leq m/3$, $A\cup B=[m]$, and $\omega_i(A)=\omega_i(B)$ for all $i\in [t]$. Note that then $|A\cap B|\geq m/3$, hence $|A\cap B|\geq k$, where $k\coloneqq |A\setminus B|$. As before, we may assume that $k\geq 3$.
    
    Let us arbitrarily enumerate $A\setminus B$ and $B\setminus A$ as $\{a_1,\ldots,a_k\}$ and $\{b_1,\ldots,b_k\}$, respectively. Further, let $\{c_1,\ldots,c_k\}$ be $k$ arbitrary elements of $A \cap B$. 
    
    We now construct a graph $G$ with edge set $[m]$ as follows; see the right panel of Figure~\ref{fig:lower-bound2}. First, create $2k$ vertices $\{u_1,\ldots,u_k\}$ and $\{v_1,\ldots,v_k\}$.
    For every $i \in [k]$ we connect vertices $u_i$ and $u_{i+1}$ using edge
    $a_i$, where $u_{k+1}=u_1$. Similarly, for every $i \in [k]$ we connect vertices $v_i$ with
    $v_{i+1}$ with edge $b_i$, where $v_{k+1}=v_1$. Next, for every $i \in [k]$ we connect vertices
    $u_i$ to $v_i$ with the edge $c_i$. 
    So far we have defined the endpoints of all the edges apart from the edges of $(A\cap B)\setminus \{c_1,\ldots,c_k\}$. To accommodate the remaining edges, replace the edge $c_1$ with a path of length $|A\cap B|-k+1$ connecting $u_1$ and $v_1$, and let the edges of this path be $(A\cap B)\setminus \{c_2,\ldots,c_k\}$.
    
    Note that removing $u_1$ and $v_1$ turns $G$ into the union of a path and a $2\times (k-1)$ grid, which is a graph of pathwidth at most $2$. It follows that $G$ itself has pathwidth at most $4$. Finally, it is easy to see that $G$ has exactly two Hamiltonian cycles, with edge sets $A$ and $B$, respectively.
\end{proof}

\begin{corollary}[Lower bound for Hamiltonian cycles]\label{cor:hc-lb}
 There does not exist an isolation scheme for Hamiltonian cycles on graphs of pathwidth at most $4$ that would use $o(\log n)$ random bits and assign polynomially bounded weights.
\end{corollary}


\subsection{Conditional Lower Bounds}

An important result in the complexity theory is the randomized reduction from
languages in $\mathsf{NP}$ to \textsc{Unique SAT} due to Valiant and Vazirani~\cite{usat}. This reduction
is the essential procedure in plenty fundamental results in the Computational
Complexity, most notably Toda's Theorem~\cite{toda}.

We would like to note, that the question of derandomization of this red was already subject to a rigorous
research. For example, \citet{dell13} showed that probably we cannot hope to improve the
probability of success of Valiant and Vazirani~\cite{usat} reduction to be $2/3$
unless $\mathsf{NP} \subseteq \mathsf{P/poly}$.
In the opposite success probability regime, Calabro et al.~\citet{calabro08}
gave a randomized polynomial time reduction from $k$-\textsc{SAT} to
\textsc{Unique $k$-SAT} which works with probability $2^{-\Oh(n \log^2{k}/k)}$. Their
bound was subsequently improved by~\cite{traxler08} to $2^{-\Oh(n \log{k}/k)}$
Very recently~\citet{vyas20} showed that if \textsc{Unique $k$-SAT} on $n$ variable
admits an $2^{n(1-f(k)/k)}$ time algorithm for some unbounded $f$, then \textsc{$k$-SAT}
is in $2^{n(1-f(k)(1-\eps)/k)}$ time for every $\eps > 0$. These reductions work
in exponential time.

The question of whether you can derandomize the
reduction of Valiant and Vazirani~\cite{usat} was already subject to a rigorous
research. Dell et al.~\cite{dell13} showed that improving the
success probability of Valiant and Vazirani~\cite{usat} reduction to $2/3$ is
not possible unless $\mathsf{NP} \subseteq \mathsf{P/poly}$. 
In the opposite success probability regime, Calabro et al.~\citet{calabro08}
gave a randomized polynomial time reduction from \textsc{$k$-SAT} to
\textsc{Unique $k$-SAT} which works with probability $2^{-\Oh(n \log^2{k}/k)}$. Their
bound was subsequently improved by~\cite{traxler08} to $2^{-\Oh(n \log{k}/k)}$
Very recently~\citet{vyas20} showed that if \textsc{Unique $k$-SAT} on $n$ variable
admits an $2^{n(1-f(k)/k)}$ time algorithm for some unbounded $f$, then \textsc{$k$-SAT}
is in $2^{n(1-f(k)(1-\eps)/k)}$ time for every $\eps > 0$. These reductions work
in exponential time.

Montoya and M{\"{u}}ller~\citet{parametereized-randomness} considered a parameterized version of
result of Valiant and Vazirani~\cite{usat} and showed that certain
parameterized problems are also as hard as their unique variants.

%
%
The lower bounds in this section build on the following working assumption that these types of randomized reductions cannot be derandomized in the following strong sense:
\begin{conjecture}[Linear-Random-Bits Conjecture]
    \label{lrbc}
    There is no randomized polynomial time reduction from \textsc{SAT} to \textsc{Unique SAT} that uses $o(n)$ random bits, where
    $n$ is the number of variables of the original instance.
\end{conjecture}

Falsifying Conjecture~\ref{lrbc} would on its own contribute to
great progress in computational complexity and hopefully inspire novel ideas
in derandomization.

The Conjecture~\ref{lrbc} enables us to express the natural
barrier in currently known techniques. We hope to spark an interest in more
randomness efficient reduction to \textsc{Unique SAT}, as it would hopefully lead to
better Isolation Schemes.  All our isolation schemes can be
even implemented in the restricted NC setting. We could even focus on weaker and more
believable Conjecture~\ref{lrbc} to exclude NC-reductions from \textsc{SAT} to \textsc{Unique SAT}
with $o(n)$ random bits. This weaker conjecture would enable us to exclude
isolation schemes that can be computed with NC circuits only.

Moreover, note that the reductions proposed in previous work blow up the size of
an instance to be $n^{\Oh(1)}$. We do not restrict a size of the output \textsc{Unique SAT}
instance, outside the fact that reduction needs to be in polynomial time.

In this section we prove the following theorem.
\begin{theorem}[Conditional Lower Bound for \textsc{Maximum Independent Set}]
    \label{thm:mis-lb2}
    No isolation scheme that can be computed in polynomial time and uses $o(td)$
    random bits with polynomially bounded maximum weight exists for \textsc{Maximum Independent Set} unless
    Conjecture~\ref{lrbc} is false.
\end{theorem}


To prove Theorem~\ref{thm:mis-lb2} we use the fact that \textsc{Maximum Independent Set}
is reducible to \textsc{SAT} with a \emph{parsimonious} reduction (cf.,~\cite[Exercise
2.30 and 2.28]{computational-complexity}). Parsimonious reduction is a reduction
that preservers number of solutions~\cite{computational-complexity}. In fact,
all known natural reductions between $\mathsf{NP}$-complete problems are parsimonious or
can be easily modified to be
parsimonious~\cite[Section 6.2.1]{computational-complexity}.

\begin{figure}[ht!]
    \centering
    \includegraphics[width=\textwidth]{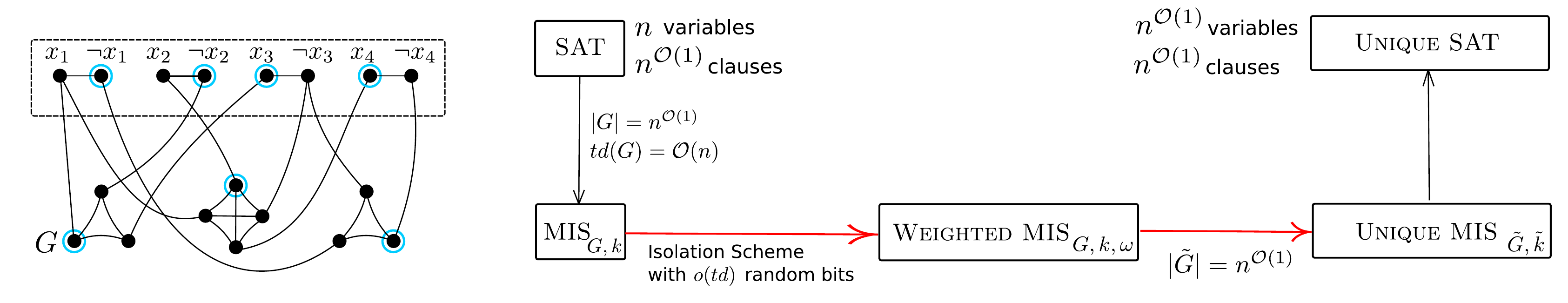}
    \caption{The left scheme presents the example construction of the graph $G$ in the
        reduction from \textsc{SAT} to \textsc{Maximum Independent Set}. Doted vertices are part
        of the variable clauses. Note that after removing them, the connected
        components are $\Oh(n)$ therefore treedepth of $G$ is $\Oh(n)$. The
    right scheme presents the schematic view of proof of
    Theorem~\ref{thm:mis-lb2}. We use the fact that reductions behind black arrows
    are known. In the proof we present red reductions.}
\label{fig:lower-bound-mis3}
\end{figure}

\begin{proof}
    First, we reduce \textsc{SAT} on $n$ variables to \textsc{Maximum Independent Set} problem on
    $N = n^{\Oh(1)}$ vertices by a standard reduction~\cite{Karp72}. The reduction
    gives us graph $G$ and number $k \in \nat$, such that there is no
    independent set of size greater than $k$ on $G$ and $G$ has maximum
    independent set of size exactly $k$ if the original formula was satisfiable.
    Moreover graph $G$ has treedepth $\Oh(n)$ because after removing vertices
    responsible for variable vertices graph $G$ has connected components of size
    $\Oh(n)$ and there are only $\Oh(n)$ variable vertices. See
    Figure~\ref{fig:lower-bound-mis3} for example construction of graph $G$ in
    reduction. Note that the reduction~\cite{Karp72} is not 
    parsimonious and we do not need it to be at this point.

    Now, we assume, that there exists an isolation scheme for \textsc{Maximum Independent Set}
    on $N$-vertex graphs with $d$ bounded treedepth that use $o(d)$ random bits and maximum
    weight is $M = n^{\Oh(1)}$. 

    Now, we invoke our isolation scheme on $N$-vertices graph $G$ and select a
    random integer $W \in_R [N M]$. We get a graph $G$ with weights
    $\omega_1,\ldots,\omega_{N} \in [n^{\Oh(1)}]$ of vertices and the property
    that with $1/n^{\Oh(1)}$ probability there exists exactly one independent
    set in $G$ of size exactly $k$ and weight exactly $W$ (if the original
    formula was satisfiable). 

    We modify our graph as follows. For every vertex $v_i \in V(G)$ the graph
    $\wt{G}$ have $\ell_i := 2 N M + w_i$ vertices $\{v^i_1,\ldots,v^i_{\ell_i}\}$.
    For each edge $(v_i,v_j) \in V(G)$ we add an edge $(v^i_a,v^j_b) \in V(\wt{G})$
    for every $a \in [\ell_i]$ and $b \in [\ell_j]$. Set $\wt{k} = 2 k N M + W$. The
    transformed graph $\wt{G}$ has the property:

    \begin{enumerate}[align=left, font=\normalfont, label=(\roman*)]
        \item If every independent set of $G$ is $<k$, then
            every independent set of $\wt{G}$ is $<\wt{k}$, and
        \item If $G$ has independent set of size $k$, then with $1/n^{\Oh(1)}$
            probability $\wt{G}$ has a \emph{unique} independent set of size
            $\wt{k}$.
    \end{enumerate}

    Graph $\wt{G}$ is an instance of unweighted \textsc{Maximum Independent Set} with the property that
    it has a unique independent set of size $\wt{k}$ if the original formula was
    satisfiable. Finally, we reduce our maximum independent set instance $\wt{G}$
    to the \textsc{SAT} with parsimonious reduction~\cite{computational-complexity}. This
    guarantees that with inversely polynomial probability the final instance of
    \textsc{SAT} has unique solution (if the original formula was a yes-instance).

    Because graph $G$ has treedepth $\Oh(n)$, our isolation scheme uses
    $o(n)$ random bits and therefore Conjecture~\ref{lrbc} is false. See
    Figure~\ref{fig:lower-bound-mis3} for a overall scheme of the reduction to
    \textsc{Unique SAT} and sizes and parameters of the produced instances. Note, that
    the final instance of \textsc{Unique SAT} may be polynomially larger than the
    original instance and it does not contradict Conjecture~\ref{lrbc}.
\end{proof}

Observe that  this lower-bounds framework works for many other $\mathsf{NP}$-complete
problems. Basically, all we need is (1) the reduction from \textsc{SAT} creates a graph
with $\Oh(n)$ treedepth/treewidth and (2) there is a polynomial time reduction
from weighted to unweighted problem. We can generalize our lower bound to work for
\textsc{Hamiltonian Cycle} in bounded treewidth graphs. First we reduce \textsc{SAT} to
\textsc{Hamiltonian Cycle} to the graph with $t = \Oh(n)$ and apply $o(t)$-random bits
isolation scheme.  We observe that we can subdivide
every edge of weight $\omega$ with a path of length $\omega$. This way we
construct an instance of \textsc{Subset TSP} on unweighted graphs (where terminals are vertices
of the original graphs) and use a parsimonious reduction from \textsc{Subset TSP} to
\textsc{SAT}.
We sum up this observation with the following remark.

\begin{remark}
    Assuming Conjecture~\ref{lrbc} no isolation scheme that can be computed in polynomial time and uses $o(t)$
    random bits with polynomially bounded maximum weight exists for \textsc{Hamiltonian
    Cycle} in graphs of treewidth bounded by $t$.
\end{remark}

For \textsc{Hamiltonian Cycle} in planar graphs, we can use the fact that
$\mathsf{NP}$-hardness
reduction~\cite{garey1976planar} from \textsc{SAT} on $m$-variables produces a graph $G$
with $\Oh(m^2)$ vertices. Similarly, we subdivide every edge of weight $\omega$
with a $\omega$-length path. We arrive at the instance of \textsc{Subset TSP} in
unweighted graphs and use a parsimonious reduction from \textsc{Subset TSP} to \textsc{SAT}.
Therefore any isolation scheme that needs $o(\sqrt{n})$ random bits and uses
polynomial weights would analogously contradict Conjecture~\ref{lrbc}.

\begin{remark}
    Assuming Conjecture~\ref{lrbc} no isolation scheme that can be computed in
    polynomial time and uses $o(\sqrt{n})$
    random bits with polynomially bounded maximum weight exists for \textsc{Hamiltonian
    Cycle} in planar graphs (where $n$ is the number of vertices of the graph).
\end{remark}

\section{Isolation of local vertex selection problems}
\label{sec:indepedent_set}

Recall that an independent set in a graph is a set of pairwise nonadjacent vertices, and an independent set is {\em{maximum}} if it has the largest possible cardinality.
In this section we prove Theorem~\ref{thm:mis}, which in plain words can be restated as follows. For a graph $G$, let
$\MIS(G)$ denote the set of all maximum independent sets in $G$. Suppose we
consider a graph $G\in \Gg_d$, say on vertex set $[n]$, given together with an
elimination forest $F$ of height at most $d$. The isolation scheme of
Theorem~\ref{thm:mis} is a family of $\ell=2^{\Oh(d)}$ weight functions $\omega_1,\ldots,\omega_\ell\colon [n]\times [d]\to [W]$, where $W=\Oh(n^6)$. Function $\omega_i$ assigns to each $v\in [n]$ the weight $\omega_i(v,\lvl_F(v))$. The requirement is that for all $G$ and $F$ as above, at least half of the functions $\omega_1,\ldots,\omega_\ell$ isolates $\MIS(G)$.

In the following, when discussing level-aware isolation schemes, we consider all weight functions as acting on a single argument: the vertex in question. The level of this vertex is supplied implicitly, as there is always some elimination forest of the graph fixed in the context.

In Section~\ref{sec:maximum-matching-treedepth} we show how to extend the
reasoning from this section to an example edge selection problem --- selection of maximum matchings.

\subsection{Exchange property}

In fact, we will prove a more general result than Theorem~\ref{thm:mis}. Namely, we isolate an abstract property of a vertex selection problem, which we call the {\em{exchange property}}, which is enjoyed by $\MIS(\cdot)$ and which is sufficient for our technique to work. There are multiple other problems of local nature that also have this property, hence the approach is indeed more general.

We will rely on the following two definitions.

\begin{definition}[Pivotal vertex]
    Let $G$ be a graph and $F$ be an elimination forest of $G$.
    For a family $\Ff \subseteq 2^{V(G)}$ and two different sets of vertices $A,B \in \Ff$, we say that
    a vertex $u$ is \emph{pivotal} for $A$ and $B$ in $F$ if $u \in A\triangle B$ and $\tail_F(u) \cap A = \tail_F(u) \cap B$.
\end{definition}

\begin{definition}[Exchange property]
    \label{exchange-prop}
    We say that a vertex selection problem $\Pp$ has the \emph{exchange property} if for every graph $G$, elimination forest $F$ of $G$, and
    weight function $\omega \colon V(G) \rightarrow \nat$ the following holds: if there exist two
    different $A,B \in \Pp(G)$ that are minimizers of $\omega$ on $\Pp(G)$, then there also exist $A',B' \in \Pp(G)$ that are minimizers of $\omega$ on $\Pp(G)$ such that there is only one pivotal vertex for $A'$ and~$B'$ in $F$.
\end{definition}

\begin{figure}[ht!]
    \centering
    \includegraphics[width=0.6\textwidth]{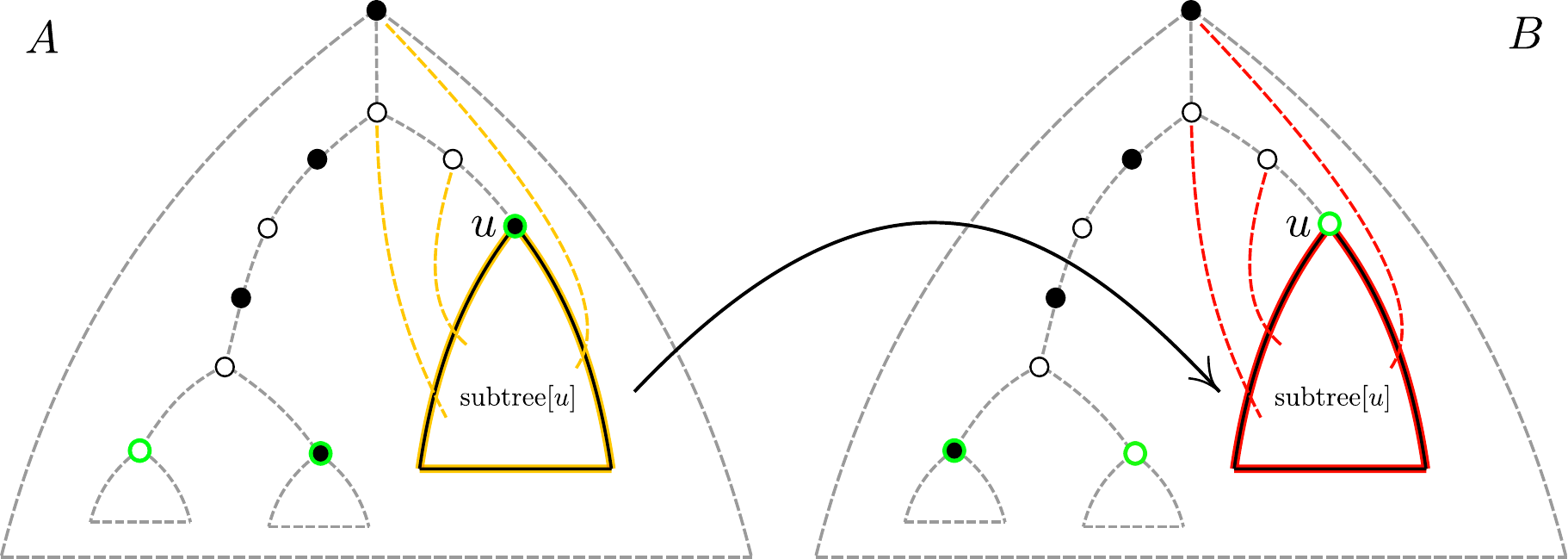}
    \caption{Schematic figure of the exchange argument used in the proof of Lemma~\ref{mis-has-exch-prop}. Both figures present the
    elimination forest $F$ of the given graph $G$. Black-filled vertices on the left are in
    $A$. Black-filled vertices on the right are in $B$. The pivotal vertices are
    green-stroked. Vertex $u \in A \setminus B$ is a chosen pivotal vertex. Because
    $A$ and $B$ are both minimizers of $\omega$, we can exchange the red
    subtree with the yellow one and construct a valid minimizer with one less pivotal vertex. }
\label{fig:definitions_layer}
\end{figure}


In the next sections we will focus on proving the following result.

\begin{theorem}
    \label{thm:exchange-iso}
    Let $\Pp$ be a vertex selection problem that enjoys the exchange property.
    Then for every $d \in \nat$, there is a level-aware isolation scheme for $\Pp$ on graphs of treedepth at most $d$ that uses $\Oh(d)$ random bits assigns weights bounded by $\Oh(n^6)$.
\end{theorem}

Therefore, Theorem~\ref{thm:mis} follows by combining Theorem~\ref{thm:exchange-iso} with the following result.

\begin{lemma}
    \label{mis-has-exch-prop}
    The vertex selection problem $\MIS(\cdot)$ has the exchange property.
\end{lemma}
\begin{proof}
    Let $G$ be a graph, $F$ be an elimination forest of $G$, and $\omega\colon V(G)\to \nat$ be a weight function.
    Assume that there exist
    two different $A, B \in \MIS(G)$ that are minimizers for $\omega$ and have
    more than one pivotal vertex. We will construct two different sets $A'$ and $B'$ that
    are also maximum independent sets in $G$ and have the same weight as
    $A$ and $B$, but have one less pivotal vertex. This is sufficient to prove the lemma, since we can repeat the
    construction to get a pair of minimizers with exactly one pivotal vertex, as required.

    Let $u \in V(G)$ be a pivotal vertex for $A$ and $B$. By symmetry, assume that $u\in B\setminus A$. We consider 
    $$A' \coloneqq A\qquad\textrm{and}\qquad B' \coloneqq (B \setminus \subtree[u]) \cup (\subtree[u] \cap A).$$
    It is easy to see that $A'$ and $B'$ have exactly one less pivotal vertex than $A$ and $B$: $u$ is pivotal for $A$ and $B$ but not for $A'$ and $B'$, and all other pivotal vertices are the same. Because $A'=A$, we only need to prove that $B' \in \MIS(G)$ and that $B'$ is a minimizer of $\omega$ on $\MIS(G)$. 

    First, note that $\tail(u) \cap A = \tail(u) \cap B$ because $u$ is pivotal of $A$ and $B$. Since the vertices of $\subtree[u]$ have only neighbors in $\tail(u)$, this implies that $B'$ is an independent set, and analogously the set $A''\coloneqq (A \setminus \subtree[u]) \cup (\subtree[u] \cap B)$ is an independent set as well. Since $|B'|+|A''|=|A|+|B|$, both $B'$ and $A''$ are independent sets, while $A$ and $B$ are maximum independent sets, it follows that both $B'$ and $A''$ must be maximum independent sets as well.
    Similarly, we have $\omega(B')+\omega(A'')=\omega(A)+\omega(B)$, so the assumption that $A$ and $B$ are both minimizers of $\omega$ on $\MIS(G)$ implies that $A''$ and $B'$ are also minimizers of $\omega$ on $\MIS(G)$.
\end{proof}

The same argument as the one used in the proof of Lemma~\ref{mis-has-exch-prop} can be also applied to other combinatorial objects in graphs, where validity of an object depends on checking the neighborhood of every vertex. For instance, it is easy to prove in this way that minimum vertex covers and minimum dominating sets also have the exchange property. Thus, Theorem~\ref{thm:exchange-iso} also applies to the corresponding vertex selection problems.

In the next sections we will work towards the proof of Theorem~\ref{thm:exchange-iso}. Therefore, let us fix a vertex selection problem $\Pp$ that has the exchange property.

\subsection{Warm-up: a deterministic isolation scheme}
\label{warmup}

Before commencing to the proof of Theorem~\ref{thm:mis}, we will show a
simple deterministic level-aware isolation scheme.
Observe that the Isolation Lemma in its most general form can be trivially
derandomized provided we allow the maximum weight to be $2^n$, where $n$ is the size of the universe. Namely, it is enough to select the
weight function $\omega(i) \coloneqq 2^i$. In general, allowing exponential weights is prohibitively expensive and not algorithmically useful, however considering this idea explains some intuition behind our techniques. As a warm-up, we now present a level-aware isolation scheme for the considered problem $\Pp$ on graphs of treedepth at most $d$ that is deterministic, but may use weights as large as $2^d$. The scheme is captured by the following lemma.

\newcommand{\detomega}{\omega_{\mathsf{det}}}
\newcommand{\rndomega}{\omega_{\mathsf{rnd}}}

\begin{lemma}[Exponential-weight deterministic isolation]
    \label{exp-mis}
    For every graph $G$ and an elimination forest $F$ of $G$, function $\detomega(v) \coloneqq 2^{\lvl_F(v)}$ isolates the family $\Pp(G)$.
\end{lemma}

Note that as Lemma~\ref{exp-mis} involves only one weight function $\detomega$, it provides a deterministic isolation scheme. Also, provided the height of $F$ is at most $d$, the assigned weights are upper bounded by $2^d$. Further, this is a level-aware isolation scheme, because $\detomega$
is also supplied with the level of the vertex in the given
elimination forest. In Section~\ref{sec:lbs} we showed that this additional
information is really necessary, since without it any isolation scheme for maximum independent sets needs to use $\Omega(\log(n))$ random bits, even on graphs of treedepth at most $4$.

\begin{proof}[Proof of Lemma~\ref{exp-mis}]
    Assume for contradiction that there exist two different sets in $\Pp(G)$ that
    are both minimizers of $\detomega(v) \coloneqq 2^{\lvl_F(v)}$ on $\Pp(G)$. By
    Lemma~\ref{mis-has-exch-prop} we know that there also exist $A,B\in \Pp(G)$ that
    are also minimizers of $\detomega$ and have exactly one pivotal vertex. 
    
    Let $u$ be the only pivotal vertex of $A$ and $B$ in $F$. Without loss of generality we
    may assume that $u \in A\setminus B$. Let 
    $$R \coloneqq A \setminus \subtree[u] = B \setminus \subtree[u],\qquad S_{A} \coloneqq A \cap \subtree(u),\qquad\textrm{and}\qquad S_{B} \coloneqq B \cap \subtree(u).$$

    We know that $\detomega(A) = \detomega(B)$. Therefore,
    $$\detomega(R) + \detomega(u) + \detomega(S_{A}) = \detomega(R) + \detomega(S_{B}),$$
    hence 
    \begin{equation}\label{eq:divisibility}
      \detomega(u) + \detomega(S_{A}) = \detomega(S_{B}).  
    \end{equation}
    Let $\ell$ be the level of
    vertex $u$. We know that $\detomega(u) = 2^\ell$. Moreover, the level of every
    vertex in $S_{A}$ and $S_{B}$ is greater than $\ell$. Therefore $\detomega(S_{A})$
    and $\detomega(S_{B})$ are divisible by $2^{\ell+1}$. We conclude that the right hand side of~\eqref{eq:divisibility} is divisible by $2^{\ell+1}$, while the left hand side is not. This is a contradiction, hence $\detomega$ must have a unique minimizer on $\Pp(G)$.
\end{proof}

\subsection{Warm-up continued: a randomized isolation scheme}

There is also a relatively simple level-aware isolation scheme for $\Pp$ on graphs of treedepth at most $d$ that uses $\Oh(d \log(n))$ random bits and assigns weights bounded by $\Oh(n)$. Namely, for every
level $i \in [d]$ chose uniformly at random a number $r_i \in [Cn]$ for some constant $C$ large enough, and let
$\rndomega(v) \coloneqq r_{\lvl_F(v)}$, where $F$ is the given elimination forest. Clearly, this scheme uses $\Oh(d \log{n})$ random bits.
It is not hard to prove that the function $\rndomega$ isolates
$\Pp(G)$ with high probability. We leave the details to the reader, as the isolation scheme presented in the next section will supersede this one.

The idea behind our proof of Theorem~\ref{thm:exchange-iso} is to combine the
deterministic isolation scheme $\detomega$, presented in Lemma~\ref{exp-mis}, the with randomized scheme sketched above. This approach is inspired by the shifting idea presented in~\cite{ccc17,isolation-lemma2}, but the technique is adapted to the setting of problems on graphs of bounded treedepth.

\subsection{Proof of Theorem~\ref{thm:exchange-iso}}
\label{proofofmis}

Fix the considered graph $G=(V,E)$ and elimination forest $F$ of height at most $d$, where $V=[n]$. We may assume that $d\geq 5 \ceil{\log{n}}$, for otherwise the deterministic isolation scheme from 
Lemma~\ref{exp-mis} isolates $\Pp(G)$ and assigns weights upper bounded by $2^d \le \Oh(n^6)$. 

\paragraph{Isolation Procedure.}
We start with the definition of the isolation scheme. Similarly to the scheme considered in Lemma~\ref{exp-mis}, the weight assigned to a vertex will only depend on its level in $F$ (and the random bits).

For every $i\in [d]$, we can uniquely encode $i$ as a pair
of integers $(e(i),f(i))$ so that 
$$i = \ceil{\log{n}} \cdot e(i) + f(i),\quad\textrm{where } 0\leq f <
\ceil{\log{n}}.$$ Note that then $e(i)\leq \kappa$, where $\kappa \coloneqq \floor{\frac{d}{\ceil{\log{n}}}}$. As $d \ge 5 \ceil{\log{n}}$, we have $\kappa>0$. 

Next, for every $e \in \{0,1\ldots,\kappa\}$ choose an integer $r_e \in [32n^5]$ independently and uniformly
at random. Then $\ol{r} \coloneqq (r_0,r_1,\ldots,r_\kappa)$ is the vector of random integers used by our weight function. Note that choosing $\ol{r}$ requires $(\kappa+1)\cdot \Oh(\log n)=\Oh(d)$ random bits, as promised. In the following, we treat $r_0,\ldots,r_\kappa$ as random variables, thus $\ol{r}$ is a random vector that is uniformly distributed in $\Omega\coloneqq [32n^5]^{\kappa+1}$.

Finally, for a vertex $v$ and $\ol{\rho}\in \Omega$, we define
\begin{displaymath}
    \omega_{\ol{\rho}}(v) \coloneqq \rho_{e(i)} \cdot 2^{f(i)},
    \text{ where } i=\lvl(v).
\end{displaymath}
With this definition in place, our isolation scheme simply samples $\ol{r}$ as above and outputs the weight function $\omega_{\ol{r}}$. We are left with arguing that $\omega_{\ol{r}}$ isolates $\Pp(G)$ with probability at least $\frac{1}{2}$.



\newcommand{\excl}[2]{\ol{#1}_{-#2}}
\newcommand{\Excl}[2]{#1_{-#2}}


\paragraph{Analysis.} 
For a vertex $v$, we write $e(v)\coloneqq e(i)$ and $f(v)\coloneqq f(i)$, where $i$ is the level of $v$ in $F$.
For every set $X\subseteq V$, let us define a linear form $\phi^X\colon \Omega\to \nat$ as follows:
$$\phi^X(\ol{\rho})\coloneqq \sum_{v\in X} \rho_{e(v)}\cdot 2^{f(v)}.$$
Intuitively, we can think of $\varphi^X$ as a compressed version of $X$ from which we can compute $\omega_{\ol{\rho}}(X)$ once it is fixed.
Indeed, $\varphi^X$ can be thought of as a compressed version of $X$ since $|\Omega| \leq 2^{\Oh(d)}$.
In particular, note that we actually have
\begin{equation}\label{eq:bobr}\phi^X(\ol{\rho})=\omega_{\ol{\rho}}(X).
\end{equation}
Further, let 
$$\Phi\coloneqq \{\phi^X\colon X\in \Pp(G)\}.$$
The following lemma is the key step in the proof.

\begin{lemma}\label{lem:diff-forms}
 Suppose $\ol{\rho}\in \Omega$ is such that $\omega_{\ol{\rho}}$ does not isolate $\Pp(G)$. Then there are two different linear forms $\alpha,\beta\in \Phi$ such that $$\alpha(\ol{\rho})=\beta(\ol{\rho})=\min_{\gamma\in \Phi} \gamma(\ol{\rho}).$$
\end{lemma}
\begin{proof}
 Since $\omega_{\ol{\rho}}$ does not isolate $\Pp(G)$, there are two different minimizers of $\omega_{\ol{\rho}}$ on $\Pp(G)$. By the exchange property, there are also two different minimizers $A$ and $B$ such that there exists exactly one pivotal vertex for $A$ and $B$, say $u$. Without loss of generality suppose that $u\in A\setminus B$. Since $A$ and $B$ are both minimizers, by~\eqref{eq:bobr} we have
 $$\phi^A(\ol{\rho})=\phi^B(\ol{\rho})=\min_{\gamma\in \Phi} \gamma(\ol{\rho}).$$
 Hence, it suffices to prove that $\phi^A\neq \phi^B$.
 
 Let $(e,f)\coloneqq (e(u),f(u))$. We claim that the coefficients of $\phi^A$ and $\phi^B$ standing by the variable $\rho_e$ are different. Letting $L\coloneqq \{v~|~e(v)=e\}$, we see that these coefficients are respectively equal to
 $$\sum_{v\in A\cap L} 2^{f(v)}\qquad \textrm{and}\qquad \sum_{v\in B\cap L} 2^{f(v)}.$$
 Suppose for contradiction that these coefficients are actually equal, that is,
 \begin{equation}\label{eq:wydra}
\sum_{v\in A\cap L} 2^{f(v)}=\sum_{v\in B\cap L} 2^{f(v)}.
 \end{equation}
 Recall that $u$ is the only pivotal vertex for $A$ and $B$, hence $A\setminus \subtree[u]=B\setminus \subtree[u]$. Therefore, from~\eqref{eq:wydra} we infer that
 \begin{equation}\label{eq:nutria}
\sum_{v\in A\cap L\cap \subtree[u]} 2^{f(v)}=\sum_{v\in B\cap L\cap \subtree[u]} 2^{f(v)}.  
 \end{equation}
 Now observe that $u\in A\cap L\cap \subtree[u]$, $u\notin B\cap L\cap \subtree[u]$, and $f(v)>f(u)$ for each $v\in L\cap \subtree(u)$. From this it follows that the right hand side of~\eqref{eq:nutria} is divisible by $2^{f(u)+1}$, while the left hand side is not. This is a contradiction.
\end{proof}

We now combine Lemma~\ref{lem:diff-forms} with the following result of Chari et al.~\cite{isolation-lemma2}.

\begin{lemma}[Proposition~2 of \cite{isolation-lemma2}]\label{lem:chari-forms}
 Let $\Phi$ be a set of linear forms over $t<N$ variables with coefficients belonging to $\{0,\ldots,N^2-1\}$. Let $\ol{r}$ be chosen from $[N^5]^t$ uniformly at random. Then the probability that there exist two different $\alpha,\beta\in \Phi$ such that $\alpha(\ol{r})=\beta(\ol{r})=\min_{\gamma\in \Phi}\gamma(\ol{r})$ is at most $\frac{1}{2}$.
\end{lemma}

Now observe that the coefficients in the forms appearing in $\Phi$ are bounded by $n\cdot 2^{\lceil \log n\rceil}<(2n)^2$, because each coefficient is a sum of at most $n$ summands of the form $2^{f(v)}$, and each of these is bounded by $2^{\lceil \log n\rceil}$. Since $\ol{r}$ is chosen uniformly at random from $[(2n)^5]^{\kappa+1}$, by combining Lemmas~\ref{lem:diff-forms} and~\ref{lem:chari-forms}, where the latter is applied for $N\coloneqq 2n$, we conclude that $\omega_{\ol{r}}$ isolates $\Pp(G)$ with probability at least $\frac{1}{2}$.

\newcommand{\MM}{\mathsf{MM}}

\section{Isolation of local edge selection problems}
\label{sec:maximum-matching-treedepth}

In the Section~\ref{sec:indepedent_set}, we designed level-aware isolation schemes
exclusively for vertex selection problems. In this
section we demonstrate that our techniques can be also applied to edge-selection problems on the example of maximum matchings: formally, we consider the edge selection problem $\MM(\cdot)$ that maps every graph $G$ to the family $\MM(G)$ consisting of all maximum-size matchings in $G$. 


Suppose we are given a graph $G=(V,E)$ and an elimination forest $F$ of $G$.
The {\em{level}} of an edge $e=uv \in E$ in $F$ is defined as
$$\lvl(e) \coloneqq \min \{\lvl(u),\lvl(v)\}.$$  
In our level-aware isolation scheme, the weight function will be supplied with two parameters: an edge and its level.

We now introduce the analogues of pivotal vertices and the exchange property.

\begin{definition}[Edge-pivotal vertex]
    For a family $\Ff \subseteq 2^E$ and two different sets of vertices $A,B \in \Ff$, we say that
    a vertex $u$ is \emph{edge-pivotal} for $A$ and $B$ if:
    \begin{itemize}
        \item for every vertex $x \in \tail(u)$ it holds that $ux \in A$ if and only if $ux \in B$;
        \item there exists $x \in \subtree(u)$, such that $ux \in A\triangle B$;
        \item no vertex $u'\in \tail(u)$ has the two properties above.
    \end{itemize}
\end{definition}

\begin{definition}[Exchange property]
    \label{exchange-prop}
    We say that an edge selection problem $\Pp$ has the {\em{exchange property}} if for every graph $G=(V,E)$, elimination forest $F$ of $G$, and weight function $\omega\colon E\to \nat$, if there exist two different minimizers of $\omega$ on $\Pp(G)$, then there also exist two different minimizers of $\omega$ on $\Pp(G)$ for which there is exactly one edge-pivotal vertex.
\end{definition}

Using a reasoning similar to that from the proof of Lemma~\ref{mis-has-exch-prop}, we get the following.

\begin{lemma}\label{mm-has-exch-prop}
 The edge selection problem $\MM(\cdot)$ has the exchange property.
\end{lemma}
\begin{proof}
 Let $G=(V,E)$ be a graph, $F$ be an elimination forest of $G$, and $\omega\colon E\to \nat$ be a weight function such that there are two different minimizers $A,B$ of $\omega$ on $\MM(G)$. Let $e$ be any edge of $A\triangle B$ with the minimum possible level. Say that $e=uv$, where $u$ is an ancestor of $v$. Let us define
 $$B'\coloneqq \left(B\cap \binom{\subtree[u]}{2}\right)\cup \left(A\setminus \binom{\subtree[u]}{2}\right).$$
 As in the proof of Lemma~\ref{mis-has-exch-prop}, using the assumptions that $A,B$ are maximum matching that are minimizers of $\omega$ on $\MM(G)$, and that $e$ is an edge of $A\triangle B$ of minimum possible level, we can easily see that $B'$ is also a maximum matching that is a minimizer of $\omega$ on $\MM(G)$. It now follows that $u$ is the only pivotal vertex for $A$ and $B'$.
\end{proof}

The remainder of this section is devoted to the 

\begin{theorem}
    \label{thm:mm-iso-edge}
 For every $d\in \nat$, there is a level-aware isolation scheme for maximum
 matchings on graphs of treedepth at most $d$ that uses $\Oh(d\log n)$ random bits assigns weights bounded by $n^{\Oh(1)}$.
\end{theorem}

Note that contrary to the situation in Section~\ref{sec:indepedent_set}, we will conduct our reasoning only for the edge selection problem $\MM(\cdot)$, as we will use some additional combinatorial properties of maximum matchings.
While strong Isolation Lemma's already exist for the case of maximum matchings~\cite{tarnawski17}, our approach uses less random bits and seems extendable to other edge selection problems. 

%
%
%
%


For the rest of this section, let us fix the graph $G=(V,E)$, and enumeration of its edges $\id\colon E\to [m]$, and an elimination forest $F$ of $G$.
Mirroring the structure from Section~\ref{sec:indepedent_set}, we first give a deterministic isolation scheme, which will be subsequently randomized in order to reduce the weights at the cost of random bits.

Recall that our weight functions take two parameters: an edge and its level. Similarly as in Section~\ref{sec:indepedent_set}, we think of weight functions as acting only on edges, while the level of an edge is inferred implicitly from the forest $F$.

\subsection{Deterministic isolation scheme}

Let us introduce the weight function
$$\detomega(e) \coloneqq \id(e) \cdot n^{2\lvl(e)}.$$
We observe the following.

\begin{lemma}[Exponential weight isolation ]
    \label{exp-matching}
    Function $\detomega$ isolates the family $\MM(G)$.
\end{lemma}
\begin{proof}
    Assume for the contrary, that there exists two different maximum matchings $A,B \in \MM(G)$ that are both minimizers of $\detomega$. Because $\MM(\cdot)$ has the edge-exchange property, we can assume that $A$ and $B$
    have exactly one edge-pivotal vertex. Let $u$ be such a vertex and let
    $\ell$ be its level in the elimination forest~$F$.

    Let $$R \coloneqq A \setminus \binom{\subtree[u]}{2}$$ be the set of edges
    from $A$ that have at least one endpoint outside of $\subtree[u]$. Note that
    since $u$ is the only edge-pivotal vertex for $A$ and $B$, it holds that $$R
    = B \setminus \binom{\subtree[u]}{2}.$$  For a vertex $v \in V$, let
    $E[v]$ be the set of edges with at least one endpoint in $v$. Finally, let
    $$S_A \coloneqq A \setminus (R \cup E[u])=A\cap \binom{\subtree(u)}{2}\qquad\textrm{and}\qquad S_B \coloneqq B \setminus (R \cup E[u])=B\cap \binom{\subtree(u)}{2}.$$

    We assumed that $\detomega(A) = \detomega(B)$. Therefore,
    $$\detomega(R) + \detomega(S_A) + \detomega(E[u] \cap A) = \detomega(R) + \detomega(S_B) + \detomega(E[u] \cap B),
    $$
    implying that
    \begin{equation}\label{eq:edge-cntr}
     \detomega(S_A) + \detomega(E[u] \cap A) = \detomega(S_B) + \detomega(E[u] \cap B)
    \end{equation}
    
    Every vertex in the maximum matching has
    degree $1$. Since $u$ is pivotal for $A$ and $B$, $u$ is incident in $A$ to an edge $e_A$ and incident in $B$ to a different edge $e_B$ such that $\lvl(e_A)=\lvl(e_B)=\ell$.
    Therefore, 
    $$\detomega(E[u] \cap A) = \id(e_A) n^{2\ell}\qquad\textrm{and}\qquad
    \detomega(E[u] \cap B) = \id(e_B) n^{2\ell}.$$ 
    Note that $\lvl(e)>\ell$ for each edge $e\in S_A\cup S_B$, hence both $\detomega(S_A)$ and $\detomega(S_B)$ is divisible by $n^{2\ell+2}$. Since $\id(e_A)\neq \id(e_B)$, we conclude that the two sides of~\eqref{eq:edge-cntr} give different remainders modulo $n^{2\ell+2}$, a contradiction.
\end{proof}

\subsection{Randomized isolation scheme}

%
%

We now proceed to the proof of Theorem~\ref{thm:mm-iso-edge}.
We begin with the construction of the weight function. First, for every $i\in [d]$ we choose a number $r_i\in [n^{10}]$ independently and uniformly at random. Thus $\ol{r}\coloneqq (r_1,\ldots,r_d)$ is a random vector, distributed uniformly in $\Omega\coloneqq [n^{10}]^d$. For $\ol{\rho}\in \Omega$, we define the weight function $\omega_{\ol{\rho}}$ as follows:
$$\omega_{\ol{\rho}}(e)\coloneqq \id(e)\cdot \rho_{\lvl(e)}.$$
Our isolation scheme simply samples $\ol{r}$ as above and returns the weight function $\omega_{\ol{r}}$. Note that $\omega_{\ol{r}}$ assigns weights upper bounded by $\Oh(n^{10}\id(e))=\Oh(n^{12})$ and uses $\Oh(d\log n)$ random bits, as promised, so it remains to prove that $\omega_{\ol{r}}$ isolates $\MM(G)$ with probability at least $\frac{1}{2}$. 

\paragraph*{Analysis.} The argument is similar to that used in the proof of Theorem~\ref{thm:exchange-iso}. For each $X\subseteq E$, we define a linear form $\phi^X\colon \Omega\to \nat$ as
$$\phi^X(\ol{\rho})\coloneqq \sum_{e\in X} \id(e)\cdot \rho_{\lvl(e)}.$$
Thus, we have
\begin{equation}\label{eq:bobr2}
\phi^X(\ol{\rho})=\omega_{\ol{\rho}}(X). 
\end{equation}
Let
$$\Phi\coloneqq \{\phi^X\colon X\in \MM(G)\}.$$
Again, the key step is captured by the following lemma.

\begin{lemma}\label{lem:MM-diff-forms}
  Suppose $\ol{\rho}\in \Omega$ is such that $\omega_{\ol{\rho}}$ does not isolate $\Pp(G)$. Then there are two different linear forms $\alpha,\beta\in \Phi$ such that $$\alpha(\ol{\rho})=\beta(\ol{\rho})=\min_{\gamma\in \Phi} \gamma(\ol{\rho}).$$
\end{lemma}
\begin{proof}
 Since $\omega_{\ol{\rho}}$ does not isolate $\MM(G)$, there are two different minimizers of $\omega_{\ol{\rho}}$ on $\MM(G)$. By the exchange property, there are also two different minimizers $A$ and $B$ such that there exists exactly one pivotal vertex for $A$ and $B$, say $u$. Since $A$ and $B$ are both minimizers, by~\eqref{eq:bobr2} we have
 $$\phi^A(\ol{\rho})=\phi^B(\ol{\rho})=\min_{\gamma\in \Phi} \gamma(\ol{\rho}).$$
 Hence, it suffices to prove that $\phi^A\neq \phi^B$.
 
 Let $i\coloneqq \lvl(u)$. We claim that the coefficients of $\phi^A$ and $\phi^B$ standing by the variable $\rho_i$ are different. Letting $L\coloneqq \{e~|~\lvl(e)=i\}$, we see that these coefficients are respectively equal to
 $$\sum_{e\in A\cap L} \id(e)\qquad \textrm{and}\qquad \sum_{e\in B\cap L} \id(e).$$
 Now recall that $u$ is the only pivotal vertex of $A$ and $B$. Hence, as both $A$ and $B$ are matchings, we observe that $A\cap L$ and $B\cap L$ differ only in the edge that is incident to $u$ (or lack thereof). Since the two edges incident to $u$ in $A\cap L$ and $B\cap L$ have different identifiers (or one is non-existent), it follows that $\sum_{e\in A\cap L} \id(e) \neq \sum_{e\in B\cap L} \id(e)$. This concludes the proof.
\end{proof}

Now observe that linear forms from $\Phi$ have coefficients upper bounded by $m\cdot n<n^4$. Hence, we can again combine Lemma~\ref{lem:MM-diff-forms} with Lemma~\ref{lem:chari-forms} (applied for $N=n^2$) to infer that $\omega_{\ol{r}}$ isolates $\MM(G)$ with probability at least $\frac{1}{2}$.

\section{Directions for further research}
\label{sec:conc}
In this paper we presented several isolation schemes for $\mathsf{NP}$-complete problems, and we showed that analogues of
decomposition-based methods such as Divide\&Conquer can also be used to design more randomness-efficient isolation schemes.
While we provide nearly matching lower bounds for all our results, at least as far as the number of random bits is concerned, we still leave open a number of interesting open questions:

\begin{enumerate}

\item Can we improve our isolation schemes to have weights that are only
    polynomial in $n$, while not increasing the number of used random bits? Note
    that in our approach, the use of large weights is crucial for the
    application of Lemma~\ref{FKS} that deals with interactions between
    different partial solutions in our isolation
    schemes.\footnote{In~\cite{isolation-lemma2} a similar lemma was used to
    obtain isolation schemes with polynomial weights, but since the objects of
    the set family are not decomposed, the authors did not have this issue of interactions
    between different partial solutions.} 
    

\item Can we shave off the log factors in the number of used random bits in our
    results?  While some of the $\log n$ factors seem to be inherent in our
    ideas, there still might be a little room. For example, Melkebeek and Prakriya~\cite{ccc17} presented an isolation scheme for reachability that uses
    $\Oh(\log^{1.5}(n))$-random bits. Perhaps with their ideas one can get the
    same guarantees for isolating Hamiltonian cycles in constant treewidth graphs.
    
\item Does the (even more) natural isolation scheme work as well? Many of our
    isolation schemes draw several random prime numbers and assign a weight that
    is obtained by concatenating the congruence class of the vertex/edge
    identifier with respect to the different primes. A more natural, but
    possibly harder to analyse, scheme would be to sample a single (larger)
    prime number and define the weights to be the congruence classes of the
    identifiers with respect to that single prime.

\item Our methods allowed us to derandomize polynomial-space algorithms for
    $H$-minor free graphs without significantly increase the running time. Can our
    methods be used to derandomize other algorithms likewise?

\end{enumerate}



\bibliographystyle{abbrv}
\bibliography{bib}



\end{document}